%% file: ec2018.tex
\documentclass[11pt,letterpaper]{article}

\newcount\Colorflag
\usepackage[margin=1in]{geometry}

\usepackage{hyperref,enumitem}
\setlist{noitemsep,leftmargin=\parindent,topsep=2pt}
\usepackage[numbers]{natbib}
\usepackage{xspace}
\usepackage{graphicx}
\usepackage{amsmath,amssymb,amsthm}
\usepackage{booktabs} %
\usepackage{color}
\usepackage{nicefrac}
\usepackage[font=small]{caption}
\usepackage{subcaption}
\usepackage{multirow}
\usepackage{algorithm}
\usepackage[noend]{algpseudocode}
\usepackage{bbm}
\usepackage[toc,page]{appendix}
\usepackage{wrapfig}
\makeatletter
\newcommand*{\inlineequation}[2][]{%
  \begingroup
    \refstepcounter{equation}%
    \ifx\\#1\\%
    \else
      \label{#1}%
    \fi
    \relpenalty=10000 %
    \binoppenalty=10000 %
    \ensuremath{%
      #2%
    }%
    ~\@eqnnum
  \endgroup
}
\makeatother

\usepackage{enumitem}

\usepackage[suppress]{color-edits}
\addauthor{vs}{blue}
\addauthor{cp}{red}
\addauthor{zf}{green}

\newcount\Comments 
\newcommand{\kibitz}[2]{\ifnum\Comments=1{\color{#1}{#2}}\fi}

\newcommand{\E}{\mathbb{E}}

\DeclareMathOperator*{\argmax}{arg\,max}
\DeclareMathOperator*{\argsup}{arg\,sup}

\newcommand{\1}{\mathbbm{1}}

\newcommand{\eps}{\epsilon}

\newcommand{\B}{\mathcal{B}}

\newcommand{\mbf}{\mathbf}

\newcommand{\calN}{\mathcal{N}}
\newcommand{\winexp}{\textsc{WIN-EXP}}
\newcommand{\winexpG}{\textsc{WIN-EXP-G}}
\renewcommand{\Pr}{\ensuremath{\mathrm{Pr}}}
\newcommand{\opt}{\ensuremath{\textsc{OPT}}}

\theoremstyle{plain}
\newtheorem{theorem}{Theorem}[section]
\newtheorem{corollary}{Corollary}[theorem]
\newtheorem{lemma}[theorem]{Lemma}
\newtheorem{example}[theorem]{Example}
\newtheorem{comment}[theorem]{Comment}
\newtheorem{definition}[theorem]{Definition}

\begin{document}

\title{Learning to Bid Without Knowing your Value}

\author{
	Zhe Feng\\
	Harvard University\\
	\texttt{zhe\_feng@g.harvard.edu}\\
	\and
	Chara Podimata\\
	Harvard University\\
	\texttt{podimata@g.harvard.edu}\\
	\and  
	Vasilis Syrgkanis\\
	Microsoft Research\\
	\texttt{vasy@microsoft.com}}

\date{\today}
\maketitle

\begin{abstract}
We address online learning in complex auction settings, such as sponsored search auctions, where the value of the bidder is unknown to her, evolving in an arbitrary manner and observed only if the bidder wins an allocation. We leverage the structure of the utility of the bidder and the partial feedback that bidders typically receive in auctions, in order to provide algorithms with regret rates against the best fixed bid in hindsight, that are \emph{exponentially faster} in convergence in terms of dependence on the action space, than what would have been derived by applying a generic bandit algorithm and almost equivalent to what would have been achieved in the full information setting. Our results are enabled by analyzing a new online learning setting with outcome-based feedback, which generalizes learning with feedback graphs. We provide an online learning algorithm for this setting, of independent interest, with regret that grows only logarithmically with the number of actions and linearly only in the number of potential outcomes (the latter being very small in most auction settings). Last but not least, we show that our algorithm outperforms the bandit approach experimentally\footnote{Our code is publicly available on \href{https://github.com/zfengharvard/bandit-sponsored-search}{github}.} and that this performance is robust to dropping some of our theoretical assumptions or introducing noise in the feedback that the bidder receives.
\end{abstract}

\maketitle

\section{Introduction}
\input{introduction}

\section{Learning in Auctions without Knowing your Value}\label{sec:binary}\label{SEC:BINARY}
\input{preliminaries}

\section{Abstraction: Learning with Win-Only Feedback}\label{sec:win-only}\label{SEC:WIN-ONLY}
\input{win-only-learning}

\section{Beyond Binary Outcomes: Outcome-Based Feedback}\label{sec:outcome-based}\label{SEC:OUTCOME-BASED}
\input{outcome-based-learning}

\subsection{Batch Rewards Per-Iteration and Sponsored Search Auctions}\label{sec:batch}\label{SEC:BATCH}
\input{batch-rewards}

\section{Continuous Actions with Piecewise-Lipschitz Rewards}\label{sec:continuous}\label{SEC:CONTINUOUS}
\input{continuous-lipschitz}

\section{Further Extensions}
In this section, we discuss an extension to switching regret and the implications on Price of Anarchy and one to the feedback graphs setting.
\subsection{Switching Regret and Implications for Price of Anarchy}\label{sec:switch-poa}\label{SEC:SWITCH-POA}
\input{switching-poa}

\subsection{Feedback Graphs over Outcomes}\label{sec:graph}\label{SEC:GRAPH}
\input{outcome-feedback-graph}
\section{Experimental Results}\label{sec:experiments}
\input{experiments}

\section{Conclusion}\label{sec:disc}
\input{discussion}

\clearpage
\bibliographystyle{ACM-Reference-Format}
\bibliography{agt}

\input{appendix}
\end{document}

%% file: introduction.tex
A standard assumption in the majority of the literature on auction theory and mechanism design is that participants that arrive in the market have a clear assessment of their valuation for the goods at sale. This assumption might seem acceptable in small markets with infrequent auction occurrences and amplitude of time for participants to do market research on the goods. However, it is an assumption that is severely violated in the context of the digital economy. 

In settings like online advertisement auctions or eBay auctions, bidders participate very frequently in auctions that they have very little knowledge about the good at sale, e.g. the value produced by a user clicking on an ad. It is unreasonable, therefore, to believe that the participant has a clear picture of this value. However, the inability to pre-assess the value of the good before arriving to the market is alleviated by the fact that due to the large volume of auctions in the digital economy, participants can employ \emph{learning-by-doing} approaches.

In this paper we address exactly the question of \emph{how would you learn to bid approximately optimally in a repeated auction setting where you do not know your value for the good at sale and where that value could potentially be changing over time}. The setting of learning in auctions with an unknown value poses an interesting interplay between exploration and exploitation that is not standard in the online learning literature: in order for the bidder to get feedback on her value she has to bid high enough to win the good with higher probability %
and hence, receive some information about that underlying value. However, the latter requires paying a higher price. Thus, there is an inherent trade-off between value-learning and cost. The main point of this paper is to address the problem of learning how to bid in such unknown valuation settings with partial \emph{win-only feedback}, so as to minimize the regret with respect to the best fixed bid in hindsight.

On one extreme, one can treat the problem as a Multi-Armed Bandit (MAB) problem, where each possible bid that the bidder could submit (e.g. any multiple of a cent between $0$ and some upper bound on her value) is treated as an arm. Then, standard MAB algorithms (see e.g. \citep{BCB2012}) can achieve regret rates that scale \emph{linearly} with the number of such discrete bids. The latter can be very slow and does not leverage the structure of utilities and the form of partial feedback that arises in online auction markets. Recently, the authors in \citep{WRP16} addressed learning with such type of partial feedback in the context of repeated single-item second-price auctions. However, their approach does not address more complex auctions and is tailored to the second-price auction.

\medskip
\noindent
{\bf Our Contributions.} %
Our first main contribution is to introduce a novel online learning setting with partial feedback, which we denote \emph{learning with outcome-based feedback} and which could be of independent interest\cpcomment{since this is the common push back that we get, should we somehow stress that the "partiality" of the feedback stems from the probabilistic outcome and not by the info that we get once we choose an action?"}. We show that our setting captures online learning in many repeated auction scenarios including all types of single-item auctions, value-per-click sponsored search auctions, value-per-impression sponsored search auctions and multi-item auctions.

Our setting generalizes the setting of learning with feedback graphs \cite{MS11,ACGM13}, in a way that is crucial for applying it to the auction settings of interest. At a high level, the setting is defined as follows: The learner chooses an action $b\in B$ (e.g. a bid in an auction). The adversary chooses an \emph{allocation function} $x_t$, that maps an action to a distribution over a set of potential outcomes $O$ (e.g. the probability of getting a click) and a \emph{reward function} $r_t$ that maps an action-outcome pair to a reward (utility conditional on getting a click with a bid of $b$). Then, an outcome $o_t$ is chosen based on distribution $x_t(b)$ and a reward $r_t(b,o_t)$ is observed. The learner also gets to observe the function $x_t$ and the reward function $r_t(\cdot, o_t)$ for the realized outcome $o_t$ (i.e. in our auction setting: she learns the probability of a click, the expected payment as a function of her bid and, \emph{if she gets clicked}, her value). 

Our second main contribution is an algorithm which we call $\winexp$, which achieves regret $O\left(\sqrt{T|O|\log(|B|)}\right)$. The latter is inherently better than the generic multi-armed bandit regret of $O\left(\sqrt{T|B|}\right)$, since in most of our applications $|O|$ will be a small constant (e.g. $|O|=2$ in sponsored search) and takes advantage of the particular feedback structure. Our algorithm is a variant of the EXP3 algorithm \cite{ACBFS02}, with a carefully crafted unbiased estimate of the utility of each action, which has lower variance than the unbiased estimate used in the standard EXP3 algorithm. This result could also be of independent interest and applicable beyond learning in auction settings. Our approach is similar to the importance weighted sampling approach used in EXP3 so as to construct unbiased estimates of the utility of each possible action. Our main technical insight is how to incorporate the allocation function feedback that the bidder receives to construct unbiased estimates with small variance, leading to dependence only in the number of outcomes and not the number of actions. As we discuss in the related work, despite the several similarities, our setting has differences with existing partial feedback online learning settings, such as learning with experts \cite{ACBFS02}, learning with feedback graphs \cite{MS11,ACGM13} and contextual bandits \cite{Agarwal14}.

This setting engulfs learning in many auctions of interest where bidders learn their value for a good only when they win the good and where the good which is allocated to the bidder is determined by some randomized allocation function. For instance, when applied to the case of single-item first-price, second-price or all-pay auctions, our setting corresponds to the case where the bidders observe their value for the item auctioned at each iteration only when they win the item. Moreover, after every iteration, they observe the critical bid they would have needed to submit to win (for instance, by observing the bids of others or the clearing price). The latter is typically the case in most government auctions or in auction settings similar to eBay.

Our flagship application is that of \emph{value-per-click sponsored search auctions}. These are auctions were bidders repeatedly bid in an auction for a slot in a keyword impression on a search engine. The complexity of the sponsored search ecosystem and the large volume of repeated auctions has given rise to a plethora of automated bidding tools (see e.g. \cite{PPCTool}) and has made sponsored search an interesting arena for automated learning agents. Our framework captures the fact that in this setting the bidders observe their value for a click only when they get clicked. Moreover, it assumes that the bidders also observe the average probability of click and the average cost per click for any bid they could have submitted. The latter is exactly the type of feedback that the automated bidding tools can receive via the use of \emph{bid simulators} offered by both major search engines \cite{bidsim1,bidsim2,bidsim3,bidsim4}. In Figure \ref{fig:curves} we portray example interfaces from these tools, where we see that the bidders can observe exactly these allocation and payment curves assumed by our outcome-based-feedback formulation. Not using this information seems unreasonable and a waste of available information. Our work shows how one can utilize this partial feedback given by the auction systems to provide improved learning guarantees over what would have been achieved if one took a fully bandit approach. In the experimental section, we also show that our approach outperforms that of the bandit one even if the allocation and payment curves provided by the system have some error that could stem from errors in the machine learning models used in the calculation of these curves by the search engines. Hence, even when these curves are not fully reliable, our approach can offer improvements in the learning rate.
\begin{figure}[htpb]
\begin{subfigure}{.45\textwidth}
\includegraphics[scale=.25]{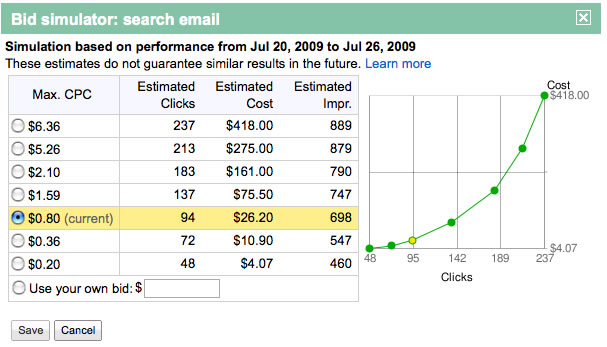}
\end{subfigure}
~~
\begin{subfigure}{.45\textwidth}
\includegraphics[scale=.41]{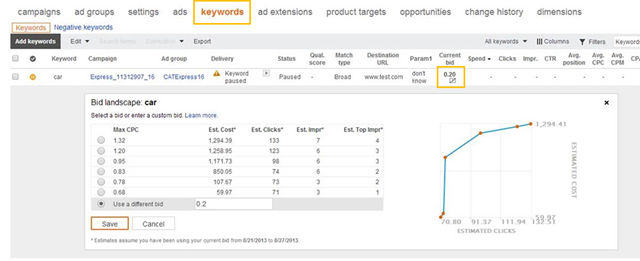}
\end{subfigure}
\captionsetup{justification=centering}
\caption{Example interfaces of bid simulators of two major search engines, Google Adwords (left) and BingAds (right), that enables learning the allocation and the payment function. (sources \cite{GoogleImage,MicrosoftImage})}\label{fig:curves}
\vspace{-8pt}
\end{figure}

We also extend our results to cases where the space of actions is a continuum (e.g. all bids in an interval $[0,1]$). We show that in many auction settings, under appropriate assumptions on the utility functions, a regret of $O\left(\sqrt{T\log(T)}\right)$ can be achieved by simply discretizing the action space to a sufficiently small uniform grid and running our $\winexp$ algorithm. This result encompasses the results of \cite{WRP16} for second price auctions, learning in first-price and all-pay auctions, as well as learning in sponsored search with smoothness assumptions on the utility function. We also show how smoothness of the utility can easily arise due to the inherent randomness that exists in the mechanism run in sponsored search.

Finally, we provide two further extensions: \emph{switching regret} and \emph{feedback-graphs over outcomes}. The former adapts our algorithm to achieve good regret against a sequence of bids rather than a fixed bid, which has implications on the faster convergence to approximate efficiency of the outcome (price of anarchy). Feedback graphs address the idea that in many cases the learner could be receiving information about other items other than the item he won (through correlations in the values for these items). This essentially corresponds to adding a feedback graph over outcomes and when outcome $o_t$ is chosen, then the learner learns the reward function $r_t(\cdot, o)$ for all neighboring outcomes $o$ in the feedback graph. We provide improved results that mainly depend on the dependence number of the graph rather than the number of possible outcomes.

\medskip
\noindent
{\bf Related Work.} Our work lies on the intersection of two main areas: No regret learning in Game Theory and Mechanism Design and Contextual Bandits.

\vspace{5pt}
\noindent
\emph{No regret learning in Game Theory and Mechanism Design.} No regret learning has received a lot of attention in the Game Theory and Mechanism Design literature \citep{CHN14}. Most of the existing literature, however, focuses on the problem from the side of the auctioneer, who tries to maximize revenue through repeated rounds without knowing a priori the valuations of the bidders \citep{ACDKR15,ARS14,BKRW04,BMM15,CBGM13,CR15,DRY15,KN14,MM14,OS11,MV17,Feldman2016,Koren17a}.
These works are centered around different auction formats like the sponsored search ad auctions, the pricing of inventory and the single-item auctions. 
Our work is mostly related to \citet{WRP16}, who adopt the point of view of the bidders in repeated second-price auctions and who also analyze the case where the true valuation of the item is revealed to the bidders only when they win the item. Their setting falls into the family of settings for which our novel and generic $\winexp$ algorithm produces good regret bounds and as a result, we are able to fully retrieve the regret that their algorithms yield, up to a tiny increase in the constants. Hence, we give an easier way to recover their results. Closely related to our work are the works of \cite{DT13} and \citep{BG17}. \citet{DT13} analyzes a setting where bidders have to experiment in order to learn their valuations, and show that the seller can increase revenue by offering an initial credit to them, in order to give them incentives to experiment. \citet{BG17} introduce a family of dynamic bidding strategies in repeated second-price auctions, where advertisers adjust their bids throughout the campaign. They analyze both regret minimization and market stability. There are two key differences to our setting; first, Balseiro and Gur consider the case where the goal of the bidders is the expediture rate in a way that guarantees that the available campaign budget will be spent in an optimal \emph{pacing} way and second, because of their target being the expenditure rate at every timestep $t$, they assume that the bidders get information about the value of the slot being auctioned and based on this information they decide how to adjust their bid. Moreover, several works analyze the properties of auctions when bidders adopt a no-regret learning strategy \citep{BHLR08,CKKKLLT12,R09}. None of these works, however, addresses the question of learning more efficiently in the unknown valuation model and either invokes generic MAB algorithms or develops tailored full information algorithms when the bidder knows his value. Another line of research takes a Bayesian approach to learning in repeated auctions and makes large market assumptions, analyzing learning to bid with an unknown value under a Mean Field Equilibrium condition \citep{AJ13,IJS11,BBW15}\footnote{No-regret learning is complementary and orthogonal to the mean field approach, as it does not impose any stationarity assumption on the evolution of valuations of the bidder or the behavior of his opponents.}.

\vspace{5pt}
\noindent
\emph{Learning with partial feedback.} Our work is also related to the literature in \emph{learning with partial feedback} \cite{Agarwal14,BCB2012}. To establish this connection we observe that the \emph{policies} and the \emph{actions} in contextual bandit terminology translate into \emph{discrete bids} and \emph{groups of bids for which we learn the rewards} in our work. The difference between these two is the fact that for each \emph{action} in contextual bandits we get a single reward, whereas for our setting we observe a \emph{group} of rewards; one for each action in the group. Moreover, the fact that we allow for randomized outcomes adds extra complication, non existent in contextual bandits.
In addition, our work is closely related to the literature in \emph{online learning with feedback graphs} \citep{ACDK15,ACGM13,CHK16,MS11}. In fact, we propose a new setting in online learning, namely, \emph{learning with outcome-based feedback}, which is a generalization of learning with feedback graphs and is essential when applied to a variety of auctions which include sponsored search, single-item second-price, single-item first-price and single-item all-pay auctions. Moreover, the fact that the learner only learns the probability of each outcome and not the actual realization of the randomness, is similar in nature to a feedback graph setting, but where the bidder does not observe the whole graph. Rather, she observes a distribution over feedback graphs and for each bid she learns with what probability each feedback graph would arise. For concreteness, consider the case of sponsored search and suppose for now that the bidder gets even more information than what we assume and also observes the bids of her opponents. She still does not observe whether she would get a click if she falls on the slot below but only the probability with which she would get a click in the slot below. If she could observe whether she would still get a click in the slot below, then we could in principle construct a feedback graph that would say that for all bids were the bidder gets a slot her reward is revealed, and for every bid that she does not get a click, her reward is not revealed. However, this is not the structure that we have and essentially this corresponds to the case where the feedback graph is not revealed, as analyzed in \cite{CHK16} and for which no improvement over the full bandit feedback is possible. However, we show that this impossibility is amended by the fact that the learner observes the probability of a click and hence for each possible bid, she observes the probability with which each feedback graph would have happened. This is enough for a low variance unbiased estimate.

%% file: preliminaries.tex
For simplicity of exposition, we start with a simple single-dimensional mechanism design setting, but our results extend to multi-dimensional (multi-item) mechanisms, as we will see in Section \ref{sec:2}. Let $n$ be the number of bidders. Each bidder has a value $v_i\in [0, 1]$ \emph{per-unit of a good} and submits a bid $b_i\in B$, where $B$ is a discrete set of bids (e.g. a uniform $\epsilon$-grid of $[0,1]$). Given the bid profile of all bidders, the auction allocates a unit of the good to the bidders. The allocation rule for bidder $i$ is given by  $X_i(b_i, b_{-i})$. Moreover, the mechanism defines a per-unit payment function $P_i(b_i, b_{-i})\in [0,1]$. The overall utility of the bidder is quasi-linear, i.e. $u_i(b_i, b_{-i}) = \left(v_i - P_i(b_i, b_{-i})\right) \cdot X_i(b_i, b_{-i})$.

\paragraph{Online Learning with Partial Feedback.} The bidders participate in this mechanism repeatedly. At each iteration, each bidder has some value $v_{it}$ and submits a bid $b_{it}$. The mechanism has some time-varying allocation function $X_{it}(\cdot, \cdot)$ and payment function $P_{it}(\cdot, \cdot)$. We assume that the bidder does \emph{not} know her value $v_{it}$, nor the bids of her opponents $b_{-i,t}$, nor the allocation and payment functions, \emph{before} submitting a bid. 

At the end of each iteration, she gets an item with probability $X_{it}(b_{it}, b_{-i,t})$ and observes her value $v_{it}$ for the item only when she gets one \footnote{E.g. in sponsored search, the good allocated is the probability of getting clicked, and you only observe your value if you get clicked.}. Moreover, we assume that she gets to observe her allocation and payment functions for that iteration, i.e. she gets to observe two functions $x_{it}(\cdot) = X_{it}(\cdot, b_{-i,t})$ and $p_{it}(\cdot) = P_{it}(\cdot, b_{-i,t})$. Finally, she receives utility $\left(v_{it} - p_{it}(b_{it})\right) \cdot  \1\{\text{item is allocated to her}\}$ or in other words expected utility $u_{it}(b_{it})= \left(v_{it} - p_{it}(b_{it})\right)\cdot x_{it}(b_{it})$. 
Given that we focus on learning from the perspective of a single bidder we will drop the index $i$ from all notation and instead write, $x_t(\cdot)$, $p_t(\cdot)$, $u_t(\cdot)$, $v_t$, etc. The goal of the bidder is to achieve small expected regret with respect to any fixed bid in hindsight: $R(T) = \sup_{b^*\in B} \E\left[\sum_{t=1}^T \left(u_t(b^*) - u_t(b_t)\right)\right] \leq o(T)$. %

%% file: win-only-learning.tex
Let us abstract the learner's problem to a setting that could be of interest beyond auction settings. 
\paragraph{Learning with Win-Only Feedback.} Every day a learner picks an action $b_t$ from a finite  set $B$. The adversary chooses a reward function $r_t: B\rightarrow [-1,1]$ and an allocation function $x_t: B \rightarrow [0,1]$. The learner wins a reward $r_t(b)$ with probability $x_t(b)$. Let $u_t(b)=r_t(b)x_t(b)$ be the learner's expected utility from action $b$. After each iteration, if she won the reward then she learns the whole reward function $r_t(\cdot)$, while she \emph{always} learns the allocation function $x_t(\cdot)$.
\begin{quote}{\em Can the learner achieve regret $O(\sqrt{T \log(|B|)})$ rather than bandit-feedback regret $O(\sqrt{T |B|})$?}
\end{quote}
In the case of the auction learning problem, the reward function $r_t(b)$ takes the parametric form $r_t(b) = v_t - p_t(b)$ and the learner needs to learn $v_t$ and $p_t(\cdot)$ at the end of each iteration, when she wins the item. This is in line with the feedback structure we described in the previous section.

We consider the following adaptation of the EXP3 algorithm with unbiased estimates based on the information received. It is notationally useful throughout the section to denote with $A_t$ the event of \emph{winning a reward at time $t$}. Then, we can write: $\Pr[A_t|b_t=b] = x_t(b)$ and $\Pr[A_t]=\sum_{b\in B}\pi_t(b) x_t(b)$, where with $\pi_t(\cdot)$ we denote the multinomial distribution from which bid $b$ is drawn. With this notation we define our $\winexp$ algorithm in Algorithm \ref{alg:winexp}. We note here that our generic family of the $\winexp$ algorithms can be parametrized by the step-size $\eta$, the estimate of the utility $\tilde{u}_t$ that the learner gets at each round and the feedback structure that she receives.

\vspace{-2pt}
\begin{algorithm}[h]
\begin{algorithmic}
\State Let $\pi_1(b) =\frac{1}{|B|}$ for all $b\in B$ (i.e. the uniform distribution over bids), $\eta = \sqrt{\frac{2\log\left(|B| \right)}{5T}}$
\For{each iteration t}
\State Draw a bid $b_t$ from the multinomial distribution based on $\pi_t(\cdot)$
\State Observe $x_t(\cdot)$ and if reward is won also observe $r_t(\cdot)$
\State Compute estimate of utility:

\State\indent If reward is won $\tilde{u}_t(b) = \frac{(r_t(b) - 1) \Pr[A_t | b_t=b]}{\Pr[A_t]}$; otherwise,  $\tilde{u}_t(b) = -\frac{\Pr[\neg A_t | b_t = b]}{\Pr[\neg A_t]}$.
\State Update $\pi_t(\cdot)$ as in Exponential Weights Update: $\forall b\in B: \pi_{t+1}(b) \propto \pi_{t}(b)\cdot \exp\left\{\eta \cdot \tilde{u}_t(b)\right\}$
\EndFor
\end{algorithmic}
\caption{$\winexp$ algorithm for learning with win-only feedback}\label{alg:winexp}
\end{algorithm}
\paragraph{Bounding the Regret.} 
We first bound the first and second moment of the unbiased estimates built at each iteration in the $\winexp$ algorithm.
\begin{lemma}\label{lem:moments}
At each iteration $t$, for any action $b\in B$, the random variable $\tilde{u}_t(b)$ is an unbiased estimate of the true expected utility $u_t(b)$, i.e.: $\forall b\in B: \E\left[\tilde{u}_t(b)\right] = u_t(b)-1$ and has expected second moment bounded by: $\forall b\in B: \E\left[\left(\tilde{u}_t(b)\right)^2\right]\leq \frac{4\Pr[A_t|b_t=b]}{\Pr[A_t]} + \frac{\Pr[\neg A_t|b_t=b]}{\Pr[\neg A_t]}$.
\end{lemma}

\begin{proof}
Let $A_t$ denote the event that the reward was won. We have:
\begin{align*}
\E\left[\tilde{u}_t(b)\right] &= \E\left[\frac{\left(r_t(b)-1\right)\cdot \Pr[A_t|b_t=b]}{\Pr[A_t]}\1\{A_t\} - \frac{\Pr[\neg A_t|b_t=b]}{\Pr[\neg A_t]}\1\{\neg A_t\}\right]\\
& =(r_t(b) - 1) \Pr[A_t | b_t=b] - \Pr[\neg A_t|b_t=b]\\
& = r_t(b) \Pr[A_t | b_t=b] - 1 = u_t(b) - 1
\end{align*}
Similarly for the second moment:
\begin{align*}
\E\left[\tilde{u}_t(b)^2\right] &= \E\left[\frac{\left(r_t(b)-1\right)^2\cdot \Pr[A_t|b_t=b]^2}{\Pr[A_t]^2}\1\{A_t\} +  \frac{\Pr[\neg A_t|b_t=b]^2}{\Pr[\neg A_t]^2}\1\{\neg A_t\} \right]\\
&= \frac{\left(r_t(b)-1\right)^2\cdot \Pr[A_t|b_t=b]^2}{\Pr[A_t]} +  \frac{\Pr[\neg A_t|b_t=b]^2}{\Pr[\neg A_t]} \leq \frac{4 \Pr[A_t|b_t=b]}{\Pr[A_t]} + \frac{\Pr[\neg A_t|b_t=b]}{\Pr[\neg A_t]}
\end{align*}
where the last inequality holds since $r_t ( \cdot ) \in [-1,1]$ and $x_t(\cdot)\in [0,1]$.
\end{proof}

We are now ready to prove our main theorem:
\begin{theorem}[Regret of $\winexp$]\label{thm:win-only-regr}
The regret of the $\winexp$ algorithm with the aforementioned unbiased estimates and step size $\sqrt{\frac{2\log(|B|)}{5T}}$ is: $4\sqrt{T\log(|B|)}$.
\end{theorem}

\begin{proof}
Observe that regret with respect to utilities $u_t(\cdot)$ is equal to regret with respect to the translated utilities $u_t(\cdot)-1$. We use the fact that the exponential weights update with an unbiased estimate $\tilde{u}_t(\cdot) \leq 0$ of the true utilities, achieves expected regret of the form\footnote{A detailed proof of this claim can be found in Appendix~\ref{appendix:a}.}:
\begin{align*}
R(T) \leq~& \frac{\eta}{2} \sum_{t=1}^T \sum_{b\in B} \pi_t(b) \cdot \E\left[\left(\tilde{u}_t(b)\right)^2\right] + \frac{1}{\eta} \log(|B|)
\end{align*}
Invoking the bound on the second moment by Lemma \ref{lem:moments}, we get:
\begin{align*}
R(T) \leq~& \frac{\eta}{2} \sum_{t=1}^T \sum_{b\in B} \pi_t(b) \cdot \left(\frac{4\Pr[A_t|b_t=b]}{\Pr[A_t]} + \frac{\Pr[\neg A_t|b_t=b]}{\Pr[\neg A_t]}\right) + \frac{1}{\eta} \log(|B|) \leq \frac{5}{2}\eta T + \frac{1}{\eta} \log(|B|)
\end{align*}
Picking $\eta = \sqrt{\frac{2\log(|B|)}{5T}}$, we get the theorem.
\end{proof}

%% file: outcome-based-learning.tex
\label{sec:2}
In the win-only feedback framework there were two possible outcomes that could happen: either you win the reward or not. We now consider a more general problem, where there are more than two outcomes and you learn your reward function for the outcome you won. Moreover, the outcome that you won is also a probabilistic function of your action. The proofs for the results presented in this section can be found in Appendix~\ref{appendix:outcome}. 

\paragraph{Learning with Outcome-Based Feedback.} Every day a learner picks an action $b_t$ from a finite  set $B$. There is a set of payoff-relevant outcomes $O$. The adversary chooses a reward function $r_t: B\times O\rightarrow [-1,1]$, which maps an action and an outcome to a reward and he also chooses an allocation function $x_t: B \rightarrow \Delta(O)$, which maps an action to a distribution over the outcomes. Let $x_t(b, o)$ be the probability of outcome $o$ under action $b$. An outcome $o_t\in O$ is chosen based on distribution $x_t(b_t)$. The learner wins reward $r_t(b_t, o_t)$ and observes the whole outcome-specific reward function $r_t(\cdot, o_t)$. She \emph{always} learns the allocation function $x_t(\cdot)$ after the iteration. Let $u_t(b)=\sum_{o\in O} r_t(b, o)\cdot x_t(b, o)$ be the expected utility from action $b$. 

We consider the following adaptation of the EXP3 algorithm with unbiased estimates based on the information received. It is notationally useful throughout the section to consider $o_t$ as the random variable of the outcome chosen at time $t$. Then, we can write: $\Pr_t[o_t|b] = x_t(b, o_t)$ and $\Pr_t[o_t]=\sum_{b\in B}\pi_t(b) \Pr_t[o_t|b]=\sum_{b\in B}\pi_t(b)\cdot x_t(b, o_t)$. With this notation and based on the feedback structure, we define our $\winexp$ algorithm for learning with outcome-based feedback in %
Algorithm \ref{alg:winexp2}. %

\begin{theorem}[Regret of $\winexp$ with outcome-based feedback]\label{thm:outcome-based-main}
The regret of Algorithm \ref{alg:winexp2} with $\tilde{u}_t(b) = \frac{(r_t(b, o_t) - 1) \Pr_t[o_t | b]}{\Pr_t[o_t]}$ and step size $\sqrt{\frac{\log(|B|)}{2T|O|}}$ is: $2\sqrt{2T|O|\log(|B|)}$.
\end{theorem}

\begin{algorithm}[H]
	\begin{algorithmic}
		\State Let $\pi_1(b) =\frac{1}{|B|}$ for all $b\in B$ (i.e. the uniform distribution over bids), $\eta = \sqrt{\frac{\log\left(|B| \right)}{2T|O|}}$
		\For{each iteration t}
		\State Draw an action $b_t$ from the multinomial distribution based on $\pi_t(\cdot)$
		\State Observe $x_t(\cdot)$, observe chosen outcome $o_t$ and associated reward function $r_t(\cdot, o_t)$
		\State Compute estimate of utility: 
		\begin{equation}
		\tilde{u}_t(b) = \frac{(r_t(b, o_t) - 1) \Pr_t[o_t | b]}{\Pr_t[o_t]}
		\end{equation}
		\State Update $\pi_t(\cdot)$ based on the Exponential Weights Update: 
		\begin{equation}
		\forall b\in B: \pi_{t+1}(b) \propto \pi_{t}(b)\cdot \exp\left\{\eta \cdot \tilde{u}_t(b)\right\}
		\end{equation}
		\EndFor
	\end{algorithmic}
	\caption{$\winexp$ algorithm for learning with outcome-based feedback}\label{alg:winexp2}
\end{algorithm}

\paragraph{Applications to Learning in Auctions.} We now present a series of applications of the main result of this section to several learning in auction settings, even beyond single-item or single-dimensional ones.

\begin{example}[Second-price auction]\label{ex:weed}
Suppose that the mechanism ran at each iteration is just the second price auction. Then, we know that the allocation function $X_i(b_i, b_{-i})$ is simply of the form: $\1\{b_i\geq \max_{j\neq i} b_j\}$ and the payment function is simply the second highest bid. In this case, observing the allocation and payment functions at the end of the auction boils down to observing the highest other bid. %
In fact, in this case we have a trivial setting where the bidder gets an allocation of either $0$ or $1$ and if we let $B_t=\max_{j\neq i}b_{jt}$, then the unbiased estimate of the utility takes the simpler form (assuming the bidder always loses in case of ties) of: if $b_t > B_t$ : $\tilde{u}_t(b)= \frac{(v_{it} - B_{t}-1)\1\{b> B_{t}\}}{\sum_{b'>B_{t}} \pi_t(b')}$ and $\tilde{u}_t(b)=\frac{\1\{b\leq B_{t}\}}{\sum_{b'\leq B_{t}} \pi_t(b')}$ in any other case.
Our main theorem gives regret $4\sqrt{T\log(|B|)}$. We note that this theorem recovers exactly the results of \citet{WRP16}, by using as $B$ a uniform $1/\Delta^o$ discretization of the bidding space, for an appropriately defined constant $\Delta^o$ (see Appendix~\ref{sec:app-weed} for an exact comparison of the results).
\end{example}

\begin{example}[Value-per-click auctions]
This is a variant of the binary outcome case analyzed in Section \ref{sec:win-only}, where $O=\{A, \neg A\}$, i.e. get clicked or not. Hence, $|O|=2$, and $r_t(b, A) = v_t - p_t(b)$, while $r_t(b, \neg A)=0$. Our main theorem gives regret $4\sqrt{T\log(|B|)}$.
\end{example}

\begin{example}[Unit-demand multi-item auctions]
Consider the case of $K$ items at an auction where the bidder has value $v_{k}$ for only one item $k$. Given a bid $b$, the mechanism defines a probability distribution over the items that the bidder will be allocated and also defines a payment function, which depends on the bid of the bidder and the item allocated. When a bidder gets allocated an item $k$ she gets to observe her value $v_{kt}$ for that item. Thus, the set of outcomes is equal to $O=\{1,\ldots,K+1\}$, with outcome $K+1$ associated with not getting any item. The rewards are also of the form: $r_t(b, k) = v_{kt} - p_{t}(b, k)$ for some payment function $p_t(b,k)$ dependent on the auction format. Our main theorem then gives regret $2\sqrt{2 (K+1) T\log(|B|)}$.
\end{example}

%% file: batch-rewards.tex
We now consider the case of sponsored search auctions, where the learner participates in multiple auctions per-iteration. At each of these auctions she has a chance to win and get feedback on her value. To this end, we abstract the \emph{learning with win-only feedback} setting to a setting where multiple rewards are awarded per-iteration. The allocation function remains the same throughout the iteration but the reward functions can change. 

\paragraph{Outcome-Based Feedback with Batch Rewards.} Every iteration $t$ is associated with a set of \emph{reward contests} $I_t$. The learner picks an action $b_t$, which is used at \emph{all} reward contests. For each $\tau\in I_t$ the adversary picks an outcome specific reward function $r_\tau: B\times O \rightarrow [-1,1]$. Moreover, the adversary chooses an allocation function $x_t: B\rightarrow \Delta(O)$, which is not $\tau$-dependent. At each $\tau$, an outcome $o_\tau$ is chosen based on distribution $x_t(b_t)$ and independently. The learner receives reward $r_\tau(b_t, o_\tau)$ from that contest. The overall realized utility from that iteration is the average reward: $\frac{1}{|I_t|}\sum_{\tau\in I_t} r_\tau(b_t, o_{\tau})$, 
while the expected utility from any bid $b$ is: $u_t(b) = \frac{1}{|I_t|}\sum_{\tau\in I_t} \sum_{o\in O} r_\tau(b, o) \cdot x_t(b, o)$. We assume that at the end of each iteration the learner receives as feedback the average reward function conditional on each realized outcome, i.e. if we let $I_{to} = \{\tau \in I_t: o_\tau = o\}$, then the learner learns the function: $Q_t(b, o) = \frac{1}{|I_{to}|}\sum_{\tau\in I_{to}} r_\tau(b, o)$
(with the convention that $Q_t(b, o)=0$ if $|I_{to}|=0$) as well as the realized frequencies $f_t(o) = \frac{|I_{to}|}{|I_t|}$ for all outcomes $o$. 

With this at hand we can define the \emph{batch-analogue} of our unbiased estimates of the previous section. To avoid any confusion we define: $\Pr_t[o|b] = x_t(b, o)$ and $\Pr_t[o]=\sum_{b\in B}\pi_t(b) \Pr_t[o|b]$, to denote that these probabilities only depend on $t$ and not on $\tau$. The estimate of the utility will be:  
\begin{equation}\label{eqn:batch-unbiased-2}
\tilde{u}_t(b) = \sum_{o \in O} \frac{\Pr_t \left[o|b \right]}{\Pr_t[o]} f_t(o) \left(Q_t(b,o)-1 \right)
\end{equation}

We show the full algorithm with outcome-based batch-reward feedback in Algorithm \ref{alg:winexp4}.

\begin{algorithm}[H]
	\begin{algorithmic}
		\State Let $\pi_1(b) =\frac{1}{|B|}$ for all $b\in B$ (i.e. the uniform distribution over bids), $\eta = \sqrt{\frac{\log\left(|B| \right)}{2T|O|}}$
		\For{each iteration t}
		\State Draw an action $b_t$ from the multinomial distribution based on $\pi_t(\cdot)$
		\State Observe $x_t(\cdot)$, chosen outcomes $o_\tau, \forall \tau \in I_t$, average reward function conditional on each realized outcome $Q_t(b,o)$ and the realized frequencies for each outcome $f_t(o) = \frac{|I_{to}|}{|I_t|}$.
		\State Compute estimate of utility: 
		\begin{equation}
		\tilde{u}_t(b) = \sum_{o \in O} \frac{\Pr_t \left[o|b \right]}{\Pr_t[o]} f_t(o) \left(Q_t(b,o)-1 \right)
		\end{equation}
		\State Update $\pi_t(\cdot)$ based on the Exponential Weights Update: 
		\begin{equation}
		\forall b\in B: \pi_{t+1}(b) \propto \pi_{t}(b)\cdot \exp\left\{\eta \cdot \tilde{u}_t(b)\right\}
		\end{equation}
		\EndFor
	\end{algorithmic}
	\caption{$\winexp$ algorithm for learning with outcome-based batch-reward feedback}\label{alg:winexp4}
\end{algorithm}

\begin{corollary}\label{corol:batch-rewards} The $\winexp$ algorithm with the latter unbiased utility estimates and step size $\sqrt{\frac{\log(|B|)}{2T|O|}}$, achieves regret in the outcome-based feedback with batch rewards setting at most: $2\sqrt{2T|O|\log(|B|)}$.
\end{corollary}
It is also interesting to note that the same result holds if instead of using $f_t(o)$ in the expected utility (Equation (10)%
), we used its \emph{mean value}, which is $x_t(o,b_t)=\Pr_t[o|b_t]$. This would not change any of the derivations above. The nice property of this alternative is that the learner does not need to learn the realized fraction of each outcome, but only the expected fraction of each outcome. This is already contained in the function $x_t(\cdot, \cdot)$, which we assumed was given to the learner at the end of each iteration. Thus, with these new estimates, the learner does not need to observe $f_t(o)$. In Appendix~\ref{appendix:notes} we also discuss the case where different periods can have different number of rewards and how to extend our estimate to that case.
The batch rewards setting finds an interesting application in the case of learning in sponsored search, as we describe below.
\begin{example}[Sponsored Search]\label{ex:ss}
In the case of sponsored search auctions, the latter boils down to learning the average value $\hat{v} = \frac{1}{\#clicks}\sum_{clicks} v_{click}$ for the clicks that were generated, as well as the cost-per-click function $p_t(b)$, which is assumed to be constant throughout the period $t$. Given these quantities, the learner can compute: $Q(b, A) = \hat{v} - p_t(b)$ and $Q(b, \neg A) =0$. An advertiser can keep track of the traffic generated by a search engine ad and hence, can keep track of the number of clicks from the search engine and the value generated by each of these clicks (conversion). Thus, she can estimate $\hat{v}$. Moreover, she can elicit the probability of click (aka click-through-rate or CTR) curves $x_t(\cdot)$ and the cost-per-click (CPC) curves $p_t(\cdot)$ over relatively small periods of time of about a few days. See for instance the Adwords bid simulator tools offered by Google \cite{bidsim1,bidsim2,bidsim3,bidsim4}\footnote{One could argue that the CTRs that the bidder gets in this case are not accurate enough. Nevertheless, even if they have random perturbations, we show in our experimental results that for reasonable noise assumptions, $\winexp$ is preferrable compared to EXP3.}. 
Thus, with these at hand we can apply our batch reward outcome based feedback algorithm and get regret that does not grow linearly with $|B|$, but only as $4\sqrt{T\log\left(|B|\right)}$. Our main assumption is that the expected CTR and CPC curves during this relatively small period of a few days remains approximately constant. The latter holds if the distribution of click-through-rates does not change within these days and if the bids of opponent bidders also do not significantly change. This is a reasonable assumption when feedback can be elicited relatively frequently, which is the case in practice.
\end{example}

%% file: continuous-lipschitz.tex
Until now we only considered discrete action spaces. In this section, we extend our discussion to continuous ones; that is, we allow the action of each bidder to lie in a continuous action space $\B$ %
(e.g. a uniform interval in $[0,1]$). %
To assist us in our analysis, we are going to use the discretization result in \cite{K05} \footnote{\citet{K05} discuss the uniform discretization of continuum-armed bandits and \citet{KSU08} extend the results for the case of Lipschitz-armed bandits.}. For what follows in this section, let $R(T, \B) = \sup_{b^*\in \B} \E\left[\sum_{t=1}^T \left(u_t(b^*) - u_t(b_t)\right)\right]$
be the regret of the bidder, after $T$ rounds with respect to an action space $\B$. Moreover, for any pairs of action spaces $B$ and $\B$ we let: $DE(B, \B) = \sup_{b \in \B} \sum_{t=1}^T u_t(b) - \sup_{b' \in B} \sum_{t=1}^T u_t(b')$,
denote the discretization error incurred by optimizing over $B$ instead of $\B$. The proofs of this section can be found in Appendix~\ref{appendix:cont}.

\begin{lemma}(\cite{K05,KSU08})\label{lem:regret-de}
Let $\B$ be a continuous action space and $B$ a discretization of $\B$. Then:
\begin{equation*}
R(T, \B) \leq R(T, B) + DE(B, \B)
\end{equation*}
\end{lemma}

Observe now that in the setting of \citet{WRP16} the discretization error was: $DE(B, \B)=0$ if $\eps < \Delta^o$, as we discussed in Section \ref{sec:2} and that was \emph{the key} that allowed us to recover this result, without adding an extra $\eps T$ in the regret of the bidder. To achieve that, we construct the following general class of utility functions:

\begin{definition}[$\Delta^o$-Piecewise Lipschitz Average Utilities]\label{defn:piecewise}
A learning setting with action space $\B=[0,1]^d$, is said to have $\Delta^o$-Piecewise Lipschitz Cumulative Utilities if the average utility function $\frac{1}{T}\sum_{t=1}^T u_t(b)$ satisfies the following conditions: the bidding space $[0,1]^d$ is divided into $d$-dimensional cubes with edge length at least $\Delta^o$ and within each cube the utility is $L$-Lipschitz with respect to the $\ell_{\infty}$ norm. Moreover, for any boundary point there exists a sequence of non-boundary points whose limit cumulative utility is at least as large as the cumulative utility of the boundary point.
\end{definition}

\begin{lemma}[Discretization Error for Piecewise Lipschitz]\label{lem:de-piece}
Let $\B = [0,1]^d$ be a continuous action space and $B$ a uniform $\eps$-grid of $[0,1]^d$, such that $\eps < \Delta^o$ (i.e. $B$ consists of all the points whose coordinates are multiples of a given $\eps$). %
Assume that the average utility function is $\Delta^o$-Piecewise $L$-Lipschitz. Then, the discretization error of $B$ is bounded as: $DE(B, \B) \leq \eps L T$.
\end{lemma}

If we know the Lipschitzness constant $L$ %
mentioned above, the time horizon $T$ and $\Delta^o$, then our $\winexp$ algorithm for Outcome-Based Feedback with Batch Rewards yields regret as defined by the following theorem. In Appendix~\ref{sec:app-doubling-trick}, we also show how to deal with unknown parameters $L$, $T$ and $\Delta^o$ by applying a standard doubling trick.

\begin{theorem}\label{thm:continuous-lipschitz-known}
Let $\B=[0,1]^d$ be the action space as defined in our model and let $B$ be a uniform $\eps$-grid of $\B$. The $\winexp$ algorithm with unbiased estimates given by \inlineequation[eqn:batch-unbiased]{\tilde{u}_t(b) = \sum_{o\in O}\frac{\Pr_t[o | b]}{\Pr_t[o]} f_t(o) \left(Q_t(b, o) - 1\right)} on $B$ with $\eta = \sqrt{\frac{\log(|B|)}{2T|O|}}$,  %
$\eps=\min \left\{\frac{1}{L T},\Delta^o \right\}$
achieves expected regret at most %
$2\sqrt{2T|O|d\log \left(\max\left\{\frac{1}{\Delta^o},LT\right\} \right)} + 1$ in the outcome-based feedback with batch rewards and $\Delta^o$-Piecewise $L-$Lipschitz average utilities \footnote{Interestingly, the above regret bound can help to retrieve two familiar expressions for the regret. First, when $L=0$ (i.e. when the function is constant within each cube), in is the case for the second price auction analyzed in \cite{WRP16}, $R(T) = 2\sqrt{2dT|O|\log\left(\frac{1}{\Delta^o}\right)} + 1$. Hence, we recover the bounds from the prior sections up to a tiny increase. Second, when $\Delta^o \to \infty$, then we have functions that are $L$-Lipschitz in the whole space $\B$ and the regret bound that we retrieve is: $R(T) = 2\sqrt{2dT|O|\log\left(LT\right)} +1$, which is of the type achieved in continuous lipschitz bandit settings.}.
\end{theorem}

\begin{example}[First Price and All-Pay Auctions]
Consider the case of learning in first price or all-pay auctions. In the former, the highest bidder wins and pays her bid, while in the latter the highest bidder wins and every bidder pays her bid whether she wins or loses. Let $B_t$ be the highest other bid at time $t$. Then the average hindsight utility of the bidder in each auction is \footnote{For simplicity, we assume the bidder loses in case of ties, though we can handle arbitrary random tie-breaking rules.}:
\begin{align}
\textstyle{\frac{1}{T} \sum_{t=1}^T u_t(b)} =~& \textstyle{\frac{1}{T}\sum_{t=1}^T v_t\cdot \1\{b > B_t\}  - b \cdot \frac{1}{T}\sum_{t=1}^T \1\{b > B_t\}} \tag{first price}\\
\textstyle{\frac{1}{T} \sum_{t=1}^T u_t(b)} =~& \textstyle{\frac{1}{T}\sum_{t=1}^T v_t \cdot \1\{b > B_t\} - b} \tag{all-pay}
\end{align}
Let $\Delta^o$ be the \emph{smallest} difference between the highest other bid at any two iterations $t$ and $t'$ \footnote{This is an analogue of the $\Delta^o$ used by \cite{WRP16} in second price auctions.}. Then observe that the average utilities in this setting are $\Delta^o$-Piecewise $1$-Lipschitz: Between any two highest other bids, the average allocation, $\frac{1}{T}\sum_{t=1}^T v_t\cdot \1\{b > B_t\}$, of the bidder remains constant and the only thing that changes is his payment which grows linearly. Hence, the derivative at any bid between any two such highest other bids is upper bounded by $1$. Hence, by applying Theorem \ref{thm:continuous-lipschitz-known}, our $\winexp$ algorithm with a uniform discretization on a $\epsilon$-grid, for $\epsilon = \min\left\{\Delta^o, \frac{1}{T}\right\}$, achieves regret $4\sqrt{T\log\left(\max\left\{\frac{1}{\Delta^o},T\right\} \right)}) +1$, where we used that $|O|=2$ and $d=1$ for any of these auctions.

\end{example}

\input{sponsored-search.tex}

%% file: sponsored-search.tex
\subsection{Sponsored Search with Lipschitz Utilities}\label{sec:sponsored-lipschitz}
In this subsection, we extend our analysis of learning in the sponsored search auction model (Example \ref{ex:ss}) to the continuous bid space case, i.e., each bidder can submit a bid $b\in [0,1]$. As a reminder, the utility function is: $u_t(b) = x_t(b)(\hat{v}_t - p_t(b))$, 
where $b\in [0,1]$,  $\hat{v}_t \in [0,1]$ is the average value for the clicks at iteration $t$, $x_t(\cdot)$ is the CTR curve and $p_t(\cdot)$ is the CPC curve. These curves are implicitly formed by running some form of a Generalized Second Price auction (GSP) at each iteration to determine the allocation and payment rules. As we show in this subsection, the form of the GSP ran in reality gives rise to Lipschitz utilities, under some minimal assumptions. Therefore, we can apply the results in Section \ref{sec:continuous} to get regret bounds even with respect to the continuous bid space $\B=[0,1]$ \footnote{The aforementioned Lipschitzness is also reinforced by real world data sets from Microsoft's sponsored search auction system.}. We begin by providing a brief description of the type of Generalized Second Price auction ran in practice.

\begin{definition}[Weighted-GSP]\label{weighted-gsp}
Each bidder $i$ is assigned a \emph{quality score} $s_i\in [0,1]$. Bidders are ranked according to their score-weighted bid $s_i\cdot b_i$, typically called the \emph{rank-score}. Every bidder whose rank-score does not pass a reserve $r$ is discarded. Bidders are allocated slots in decreasing order of \emph{rank-score}. Each bidder is charged per-click the lowest bid she could have submitted and maintained the same slot. Hence, if a bidder $i$ is allocated a slot $k$ and $\rho_{k+1}$ is the rank-score of the bidder in slot $k+1$, then she is charged $\rho_{k+1}/s_i$ per-click. We denote with $U_i(\mbf{b}, \mbf{s}, r)$, the utility of bidder $i$ under a bid profile $\mbf{b}$ and score profile $\mbf{s}$.
\end{definition}
The quality scores are typically highly random, dependent on the features of the ad and the user that is currently viewing the page. Hence, a reasonable modeling assumption is that the scores $s_i$ at each auction are drawn i.i.d. from some distribution with CDF $F_i$. We now show that if the CDF $F_i$ is Lipschitz (i.e. admits a bounded density), then the utilities of the bidders are also Lipschitz.

\begin{theorem}[Lipschitzness of the utility of Weighted GSP]\label{thm:lipschitz-weighted-gsp}
Suppose that the score $s_i$ of each bidder $i$ in a weighted GSP is drawn independently from a distribution with an $L-$Lipschitz CDF $F_i$. Then, the expected utility $u_i(b_i, \mbf{b}_{-i}, r) = \E_{\mbf{s}}\left[U_i(b_i, \mbf{b}_{-i}, \mbf{s}, r)\right]$ is $\frac{2nL}{r}-$Lipschitz wrt $b_i$.
\end{theorem}

Thus, we see that when the quality scores in sponsored search are drawn from $L$-Lipschitz CDFs $F_i, \forall i \in n$ and the reserve is lower bounded by $\delta>0$, then the utilities are $\frac{2nL}{\delta}$-Lipschitz and we can achieve good regret bounds by using the $\winexp$ algorithm with batch rewards, with action space $B$ being a uniform $\epsilon$-grid, $\epsilon = \frac{\delta}{2nLT}$ and unbiased estimates given by Equation~\eqref{eqn:batch-unbiased} or Equation \eqref{eqn:batch-unbiased-2}. In the case of sponsored search the second unbiased estimate takes the following simple form:
\begin{equation}
\textstyle{\tilde{u}_t(b) = \frac{x_t(b)\cdot x_t(b_t)}{\sum_{b'\in B} \pi_t(b') x_t(b')}  \left(\hat{v}_t - p_t(b)- 1\right) - \frac{(1-x_t(b))\cdot (1-x_t(b_t))}{\sum_{b'\in B} \pi_t(b') (1-x_t(b'))}}
\end{equation}
where $\hat{v}_t$ is the average value from the clicks that happened during iteration $t$, $x_t(\cdot)$ is the CTR curve, $b_t$ is the realized bid that the bidder submitted and $\pi_t(\cdot)$ is the distribution over discretized bids of the algorithm at that iteration. We can then apply Theorem \ref{thm:continuous-lipschitz-known} to get the following guarantee:
\begin{corollary} The $\winexp$ algorithm run on a uniform $\epsilon$-grid with $\epsilon = \frac{\delta}{2nLT}$, step size $\sqrt{\frac{\log(1/\epsilon)}{4T}}$ and unbiased estimates given by Equation~\eqref{eqn:batch-unbiased} or Equation \eqref{eqn:batch-unbiased-2}, when applied to the sponsored search auction setting with quality scores drawn independently from distributions with $L$-Lipschitz CDFs, achieves regret at most: %
$4\sqrt{T\log\left(\frac{2nLT}{\delta}\right)}+1$.

\end{corollary}

%% file: switching-poa.tex
We show below that our results can be extended to capture the case where, instead of having just one optimal bid $b^*$, there is a sequence of $C \geq 1$ \emph{switches} in the optimal bids. Using the results presented in \cite{GLL12} and adapting them for our setting we get the following corollary (with proof in Appendix \ref{appendix:switch}).
\begin{corollary}\label{cor:switch} 
Let $C \geq 0$ be the number of times that the optimal bid $b^* \in \B$ switches in a horizon of $T$ rounds. Then, using Algorithm $2$ in \cite{GLL12} with $\mathcal{A} \equiv \winexp$ and any $\alpha \in (0,1)$ we can achieve expected \emph{switching regret} at most: $O \left( \sqrt{\left(C+1 \right)^2 \left( 2 + \frac{1}{\alpha} \right) 2d |O| T \log \left(\max \left\{ LT, \frac{1}{\Delta^o} \right\} \right)}\right)$ 
\end{corollary}
\noindent
This result has implications on the price of anarchy (PoA) of auctions. In the case of sponsored search where bidders' valuations are changing over time adversarially but non-adaptively, our result shows that if the valuation does not change more than $C$ times, we can compete with any bid that is a function of the value of the bidder at each iteration, with regret rate given by the latter theorem. Therefore, by standard PoA arguments \cite{LST16}, this would imply convergence to an approximately efficient outcome at a faster rate than bandit regret rates.

%% file: outcome-feedback-graph.tex
We now extend Section \ref{sec:continuous}, by assuming that there is a directed feedback graph $G=(O, E)$ over the outcomes. When outcome $o_t$ is chosen, the bidder observes not only the outcome specific reward function $r_t(\cdot, o_t)$, for that outcome, but also for any outcome $o$ in the out-neighborhood of $o_t$ in the feedback graph, which we denote with $N^{out}(o_t)$. Correspondingly, we denote with $N^{in}(o)$ the incoming neighborhood of $o$ in $G$. Both neighborhoods include self-loops. Let $G_{\epsilon}=(O_{\epsilon}, E_{\epsilon})$ be the sub-graph of $G$ that contains only outcomes for which $\Pr_t[o_t]\geq \epsilon$ and subsequently, let $N_{\epsilon}^{in}$ and $N_{\epsilon}^{out}$ be the in and out neighborhoods of this sub-graph. 

Based on this feedback graph we construct a $\winexp$ algorithm with step-size $\eta= \sqrt{\frac{\log(|B|)}{8T\alpha \ln\left(\frac{16|O|^2 T}{\alpha}\right)}}$, utility estimate $\tilde{u}_t(b) = \1\{o_t\in O_{\epsilon}\} \sum_{o \in N_{\epsilon}^{out}(o_t)}
\frac{(r_t(b, o) - 1) \Pr_t[o|b]}{\sum_{o'\in N_{\epsilon}^{in}(o)}\Pr_t[o']}$ and feedback structure as described in the previous paragraph. 
With these changes we can show that the regret grows as a function of the \emph{independence number of the feedback graph}, denoted with $\alpha$, rather than the \emph{number of outcomes}. The full Algorithm %
can be found in Appendix \ref{appendix:algos}.

\begin{theorem}[Regret of $\winexpG$]\label{thm:feedback-graph}
The regret of the $\winexpG$ algorithm with step size $\eta = \sqrt{\frac{\log(|B|)}{8T\alpha \ln\left(\frac{16|O|^2 T}{\alpha}\right)}}$ is bounded by: $R(T) \leq 2 \sqrt{8\alpha T \log(|B|) \ln\left(\frac{16|O|^2 T}{\alpha}\right)} + 1$.
\end{theorem}

\noindent
In the case of learning in auctions, the feedback graph over outcomes can encode the possibility that winning an item can help you uncover your value for other items. For instance, in a combinatorial auction for $m$ items, the reader should think of each node in the feedback graph as a bundle of items. Then the graph encodes the fact that winning bundle $o$ can teach you the value for all bundles $o'\in N^{out}(o)$. If the feedback graph has small dependence number then a much better regret is achieved than the dependence on $\sqrt{2^m}$, that would have been derived by our outcome-based feedback results of prior sections, if we treated each bundle of items separately as an outcome. 

%% file: experiments.tex
\label{sec:experiments}

In this section, we present our results from our comparative analysis between EXP3 and WIN-EXP on a simulated sponsored search system that we built and which is a close proxy of the actual sponsored search algorithms deployed in the industry. We implemented\footnote{Our code is publicly available on \href{https://github.com/zfengharvard/bandit-sponsored-search}{github}.} the weighted GSP auction as described in definition \ref{weighted-gsp}. The auctioneer draws i.i.d rank scores that are bidder and timestep specific; as is the case throughout our paper, here we have assumed a stochastic auctioneer with respect to the rank scores. After bidding, the bidder will always be able to observe the allocation function. Now, if the bidder gets allocated to a slot and she gets clicked, then, she is able observe the \emph{value} and the payment curve. Values are assumed to lie in $[0,1]$ and they are obliviously adversarial. Finally, the bidders choose bids from some $\epsilon$-discretized grid of $[0,1]$ (in all experiments, apart from the ones comparing the regrets for different discretizations, we use $\epsilon=0.01$) and update the probabilities of choosing each discrete bid according to EXP3 or $\winexp$. Regret is measured with respect to the best fixed discretized bid in hindsight.

We distinguish three cases of the bidding behavior of the rest of the bidders (apart from our learner): i) all of them are \emph{stochastic} adversaries drawing bids at random from some distribution, ii) there is a subset of them that are bidding \emph{adaptively}, by using an EXP3 online learning algorithm and iii) there is a subset of them that are bidding \emph{adaptively} but using a WINEXP online learning algorithm (self play). Validating our theoretical claims, in all three cases, $\winexp$ outperforms EXP3 in terms of regret. 
We generate the event of whether a bidder gets clicked or not as follows: we draw a timestep specific threshold value in $[0,1]$ and the learner gets a click in case the CTR of the slot she got allocated (if any) is greater than this threshold value. Note here that the choice of a timestep specific threshold imposes \emph{monotonicity}, i.e. if the learner did not get a click when allocated to a slot with CTR $x_t(b)$, she should not be able to get a click from slots with lower CTRs. %
We ran simulations with $3$ different distributions of generating CTRs, so as to understand what is the effect of different levels of click-through-rates on the variance of our regret: i) $x_t(b) \sim U[0.1, 1]$, ii) $x_t(b) \sim U[0.3,1]$ and iii) $x_t(b) \sim U[0.5,1]$. Finally, we address robustness of our results to errors in CTR estimation. For this, we add random noise to the CTRs of each slot and we report to the learners the allocation and payment functions that correspond to the erroneous CTRs. The noise was generated according to a normal distribution $\mathcal{N}(0, \frac{1}{m})$, where $m$ could be viewed as the number of training samples on which a machine learning algorithm was ran in order to output the CTR estimate ($m = 100, 1000, 10000$).  

For each of the following simulations, there are $N = 20$ bidders, $k = 3$ slots and we ran the experiment for each round for a total of $30$ times. For the simulations that correspond to adaptive adversaries we used $a = 4$ adversaries. Our results for the \emph{cumulative regret} are presented below. We measured ex-post regret with respect to the realized thresholds that determine whether a bidder gets clicked or not. Note that the solid plots correspond to the emprical mean of the regret, whereas the opaque bands correspond to the $10$-th and $90$-th percentile. 

\paragraph{Different discretizations.}
In Figure~\ref{fig:diff_eps} we present the comparative analysis of the estimated average regret of $\winexp$ vs EXP3 for different discretizations, $\eps$, of the bidding space when the learner faces adversaries that are stochastic, adaptive using EXP3 and adaptive using WINEXP. As it was expected from the theoretical analysis, the regret of $\winexp$, as the disretized space ($|B|$) increases exponentially, remains almost unchanged compared to the regret of EXP3. In summary, \emph{finer discretization of the bid space helps our $\winexp$ algorithm's performance, but hurts the performance of EXP3.}

\begin{figure}[htpb]
\centering
\begin{subfigure}{0.33\textwidth}
\centering
    \includegraphics[width=0.9\textwidth]{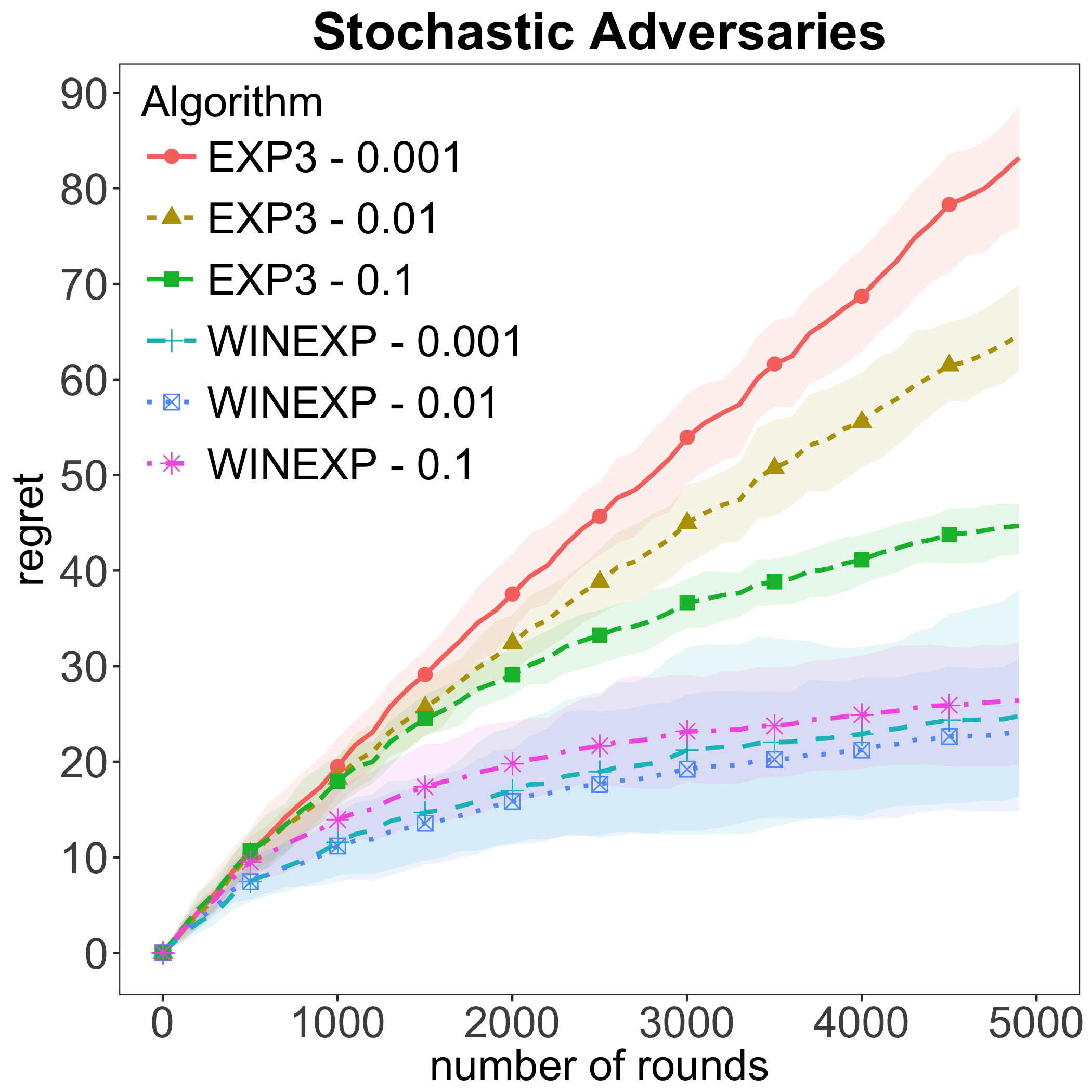}
    \label{fig:discr_ob}
\end{subfigure}%
\begin{subfigure}{0.33\textwidth}
\centering
    \includegraphics[width=0.9\textwidth]{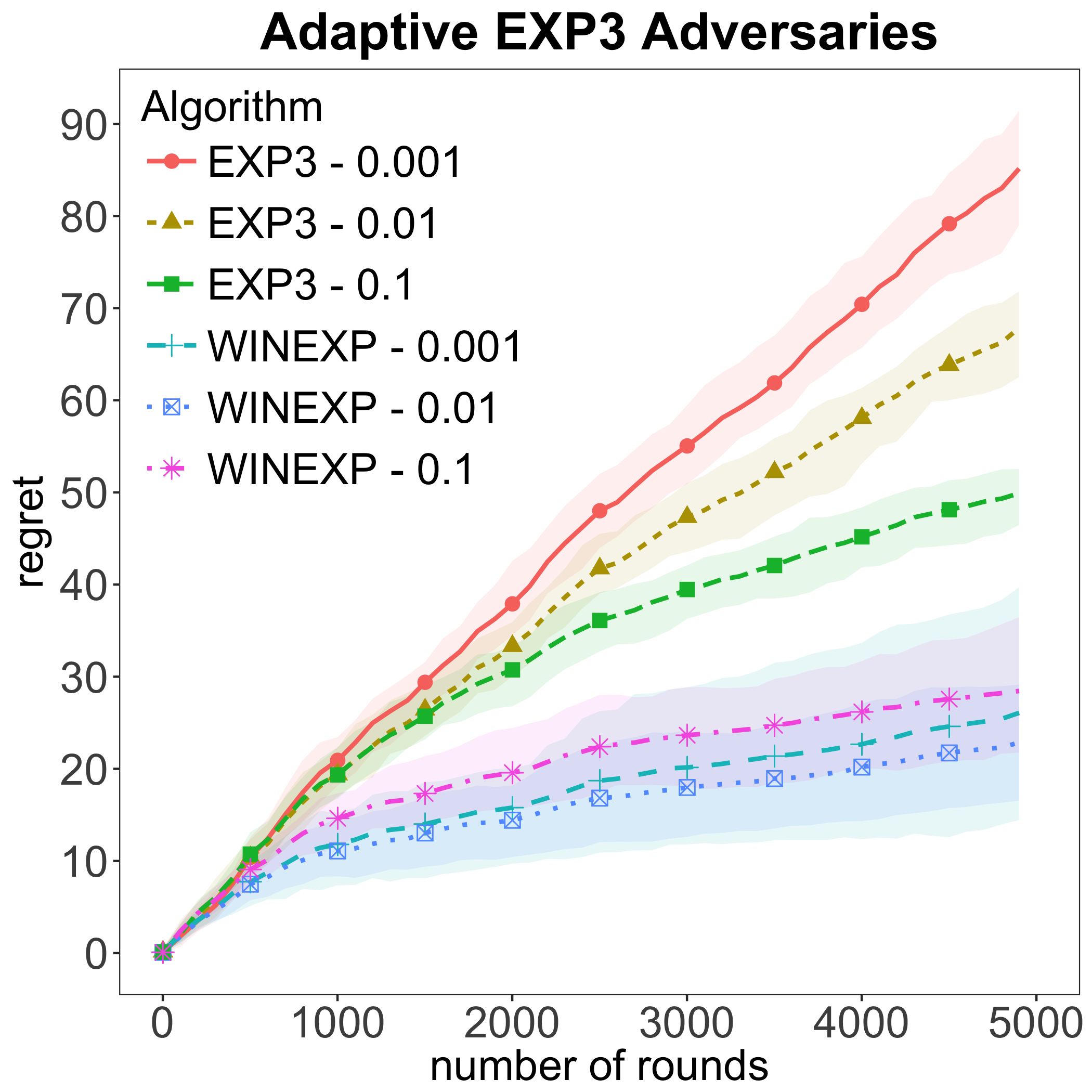}
    \label{fig:discr_sp}
\end{subfigure}%
\begin{subfigure}{0.33\textwidth}
\centering
    \includegraphics[width=0.9\textwidth]{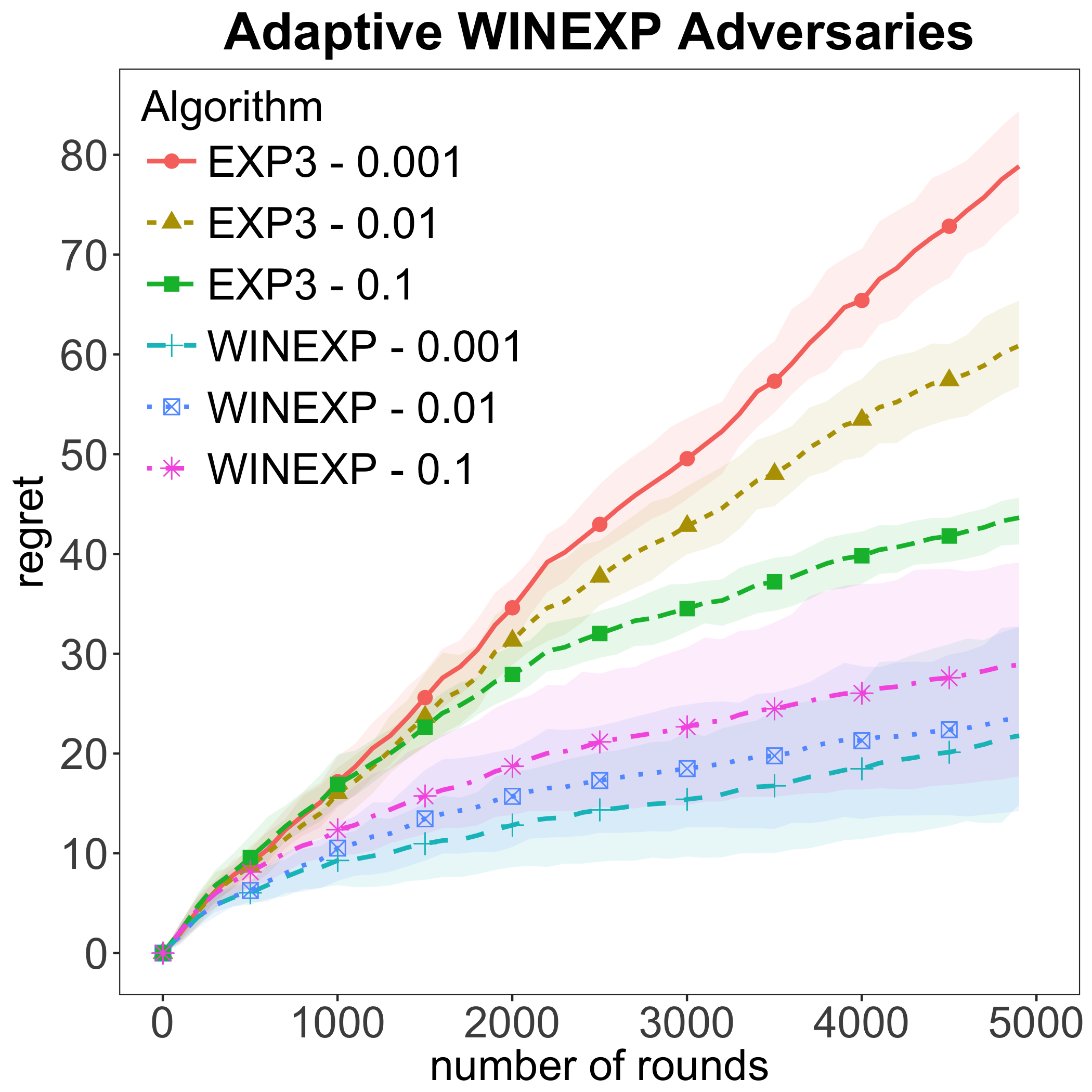}
    \label{fig:discr_winexp_adv}
\end{subfigure}
\caption{Regret of $\winexp$ vs EXP3 for different discretizations $\epsilon$ (CTR $\sim U[0.5,1]$).}
\label{fig:diff_eps}
\end{figure}

\paragraph{Different CTR Distributions.}
In Figures~\ref{fig:ob_diff_ctr}, \ref{fig:sp_diff_ctr} and \ref{fig:winexp_adv_diff_ctr} we present the results of the regret performance of $\winexp$ compared to EXP3, when the learner discretizes the bidding space with $\eps = 0.01$ and when she faces stochastic, adaptive adversaries using EXP3 and adaptive adversaries using WINEXP, respectively. For all three cases, the estimated average regret of $\winexp$ is less than the estimated average regret that EXP3 yields. 

\begin{figure}[htpb]
\centering
\begin{subfigure}{0.33\textwidth}
\centering
    \includegraphics[width=0.9\textwidth]{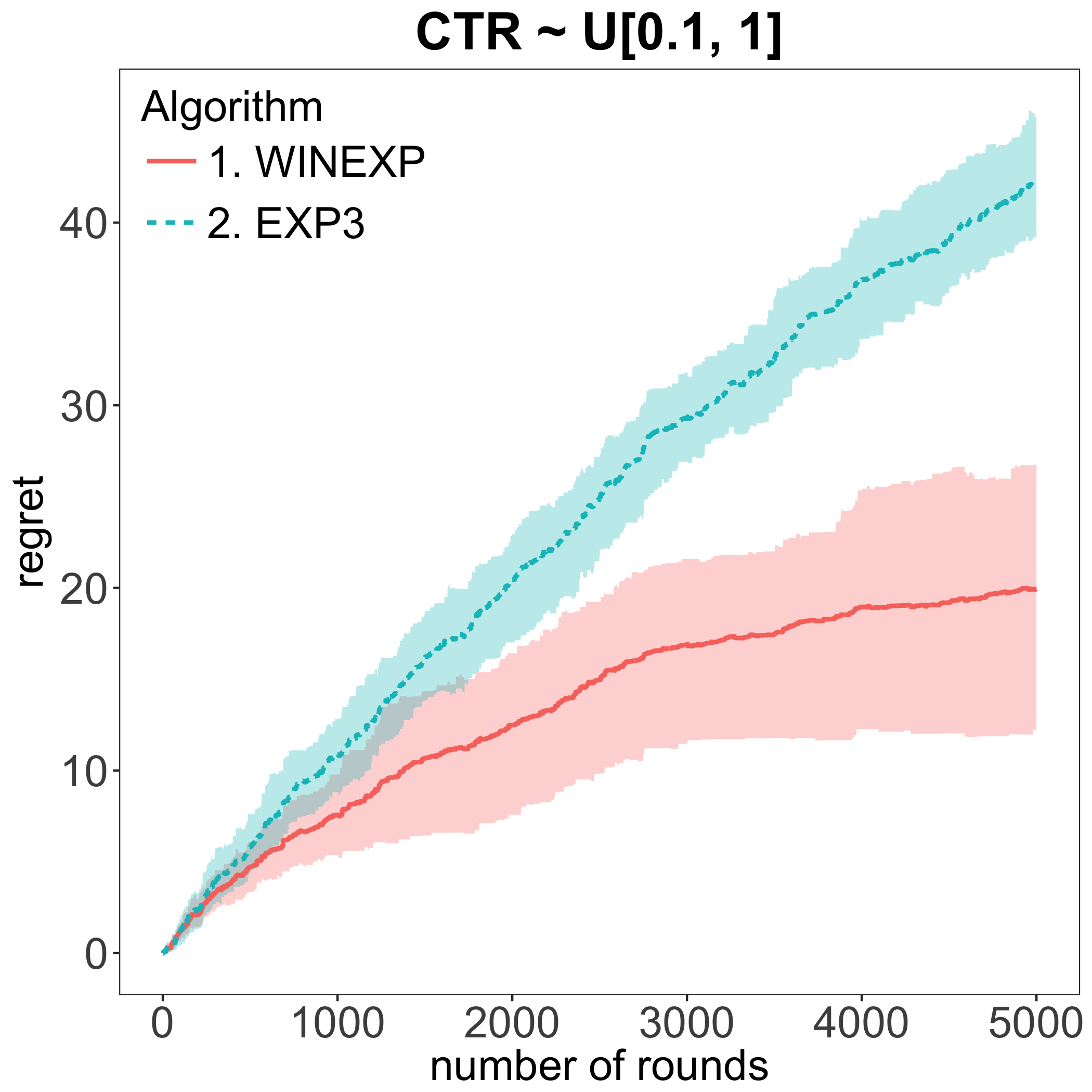}
    \label{fig:ob_cl_0.1}
\end{subfigure}%
\begin{subfigure}{0.33\textwidth}
\centering
    \includegraphics[width=0.9\textwidth]{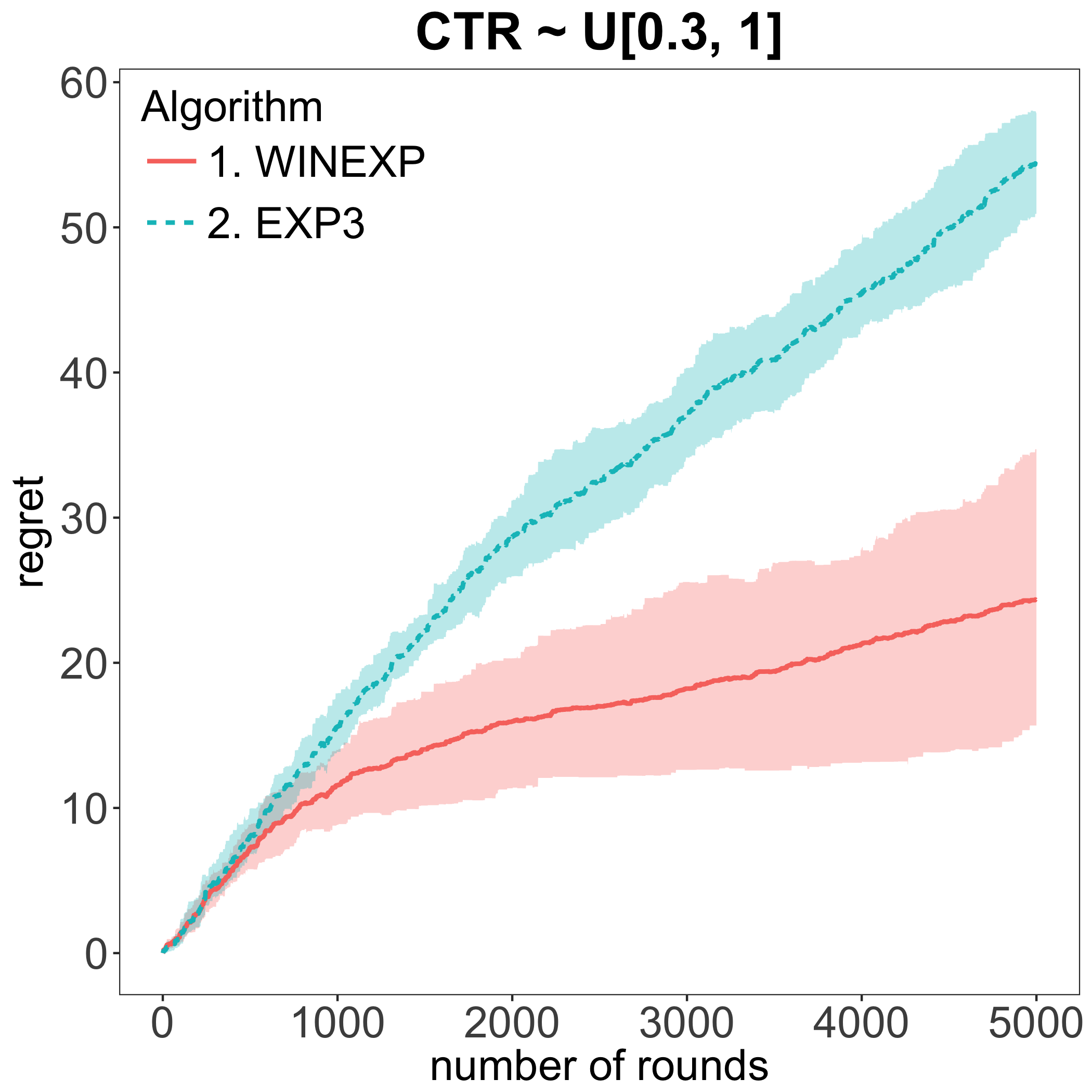}
    \label{fig:ob_cl_0.3}
\end{subfigure}%
\begin{subfigure}{0.33\textwidth}
\centering
    \includegraphics[width=0.9\textwidth]{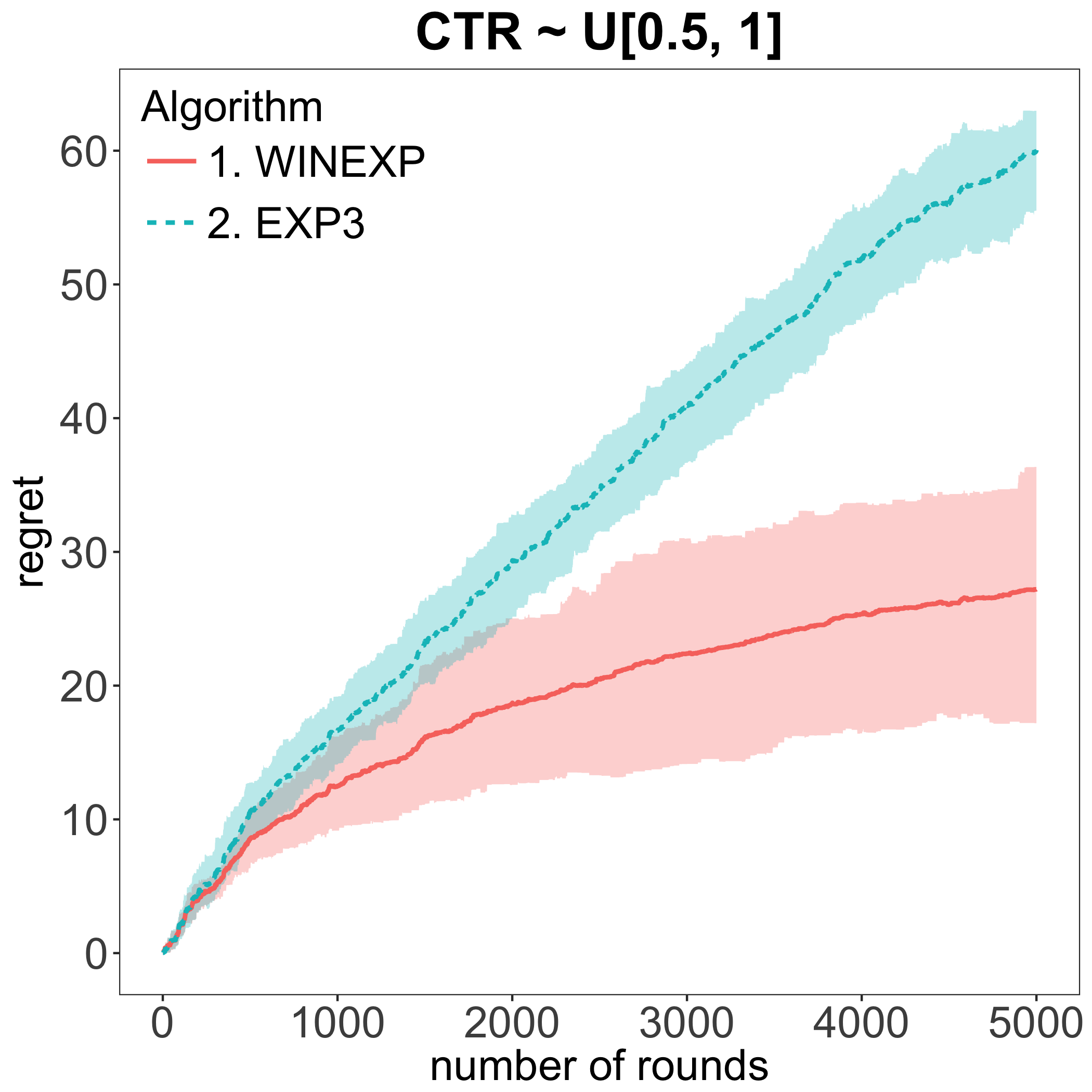}
    \label{fig:ob_cl_std}
\end{subfigure}%
\caption{Regret of $\winexp$ vs EXP3 for different CTR distributions and stochastic adversaries, $\eps = 0.01$.}
\label{fig:ob_diff_ctr}
\end{figure}

\begin{figure}[htpb]
\centering
\begin{subfigure}{0.33\textwidth}
\centering
    \includegraphics[width=0.9\textwidth]{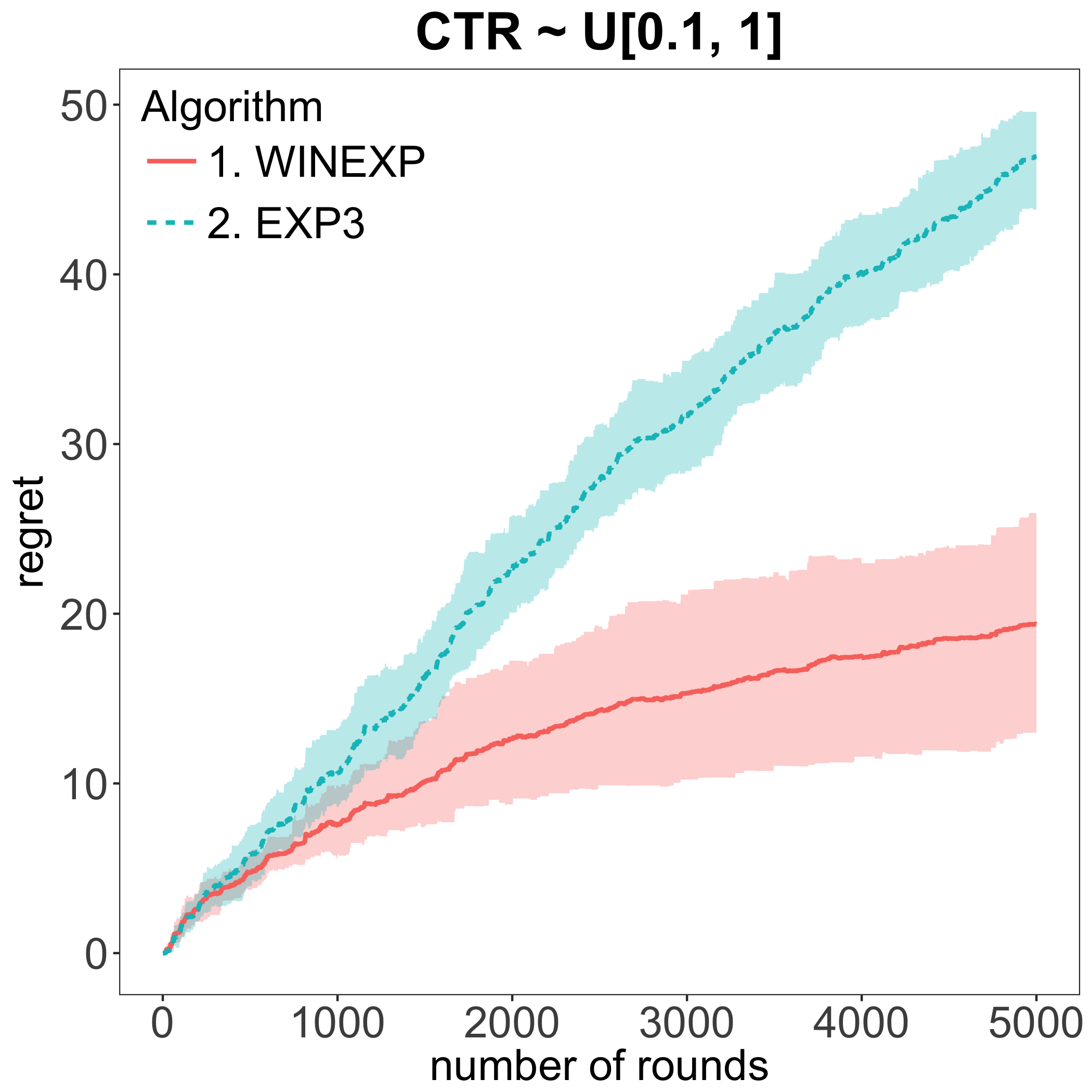}
    \label{fig:sp_cl_0.1}
\end{subfigure}%
\begin{subfigure}{0.33\textwidth}
\centering
    \includegraphics[width=0.9\textwidth]{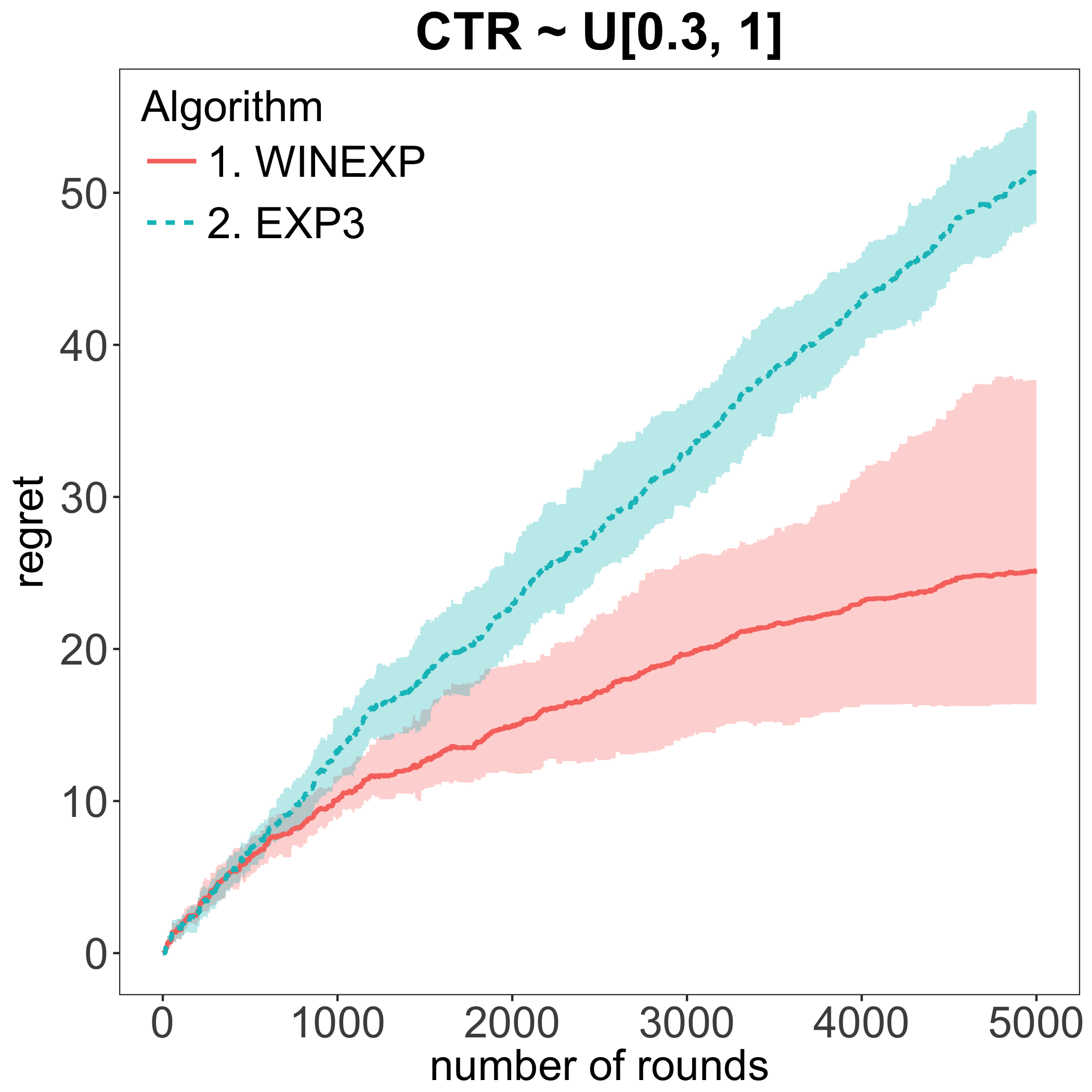}
    \label{fig:sp_cl_0.3}
\end{subfigure}%
\begin{subfigure}{0.33\textwidth}
\centering
    \includegraphics[width=0.9\textwidth]{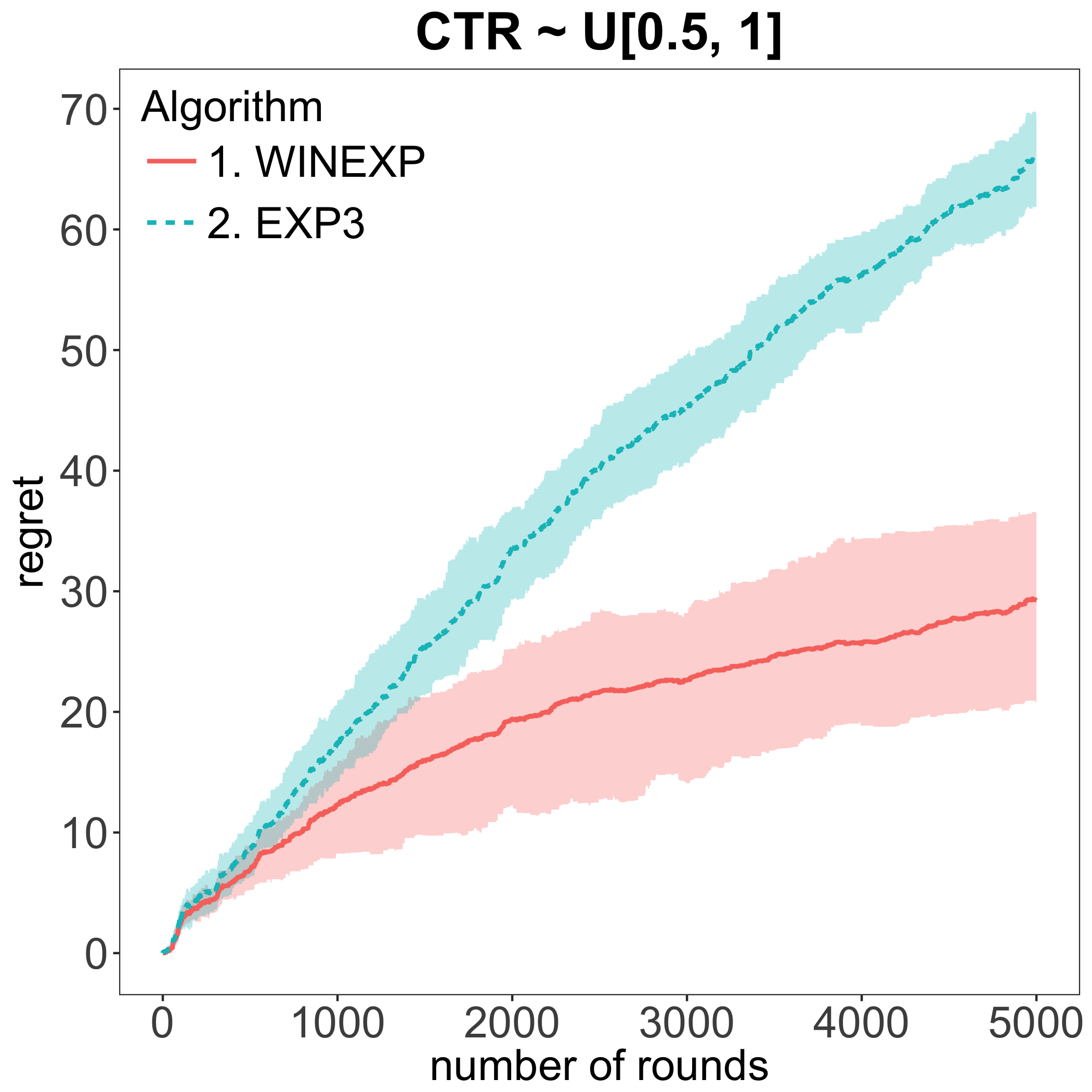}
    \label{fig:sp_cl_std}
\end{subfigure}%
\caption{Regret of $\winexp$ vs EXP3 for different CTR distributions and adaptive EXP3 adversaries, $\eps=0.01$.}
\label{fig:sp_diff_ctr}
\end{figure}

\begin{figure}[htpb]
\centering
\begin{subfigure}{0.33\textwidth}
\centering
    \includegraphics[width=0.9\textwidth]{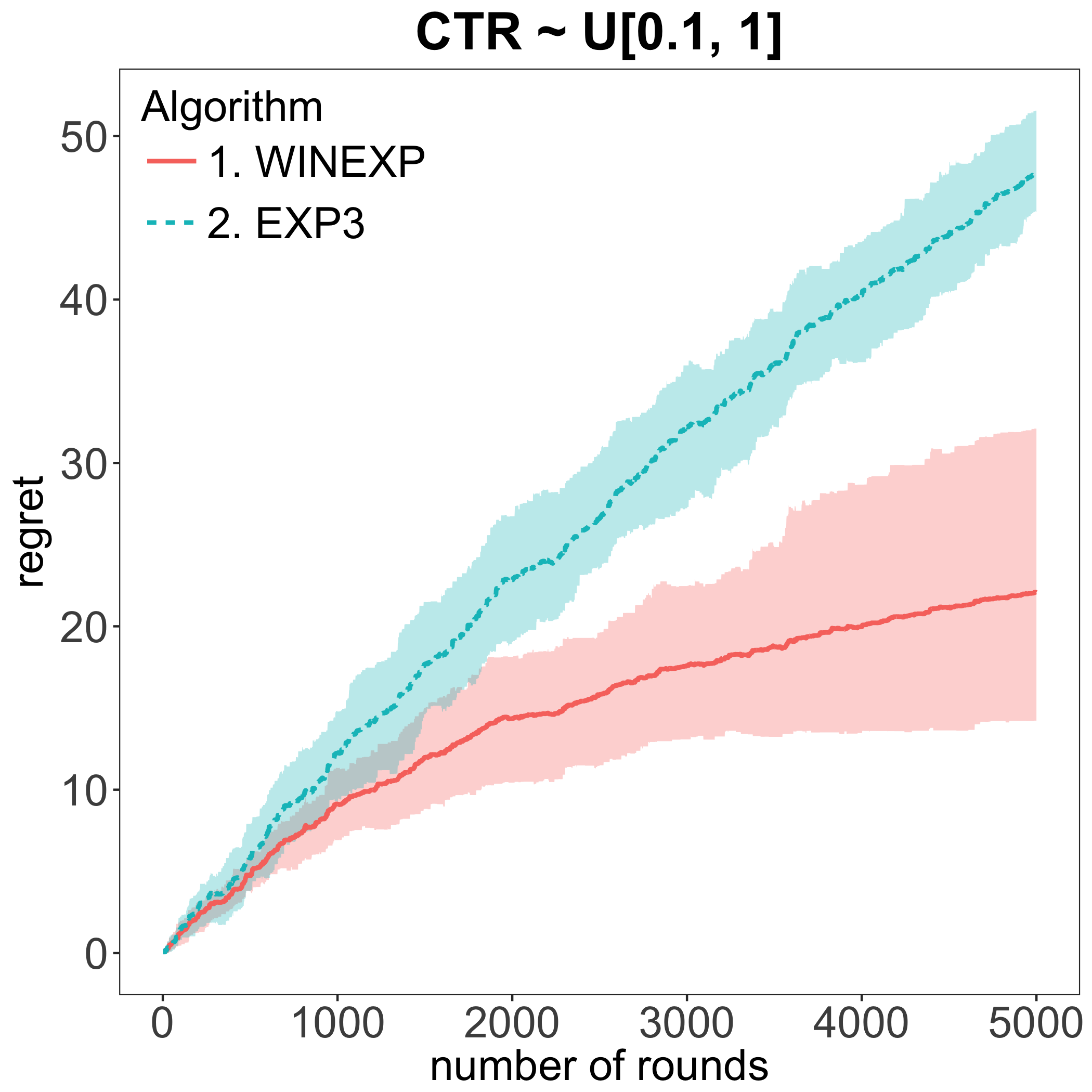}
    \label{fig:winexp_vs_winexp_cl_0.1}
\end{subfigure}%
\begin{subfigure}{0.33\textwidth}
\centering
    \includegraphics[width=0.9\textwidth]{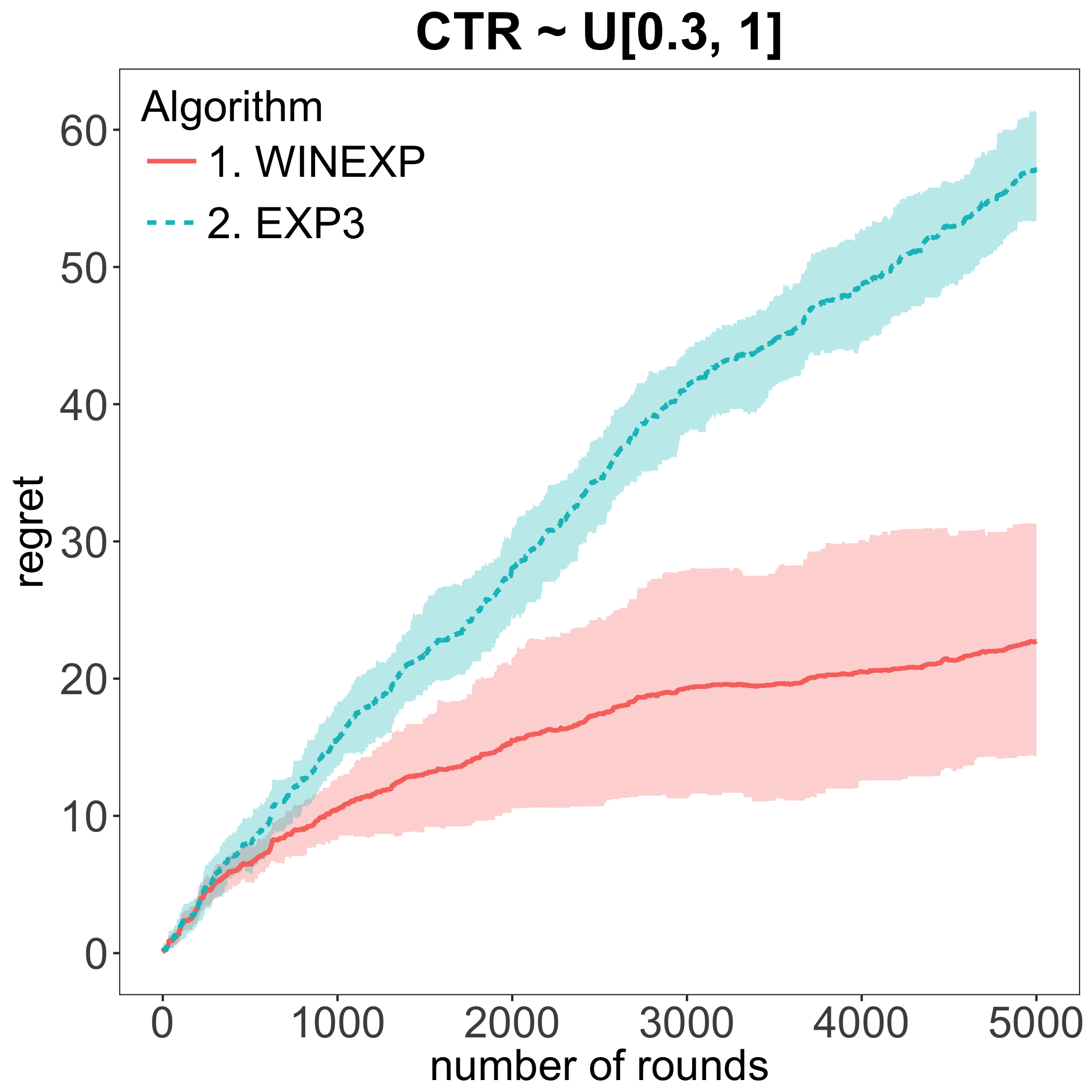}
    \label{fig:winexp_vs_winexp_0.3}
\end{subfigure}%
\begin{subfigure}{0.33\textwidth}
\centering
    \includegraphics[width=0.9\textwidth]{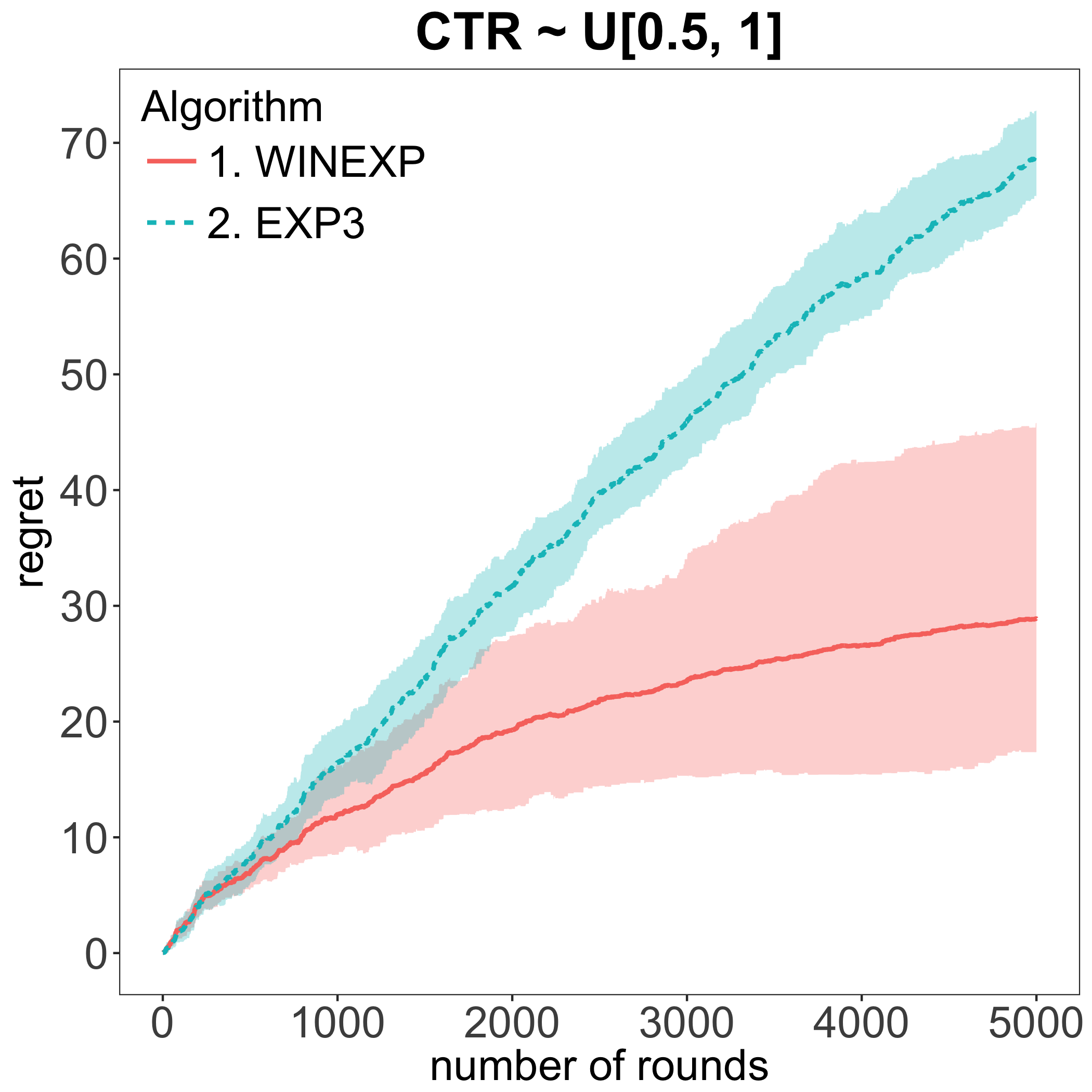}
    \label{fig:winexp_vs_winexp_05}
\end{subfigure}%
\caption{Regret of $\winexp$ vs EXP3 for different CTR distributions and adaptive WINEXP adversaries, $\eps=0.01$.}
\label{fig:winexp_adv_diff_ctr}
\end{figure}

\paragraph{Robustness to Noisy CTR Estimates.} In Figures \ref{fig:oblivious_noise}, \ref{fig:self_play_noise}, \ref{fig:winexp_adv_noise} we empirically tested the robustness of our algorithm to random perturbations of the allocation function that the auctioneer presents to the learner, for perturbations of the form $\mathcal{N}\left( 0, \frac{1}{m}\right)$, where $m$ could be viewed as the number of training examples used from the auctioneer in order to derive an approximation of the allocation curve. When the number of training samples is relatively small ($m = 100$) the empirical mean of WINEXP outperforms EXP3 in terms of regret, i.e., it is more robust to such perturbations. As the number of training samples increases, WINEXP clearly outperforms EXP3. \cpdelete{For very few $m=100$ samples $\winexp$ performs slightly worse. However, when the training samples are more than $m=1000$ $\winexp$ is more robust to such perburbation not just on average, but also at every timestep after some round $T$. }The latter validates one of our claims throughout the paper; namely, that even though the learner might not see the exact allocation curve, but a randomly perturbed proxy,\vsdelete{(as is the case when the CTR curve is given by a bid simulator for a horizon of almost a week),} $\winexp$ still performs better than the EXP3.

\begin{figure}[htpb]
\centering
\begin{subfigure}{0.33\textwidth}
\centering
    \includegraphics[width=0.9\textwidth]{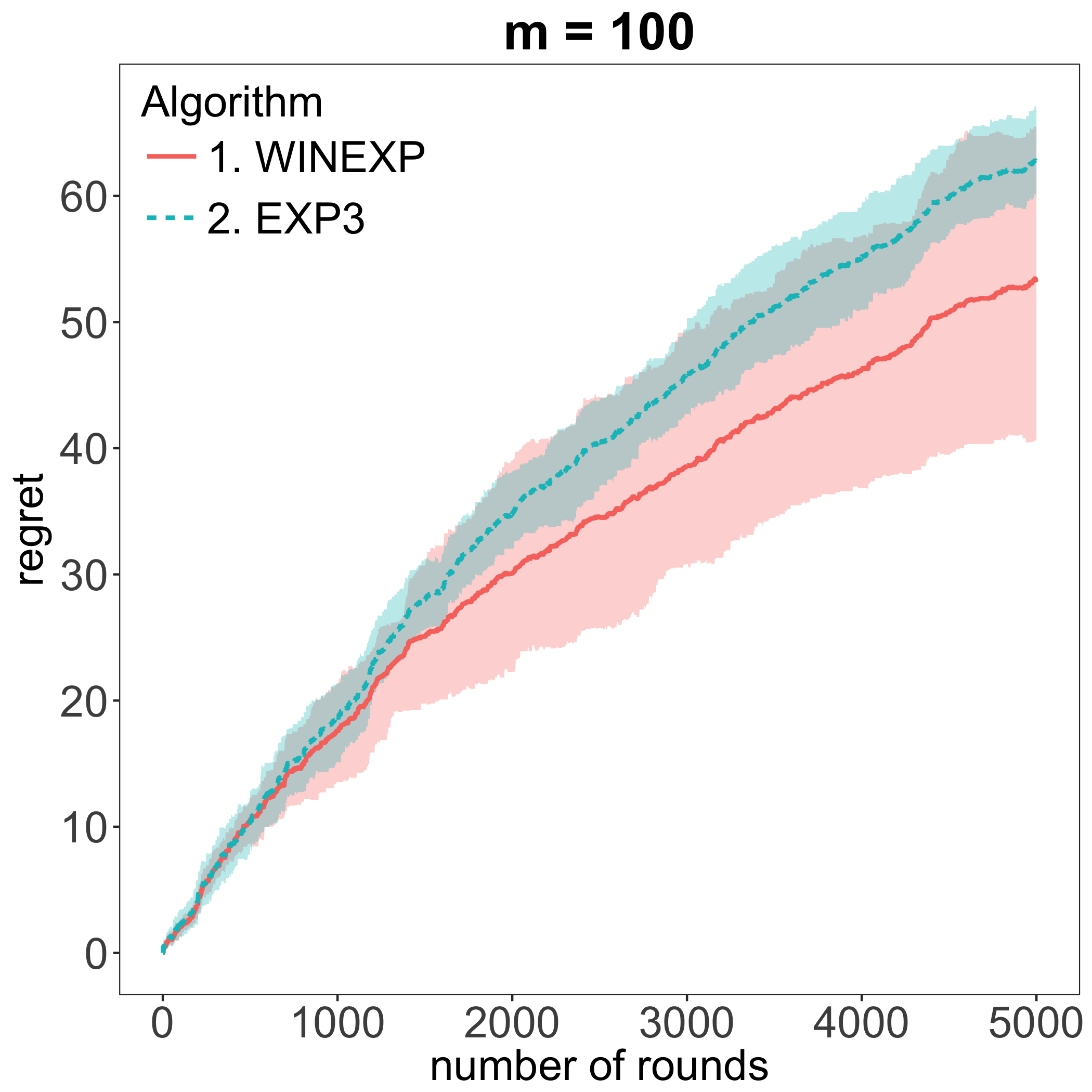}
    \label{fig:ob_no_100}
\end{subfigure}%
\begin{subfigure}{0.33\textwidth}
\centering
    \includegraphics[width=0.9\textwidth]{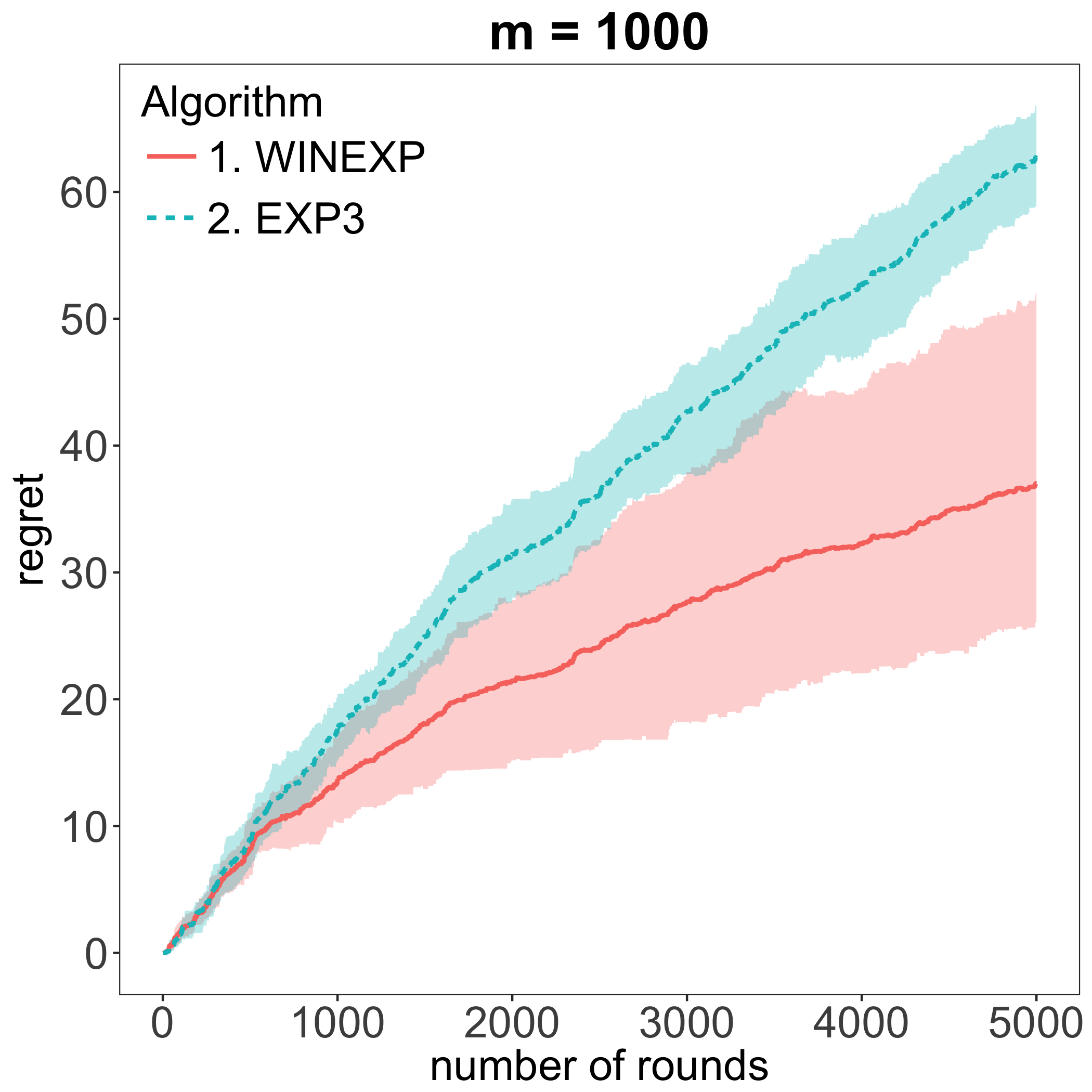}
    \label{fig:ob_no_1000}
\end{subfigure}%
\begin{subfigure}{0.33\textwidth}
\centering
    \includegraphics[width=0.9\textwidth]{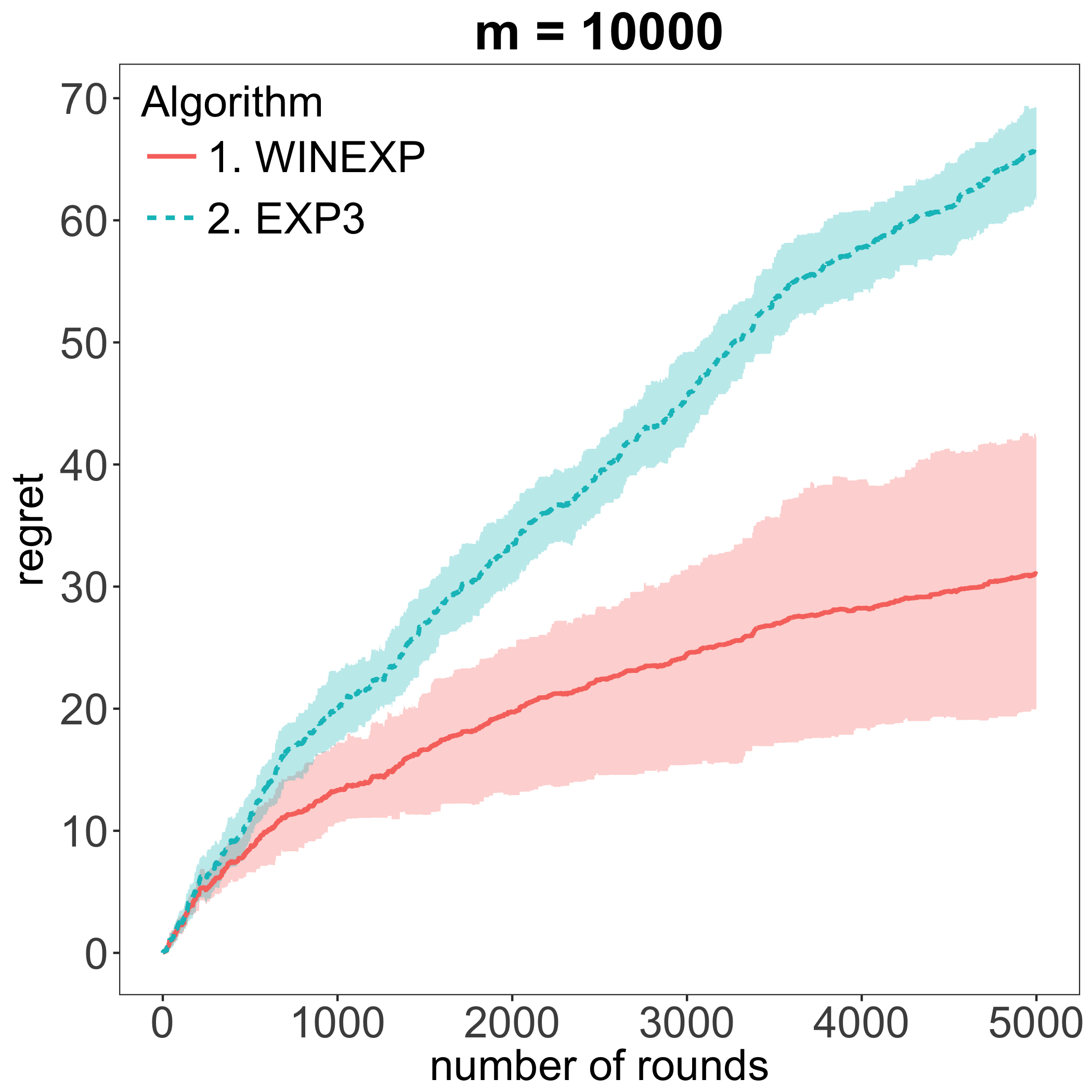}
    \label{fig:ob_no_10000}
\end{subfigure}
\caption{Regret of $\winexp$ vs EXP3 with noise $\sim \mathcal{N}\left(0,\frac{1}{m}\right)$ for stochastic adversaries, $\eps = 0.01$.}
\label{fig:oblivious_noise}
\end{figure}

\begin{figure}[htpb]
\centering
\begin{subfigure}{0.33\textwidth}
\centering
    \includegraphics[width=0.9\textwidth]{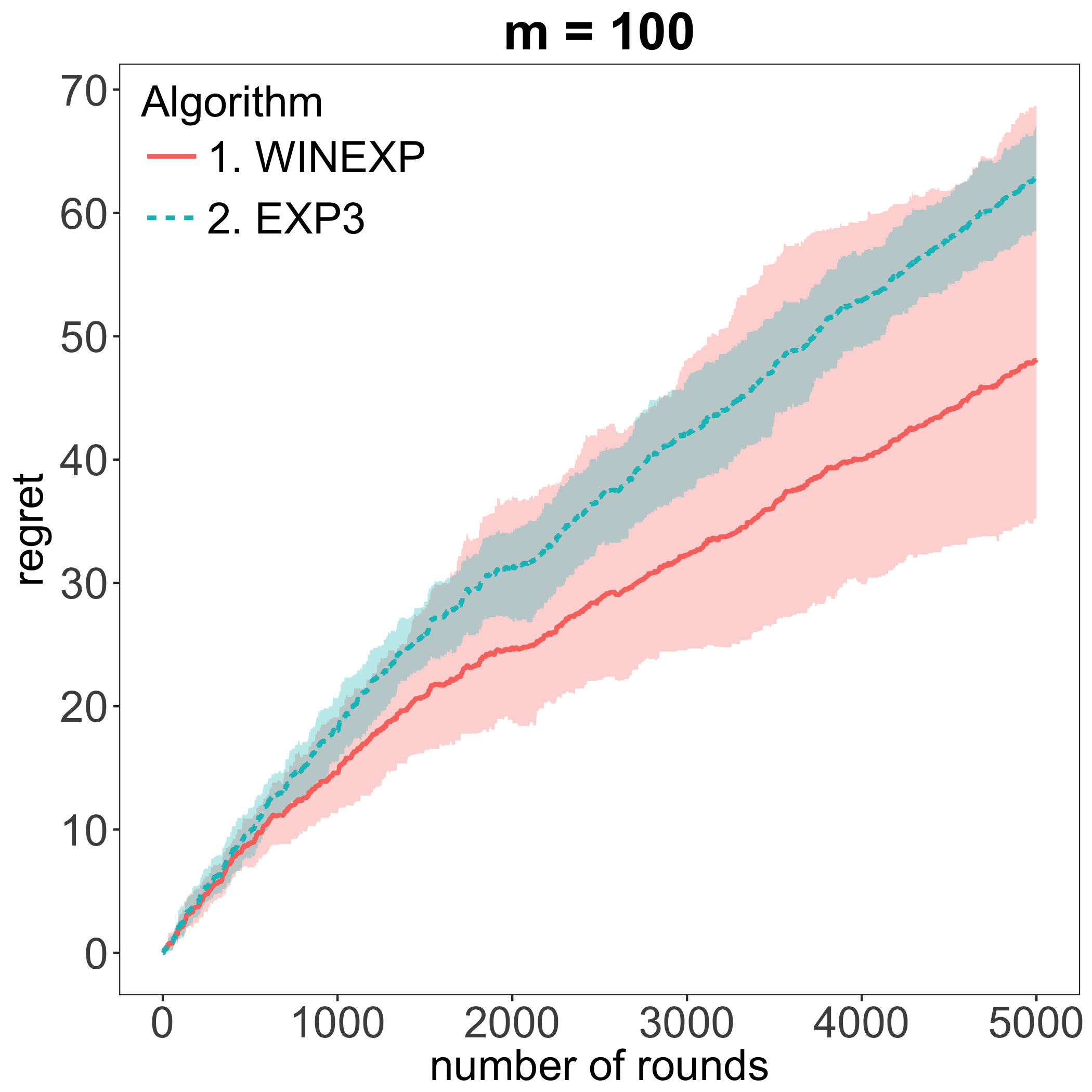}
    \label{fig:sp_no_100}
\end{subfigure}%
\begin{subfigure}{0.33\textwidth}
\centering
    \includegraphics[width=0.9\textwidth]{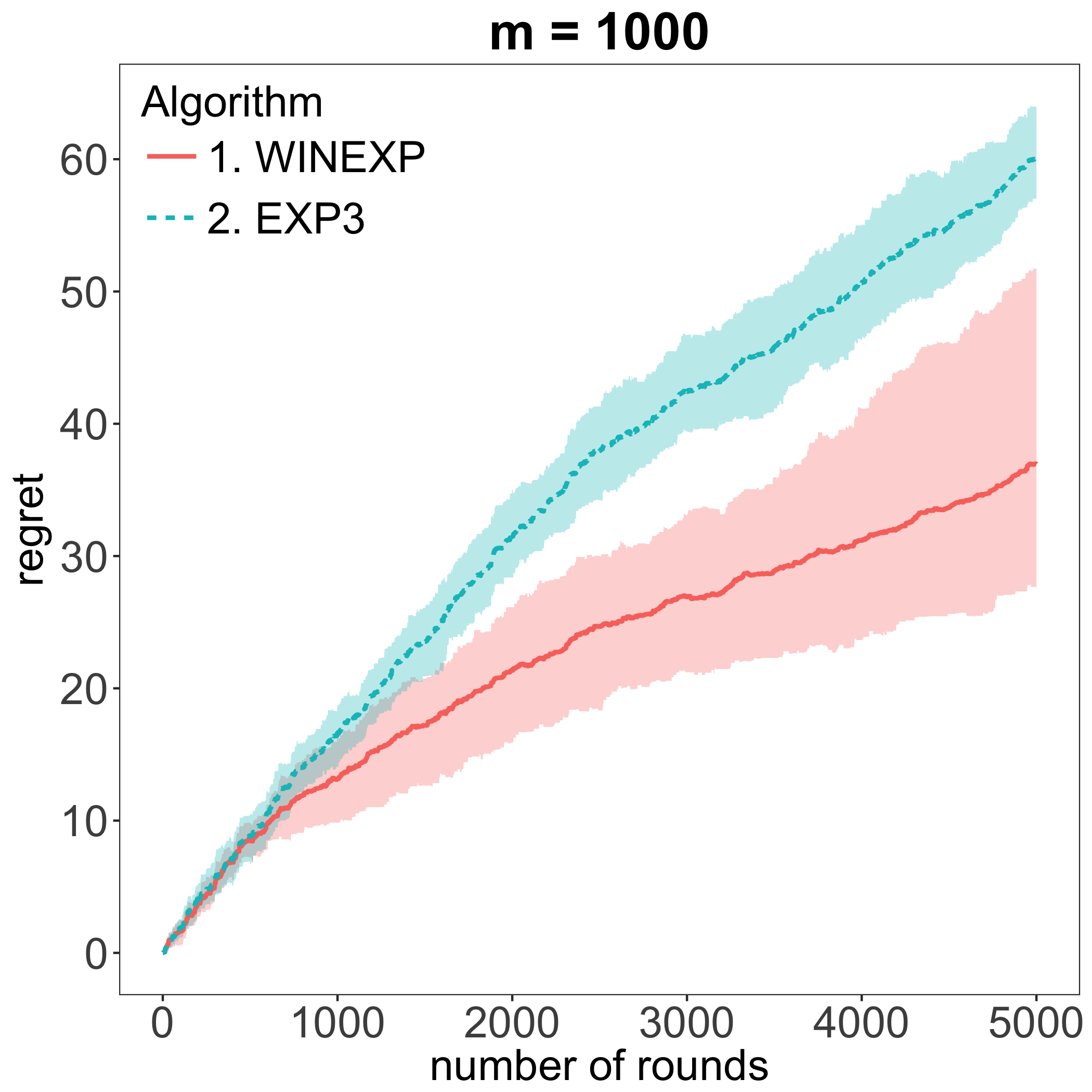}
    \label{fig:sp_no_1000}
\end{subfigure}%
\begin{subfigure}{0.33\textwidth}
\centering
    \includegraphics[width=0.9\textwidth]{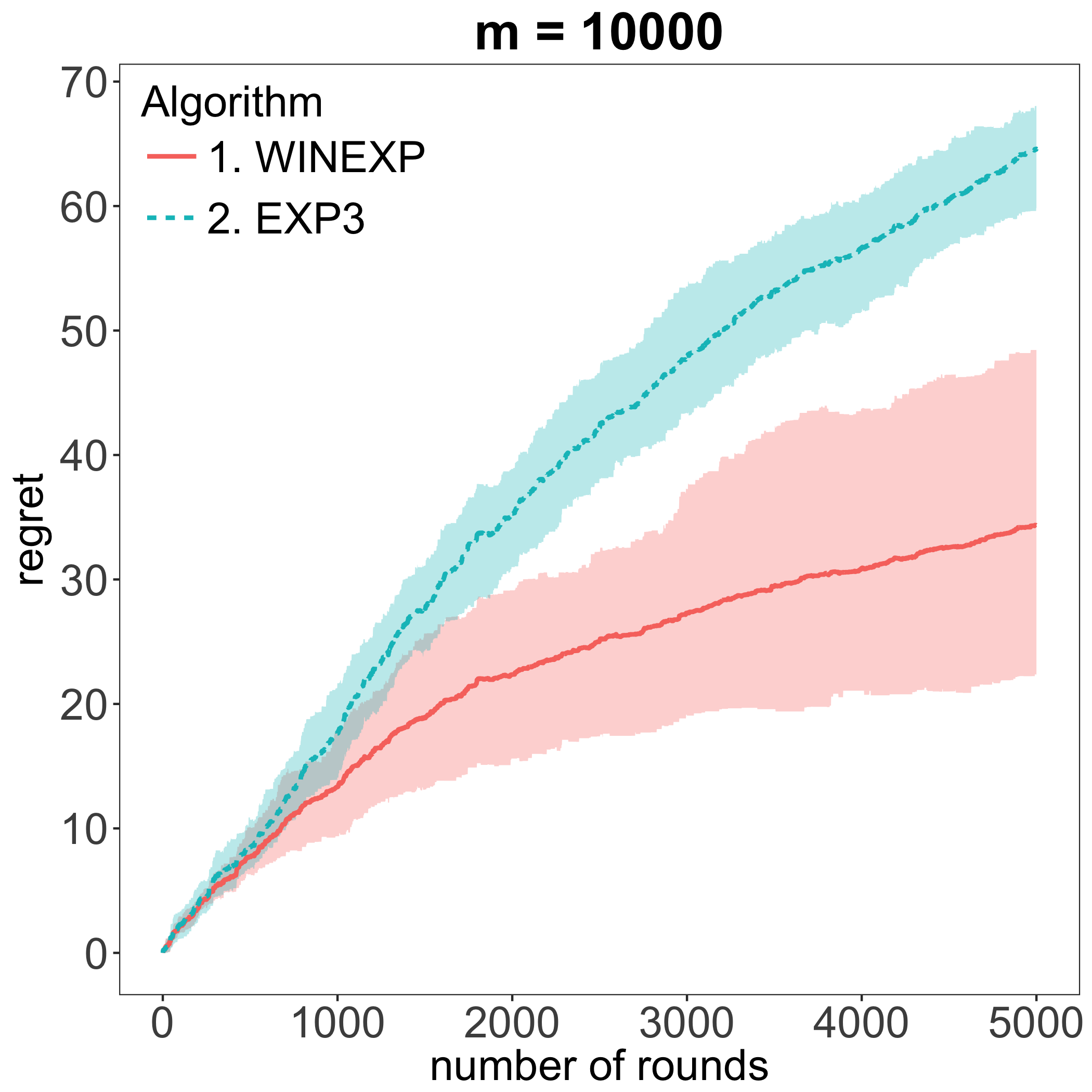}
    \label{fig:sp_no_10000}
\end{subfigure}
\caption{Regret of $\winexp$ vs EXP3 with noise $\sim \mathcal{N}\left(0,\frac{1}{m}\right)$ for adaptive EXP3 adversaries, $\eps=0.01$.}
\label{fig:self_play_noise}
\end{figure}

\begin{figure}[htpb]
\centering
\begin{subfigure}{0.33\textwidth}
\centering
    \includegraphics[width=0.9\textwidth]{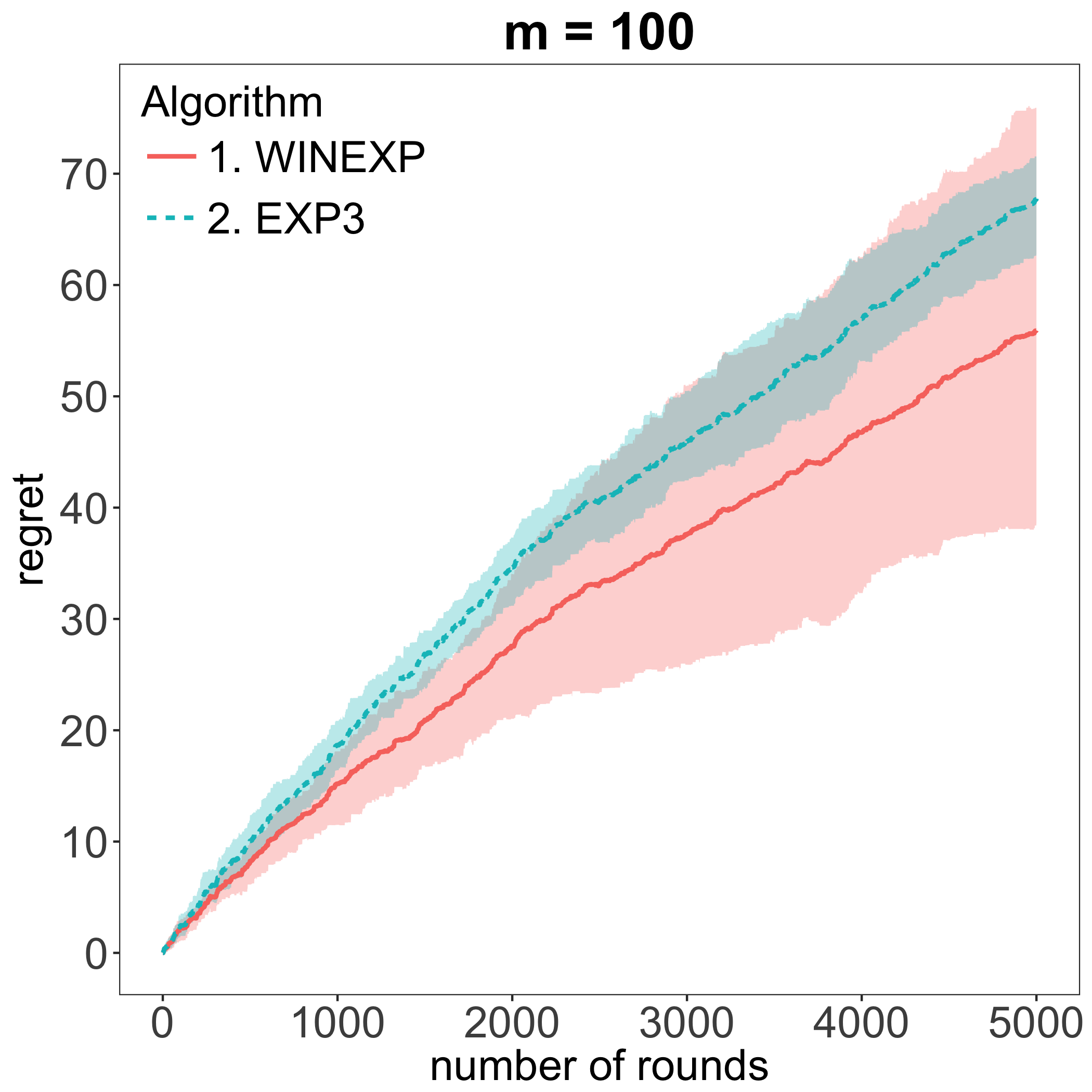}
    \label{fig:winexp_vs_winexp_no_100}
\end{subfigure}%
\begin{subfigure}{0.33\textwidth}
\centering
    \includegraphics[width=0.9\textwidth]{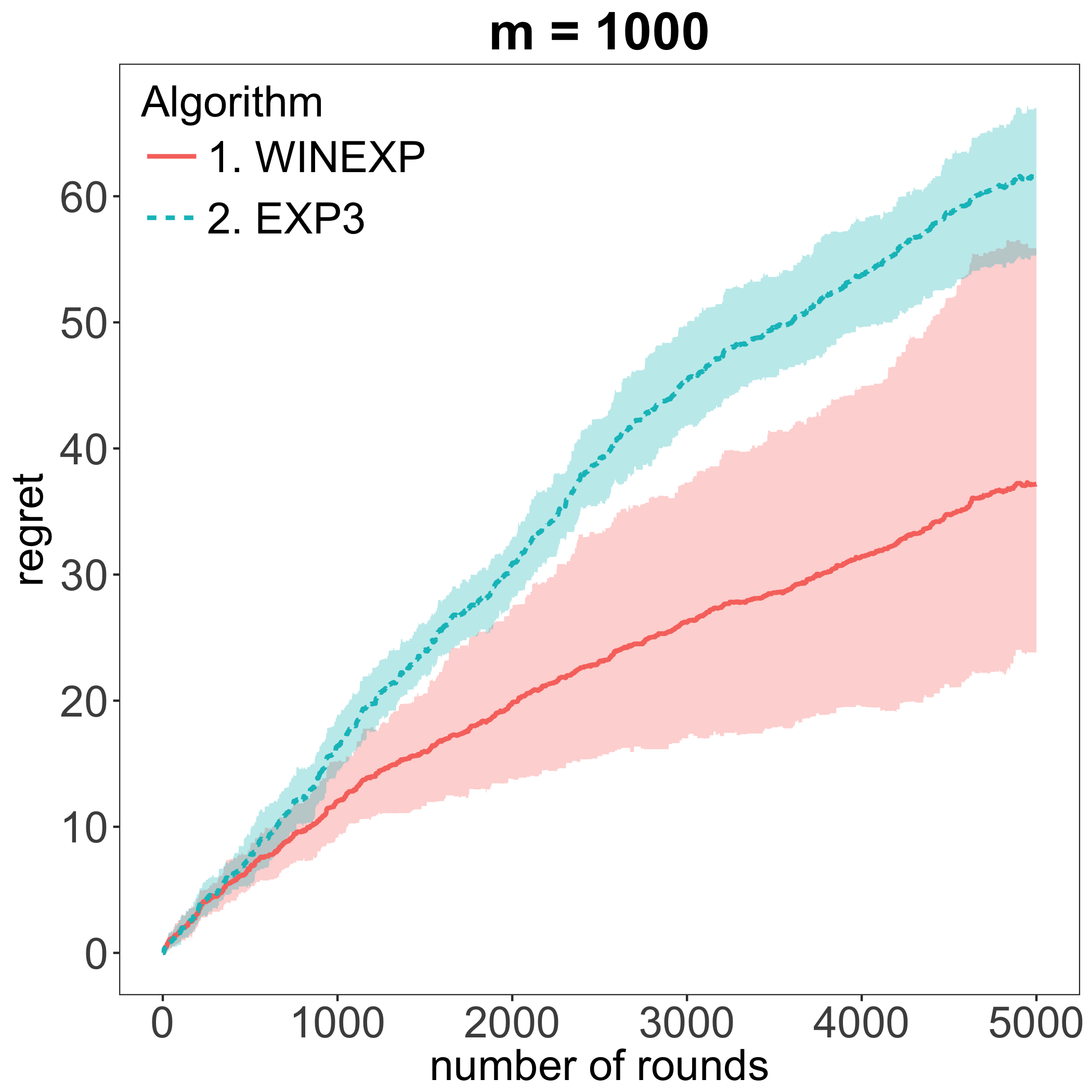}
    \label{fig:winexp_vs_winexp_no_1000}
\end{subfigure}%
\begin{subfigure}{0.33\textwidth}
\centering
    \includegraphics[width=0.9\textwidth]{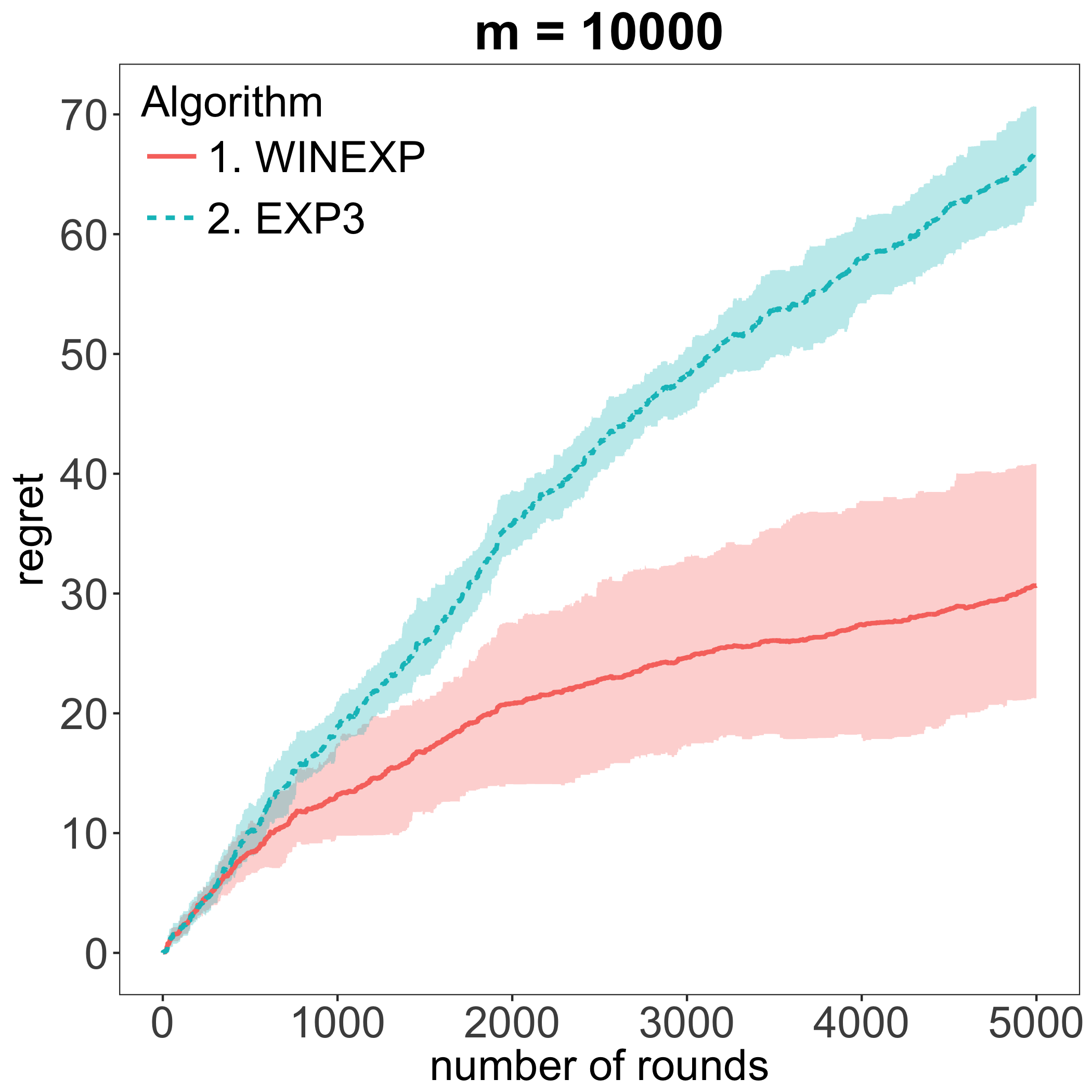}
    \label{fig:winexp_vs_winexp_no_10000}
\end{subfigure}
\caption{Regret of $\winexp$ vs EXP3 with noise $\sim \mathcal{N}\left(0,\frac{1}{m}\right)$ for adaptive WINEXP adversaries, $\eps=0.01$.}
\label{fig:winexp_adv_noise}
\end{figure}

\paragraph{Robustness to fully bandit feedback.} In Figures \ref{fig:regr_uniform}, \ref{fig:regr_normal} we present the results of the regret performance of $\winexp$ compared to EXP3, when the learner receives \emph{fully bandit feedback}, i.e., observes only the CTR and the payment for the currently submitted bid and uses \emph{logistic regression} in order to estimate the allocation curve and \emph{linear regression} in order to estimate the payment curve. For the logistic regression, at each timestep $t$ we used all $t$ pairs of bids and realized CTRs, $(b_t, x_t(b_t))$, giving exponentially more weight to more recent pairs. For the linear regression, at each timestep $t$ we used all $t$ pairs of bids and realized payments, $(b_t, p_t(b_t))$. Figure \ref{fig:regr_uniform} shows that when CTRs come from a uniform distribution, $\winexp$ with fully bandit feedback constructing the estimates with logistic and linear regression not only outperforms EXP3 in terms of regret, but manages to approximate the regret achieved in the case that one receives the true allocation and payment curves for both stochastic and adaptive adversaries. To portray the power of $\winexp$ with fully bandit feedback we also ran a set of experiments where the CTRs and the rank scores come from normal distributions $\calN(0.5, 0.16)$ and even in these cases, $\winexp$ outperforms EXP3 in terms of regret.

\begin{figure}[htpb]
\centering
\begin{subfigure}{0.33\textwidth}
\centering
    \includegraphics[width=0.9\textwidth]{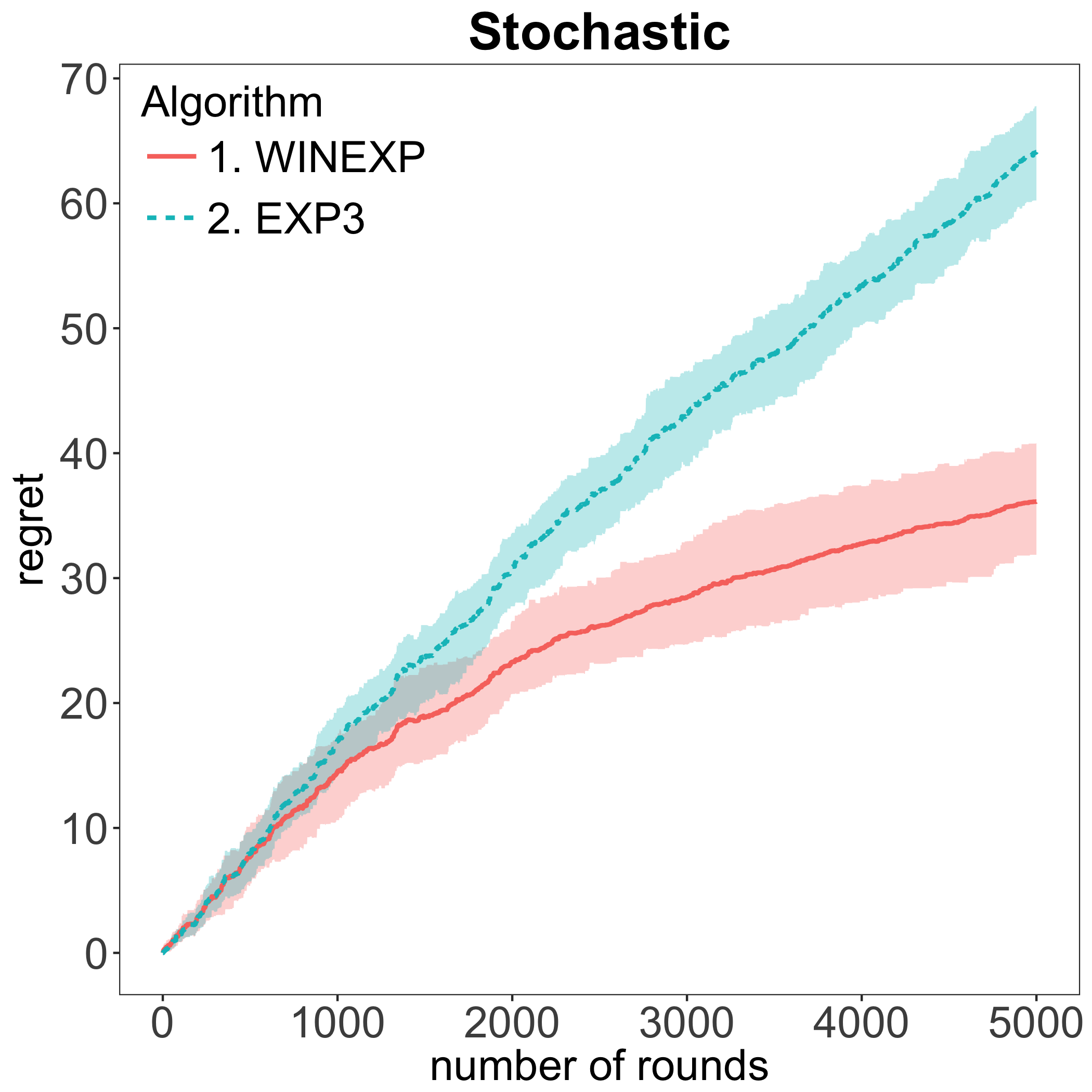}
    \label{fig:ob_regr_uniform}
\end{subfigure}%
\begin{subfigure}{0.33\textwidth}
\centering
    \includegraphics[width=0.9\textwidth]{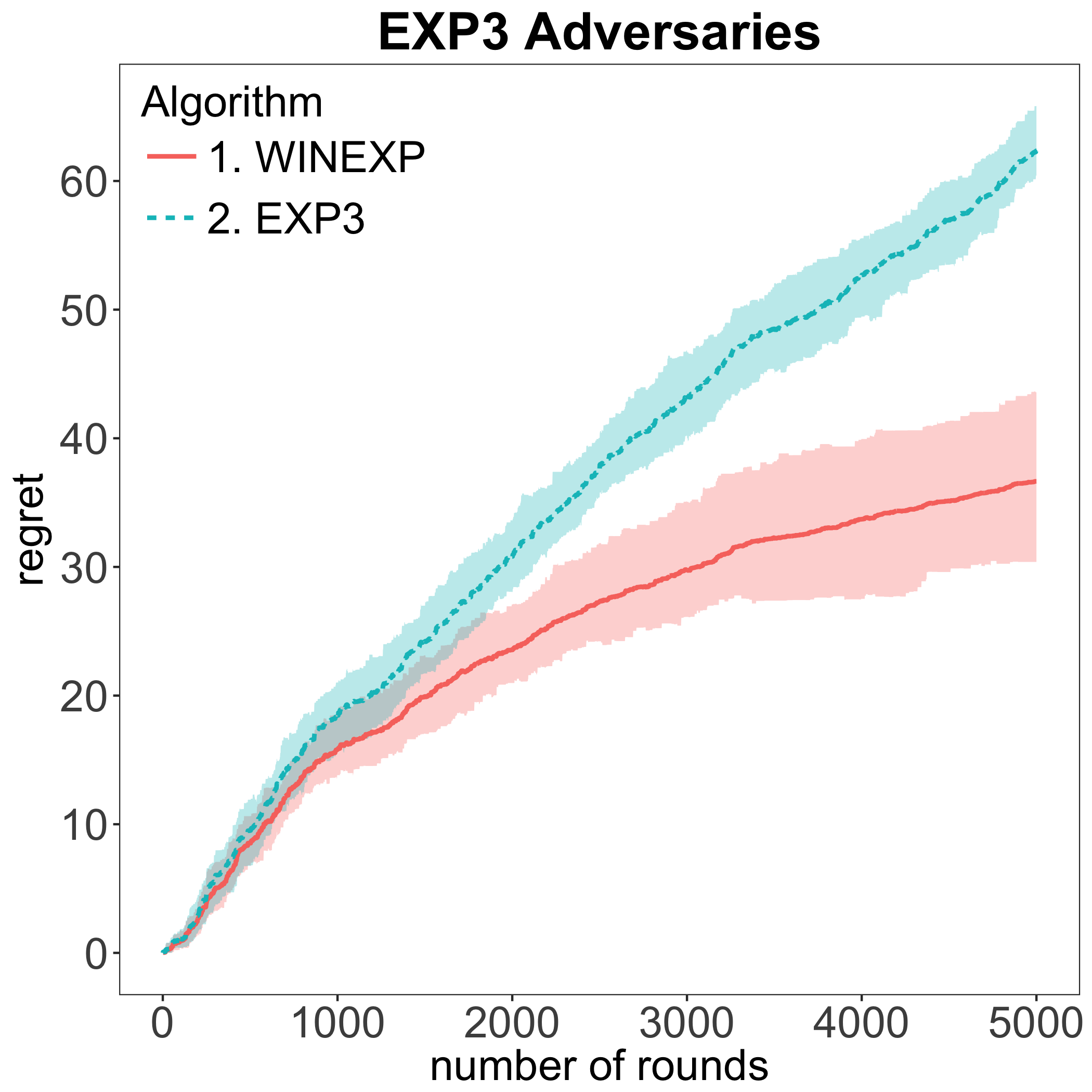}
    \label{fig:exp3_adv_regr_uniform}
\end{subfigure}%
\begin{subfigure}{0.33\textwidth}
\centering
    \includegraphics[width=0.9\textwidth]{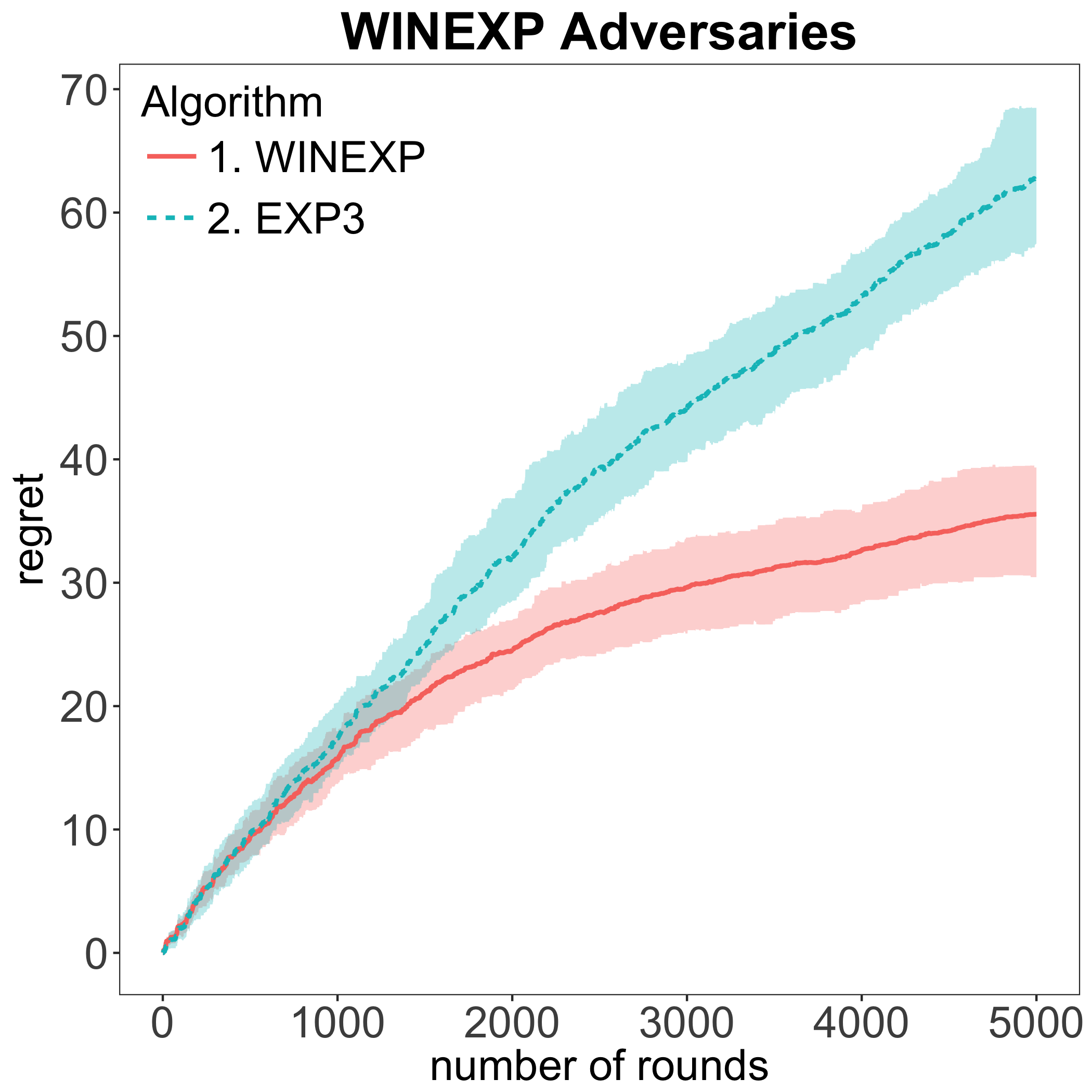}
    \label{fig:winexp_adv_regr_uniform}
\end{subfigure}
\caption{Regret of $\winexp$ vs EXP3 for the fully bandit feedback with CTR $\sim U[0.5,1]$ and $\epsilon = 0.01$.}
\label{fig:regr_uniform}
\end{figure}

\begin{figure}[htpb]
\centering
\begin{subfigure}{0.33\textwidth}
\centering
    \includegraphics[width=0.9\textwidth]{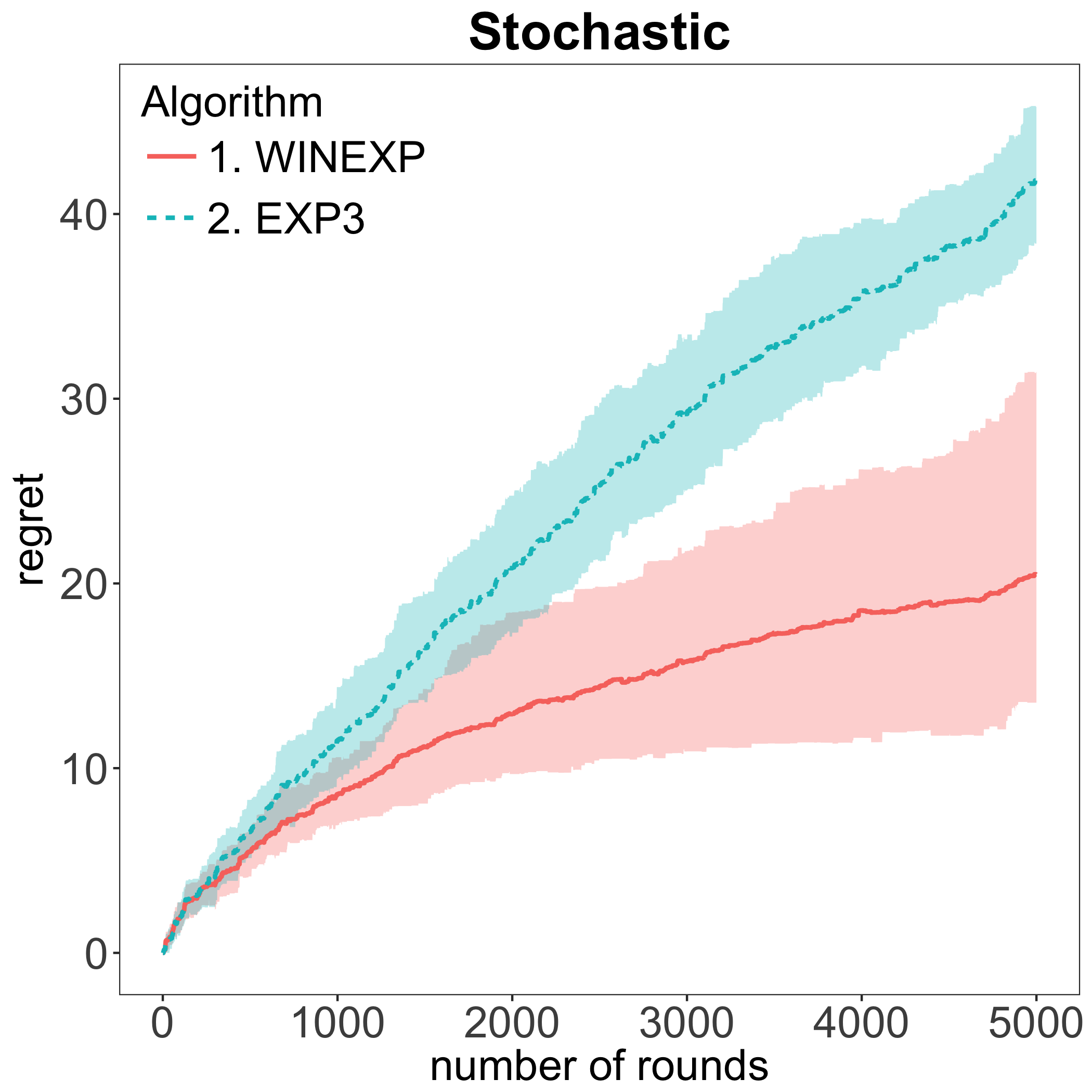}
\end{subfigure}%
\begin{subfigure}{0.33\textwidth}
\centering
    \includegraphics[width=0.9\textwidth]{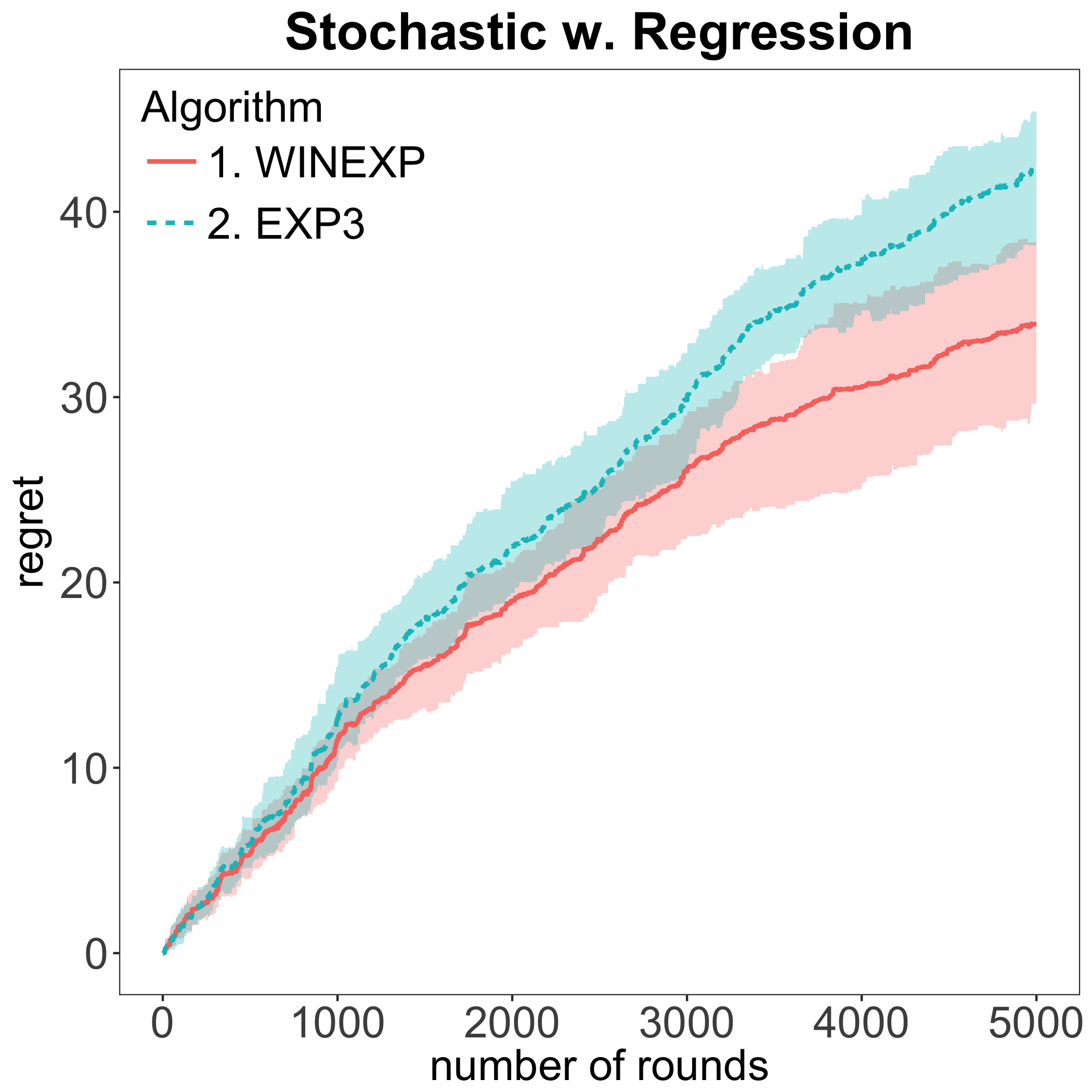}
\end{subfigure}
\hfill
\begin{subfigure}{0.33\textwidth}
\centering
    \includegraphics[width=0.9\textwidth]{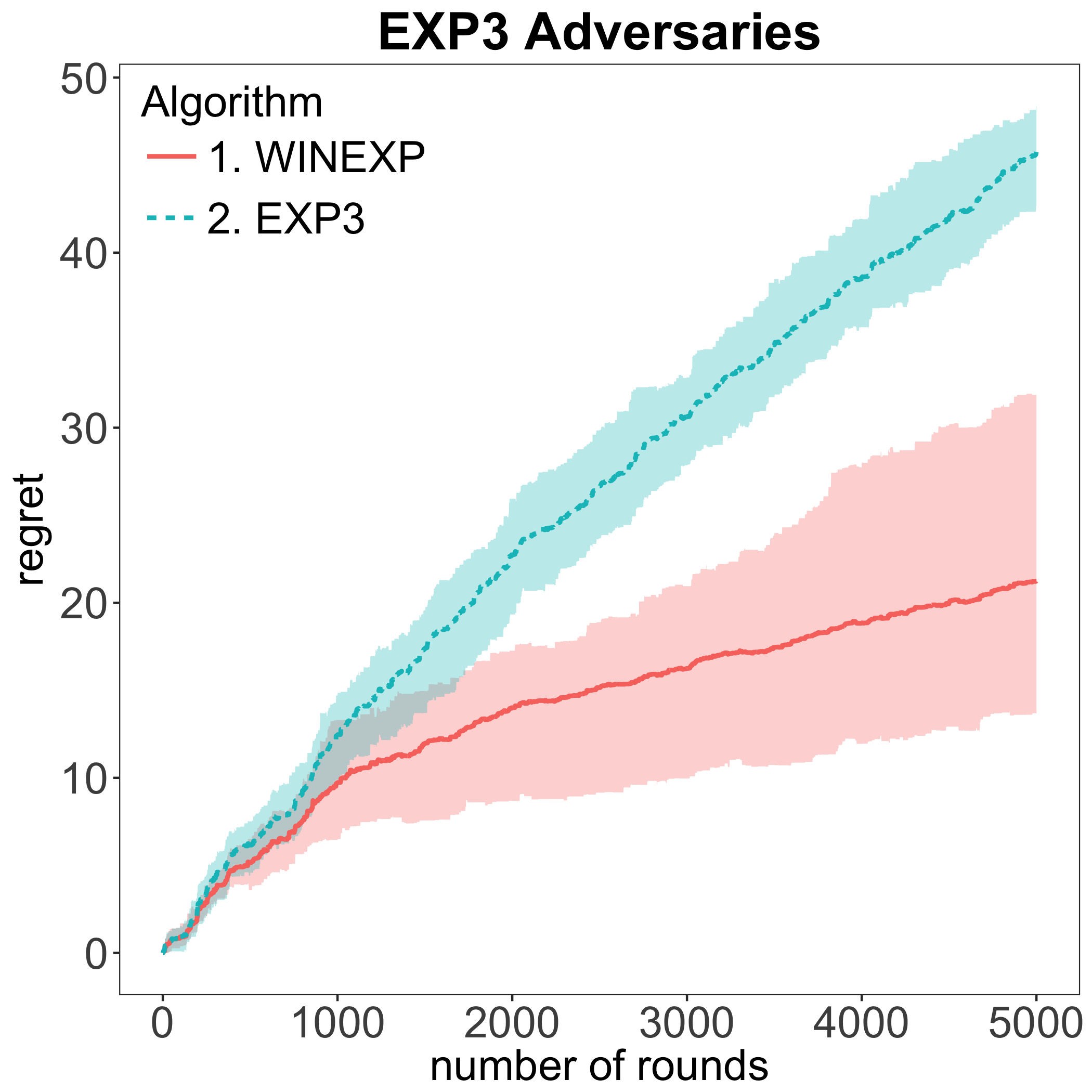}
\end{subfigure}%
\begin{subfigure}{0.33\textwidth}
\centering
    \includegraphics[width=0.9\textwidth]{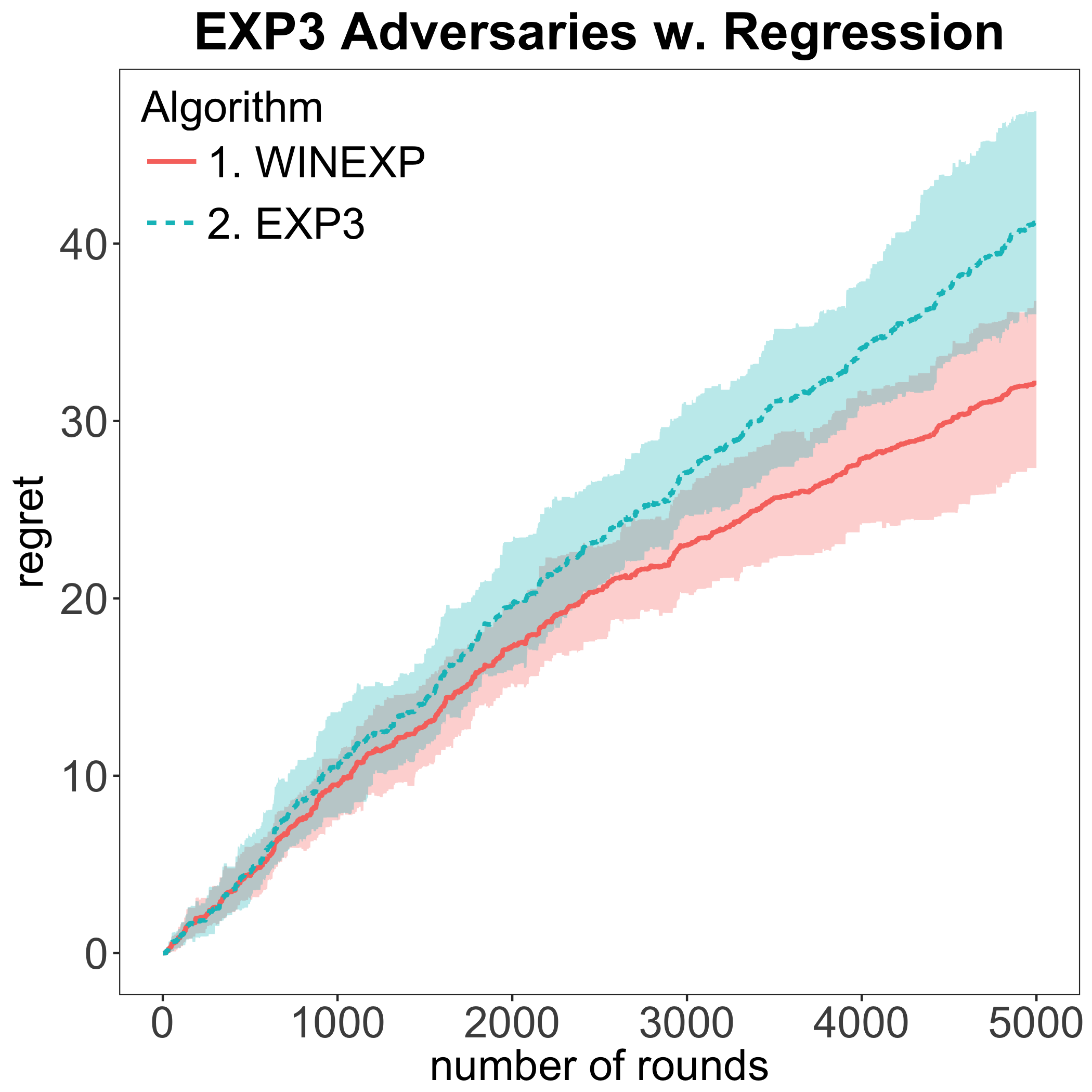}
\end{subfigure}
\hfill
\begin{subfigure}{0.33\textwidth}
\centering
    \includegraphics[width=0.9\textwidth]{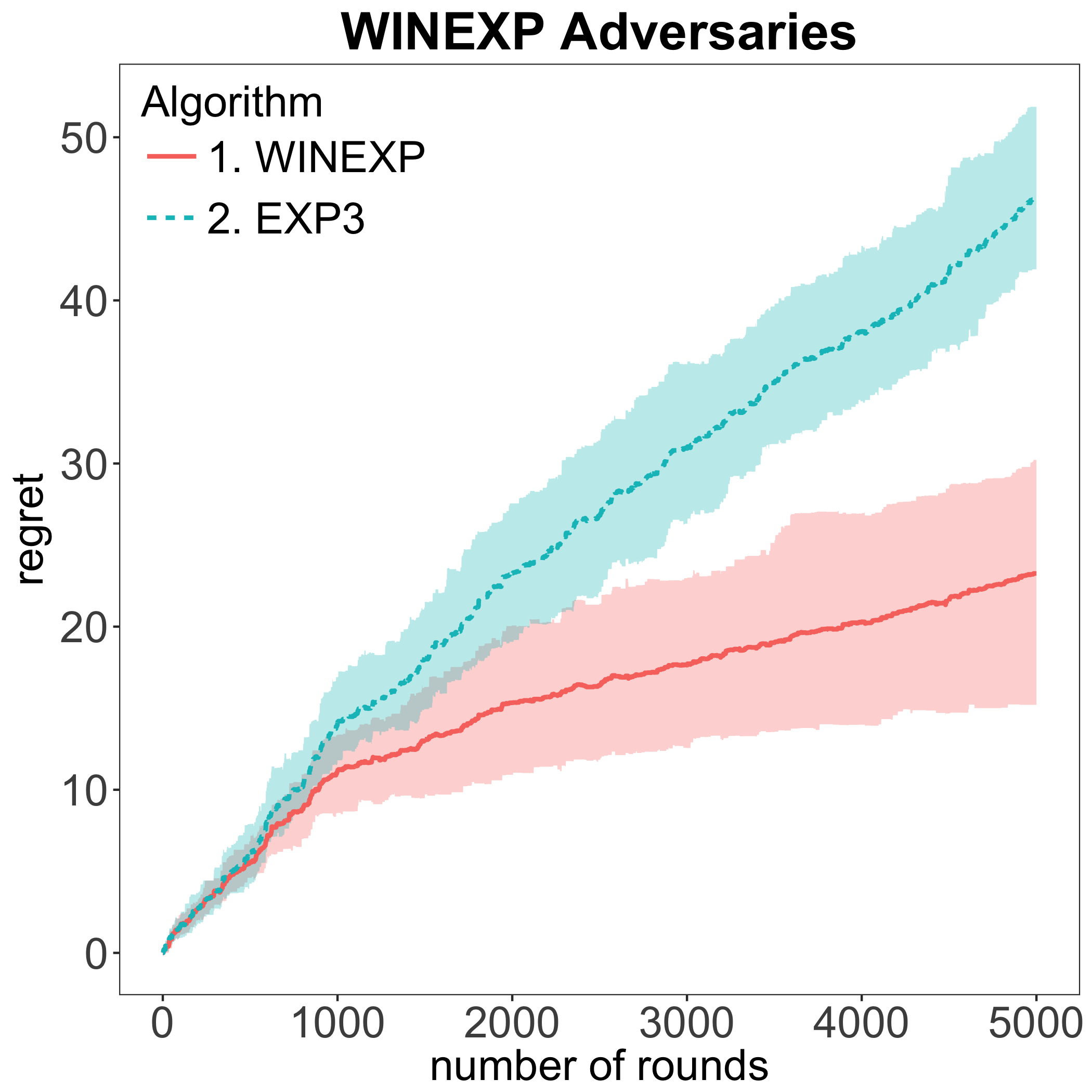}
\end{subfigure}%
\begin{subfigure}{0.33\textwidth}
\centering
    \includegraphics[width=0.9\textwidth]{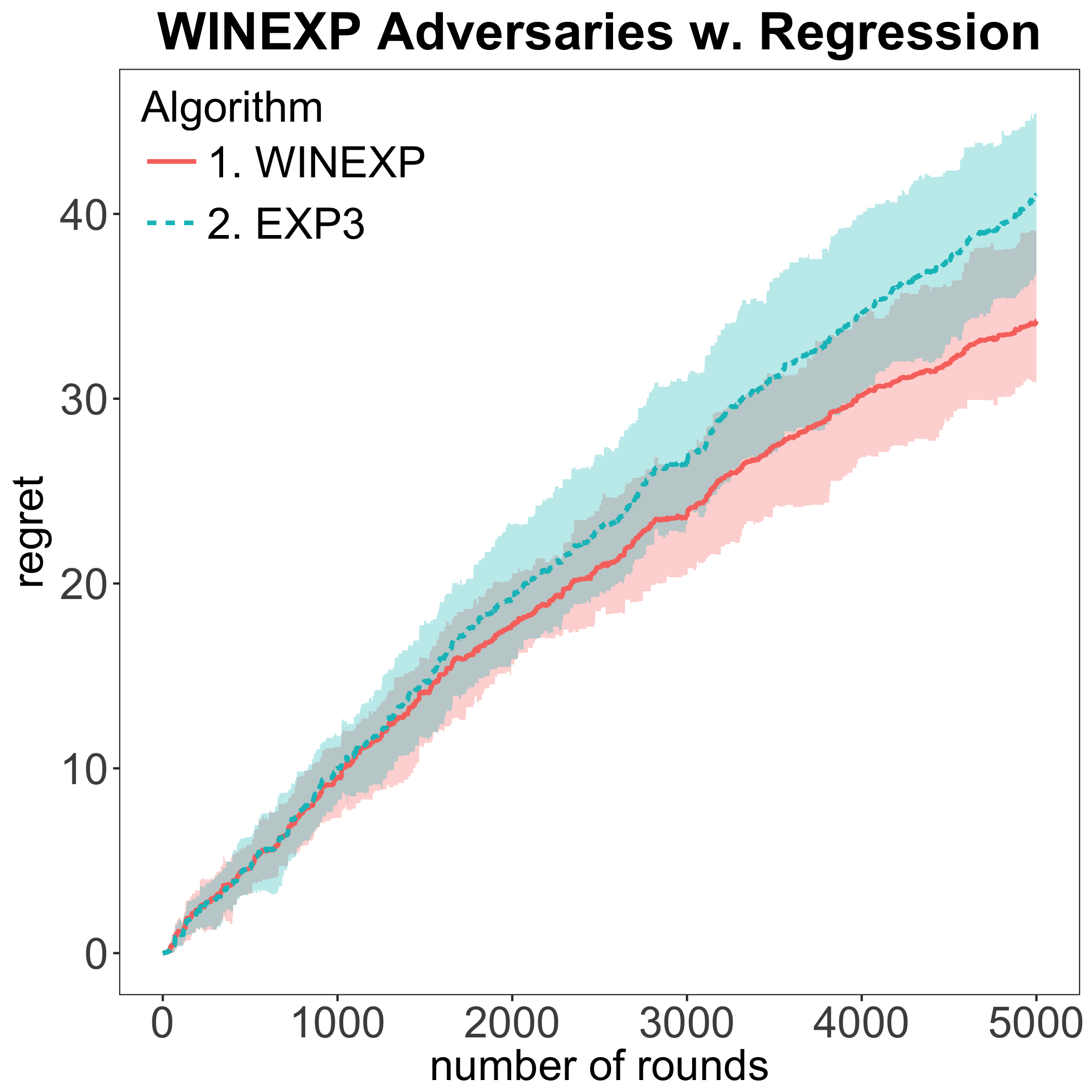}
\end{subfigure}
\caption{Regret of $\winexp$ vs EXP3 with CTR $\sim \calN(0.5, 0.16)$ and $\epsilon = 0.01$.}
\label{fig:regr_normal}
\end{figure}

%% file: discussion.tex
We addressed learning in repeated auction scenarios were bidders do not know their valuation for the items at sale. We formulated an online learning framework with partial feedback which captures the information available to bidders in typical auction settings like sponsored search and provided an algorithm which achieves almost full information regret rates. Hence, we portrayed that not knowing your valuation is a benign form of incomplete information for learning in auctions. Our experimental evaluation also showed that the improved learning rates are robust to violations of our assumptions and are valid even when the information assumed is corrupted. We believe that exploring further avenues of relaxing the informational assumptions (e.g., what if the value is only later revealed to a bidder or is contigent upon the competitiveness of the auction) or being more robust to erroneous information given by the auction system is an interesting future research direction. We believe that our outcome-based learning framework can facilitate such future work.

%% file: appendix.tex
\newpage
\section*{Appendix}
\begin{appendix}

\section{Omitted Algorithms}\label{appendix:algos}

Essentially, the family of our $\winexp$ algorithms is parametrized by the step-size $\eta$-parameter, the estimate of the utility that the learner gets at every timestep $\tilde{u}_t(b)$ and finally, the type of feedback that he receives after every timestep $t$. Clearly, both $\eta$ and the estimate of the utility depend crucially on the particular type of feedback.

In this section, we present the specifics of the algorithms that we omitted from the main body of the text, due to lack of space. 

\begin{comment}
\subsection{Outcome based batch-reward feedback}

\begin{algorithm}[H]
\begin{algorithmic}
\State Let $\pi_1(b) =\frac{1}{|B|}$ for all $b\in B$ (i.e. the uniform distribution over bids), $\eta = \sqrt{\frac{\log\left(|B| \right)}{2T|O|}}$
\For{each iteration t}
\State Draw an action $b_t$ from the multinomial distribution based on $\pi_t(\cdot)$
\State Observe $x_t(\cdot)$, chosen outcomes $o_\tau, \forall \tau \in I_t$, average reward function conditional on each realized outcome $Q_t(b,o)$ and the realized frequencies for each outcome $f_t(o) = \frac{|I_{to}|}{|I_t|}$.
\State Compute estimate of utility: 
\begin{equation}
\tilde{u}_t(b) = \sum_{o \in O} \frac{\Pr_t \left[o|b \right]}{\Pr_t[o]} f_t(o) \left(Q_t(b,o)-1 \right)
\end{equation}
\State Update $\pi_t(\cdot)$ based on the Exponential Weights Update: 
\begin{equation}
\forall b\in B: \pi_{t+1}(b) \propto \pi_{t}(b)\cdot \exp\left\{\eta \cdot \tilde{u}_t(b)\right\}
\end{equation}
\EndFor
\end{algorithmic}
\caption{$\winexp$ algorithm for learning with outcome-based batch-reward feedback}\label{alg:winexp4}
\end{algorithm}
\end{comment}

\subsection{Outcome-based feedback graph over outcomes}
\begin{algorithm}[H]
\begin{algorithmic}
\State Let $\pi_1(b) =\frac{1}{|B|}$ for all $b\in B$ (i.e. the uniform distribution over bids), $\eta = \sqrt{\frac{\log(|B|)}{8T\alpha \ln\left(\frac{16|O|^2 T}{\alpha}\right)}}$
\For{each iteration t}
\State Draw an action $b_t \sim \pi_t(\cdot)$, multinomial
\State Observe $x_t(\cdot)$, chosen outcome $o_t$ and associated reward function $r_t(\cdot, o_t)$
\State Observe and associated reward function $r_t(\cdot, \cdot)$ for all neighbor outcomes $N_\eps^{in}, N_\eps^{out}$ 
\State Compute estimate of utility:
\begin{equation}
\tilde{u}_t(b) = \1\{o_t\in O_{\epsilon}\} \sum_{o \in N_{\epsilon}^{out}(o_t)}
\frac{(r_t(b, o) - 1) \Pr_t[o|b]}{\sum_{o'\in N_{\epsilon}^{in}(o)}\Pr_t[o']}
\end{equation}
\State Update $\pi_t(\cdot)$ based on the Exponential Weights Update: 
\begin{equation}
\end{equation}
\EndFor
\end{algorithmic}
\caption{$\winexpG$ algorithm for learning with outcome-based feedback and a feedback graph over outcomes}\label{alg:winexpG}
\end{algorithm}

\section{Omitted proofs from Section \ref{SEC:OUTCOME-BASED}} \label{appendix:outcome}
We first give a lemma that bounds the moments of our utility estimate.

\begin{lemma}\label{lem:moments2}
At each iteration $t$, for any action $b\in B$, the random variable $\tilde{u}_t(b)$ is an unbiased estimate of the true expected utility $u_t(b)$, i.e.: $\forall b\in B: \E\left[\tilde{u}_t(b)\right] = u_t(b)-1$ and has expected second moment bounded by: $\forall b\in B: \E\left[\left(\tilde{u}_t(b)\right)^2\right]\leq 4\sum_{o\in O} \frac{\Pr_t[o|b]}{\Pr_t[o]}$.
\end{lemma}

\begin{proof}[Proof of Lemma \ref{lem:moments2}]
According to the notation we introduced before we have: 
\begin{align*}
\E\left[\tilde{u}_t(b)\right] &= \E_{o_t}\left[\frac{\left(r_t(b, o_t)-1\right)\cdot \Pr_t[o_t|b]}{\Pr_t[o_t]}\right] =\sum_{o\in O} \frac{\left(r_t(b, o)-1\right)\cdot \Pr_t[o|b]}{\Pr_t[o]} \textstyle{\Pr_t}[o]\\
& = \sum_{o\in O} r_t(b, o) \Pr_t[o|b]  - 1 = u_t(b) - 1
\end{align*}
Similarly for the second moment:
\begin{align*}
\E\left[\tilde{u}_t(b)^2\right] &\leq \E_{o_t}\left[\frac{(r_t(b, o_t) - 1)^2 \Pr_t[o_t|b]^2}{\Pr_t[o_t]^2} \right]
= \sum_{o\in O} \frac{(r_t(b, o) - 1)^2 \Pr_t[o|b]^2}{\Pr_t[o]^2} \Pr_t[o]\\
&  \leq \sum_{o\in O} \frac{4 \Pr_t[o|b]}{\Pr_t[o]}
\end{align*}
where the last inequality holds since $r_t(\cdot,\cdot)\in [-1,1]$.
\end{proof}

\begin{proof}[Proof of Theorem \ref{thm:outcome-based-main}]
Observe that regret with respect to utilities $u_t(\cdot)$ is equal to regret with respect to the translated utilities $u_t(\cdot) -1$. We use the fact that the exponential weight updates with an unbiased estimate $\tilde{u}_t(\cdot) \leq 0$ of the true utilities, achieves expected regret of the form:
\begin{align*}
R(T) \leq~& \frac{\eta}{2} \sum_{t=1}^T \sum_{b\in B} \pi_t(b) \cdot \E\left[\left(\tilde{u}_t(b)\right)^2\right] + \frac{1}{\eta} \log(|B|)
\end{align*}
For a detailed proof of the above, we refer the reader to Appendix \ref{appendix:a}. Invoking the bound on the second moment by Lemma \ref{lem:moments2}, we get:
\begin{align*}
R(T) \leq~& 2\eta \sum_{t=1}^T \sum_{b\in B} \pi_t(b) \cdot \sum_{o\in O} \frac{\Pr_t[o|b]}{\Pr_t[o]} + \frac{1}{\eta} \log(|B|)\\
=~& 2\eta \sum_{t=1}^T \sum_{o\in O}\sum_{b\in B} \pi_t(b) \cdot  \frac{\Pr_t[o|b]}{\Pr_t[o]} + \frac{1}{\eta} \log(|B|)\\
\leq~& 2\eta T|O| + \frac{1}{\eta} \log(|B|)
\end{align*}
Picking $\eta = \sqrt{\frac{\log(|B|)}{2T|O|}}$, we get the theorem.
\end{proof}

\subsection{Comparison with Results in Weed et al.}\label{sec:app-weed}

We note that our result in Example \ref{ex:weed} also \emph{recovers} the results of \citet{WRP16}, who work in the continuous bid setting (i.e. $b\in [0,1]$). In order to describe their results, consider the grid ${\cal L}_T$ formed by the maximum bids from other bidders $m_t = \max_{j \neq i} b_{jt}$ for all the rounds. Let $l^o=(m_t, m_{t'})$ be the widest interval in ${\cal L}_T$, that contains an optimal fixed bid in hindsight and let $\Delta^o$ denote its length. \citet{WRP16} provide an algorithm for learning the valuation, which yileds regret $4\sqrt{T\log(1/\Delta^o)}$. 

The same regret can be achieved, if we simply consider a partition of the bidding space $[0,1]$ into $\frac{1}{\epsilon}$ intervals of equal length $\epsilon$, for $\epsilon < \Delta^o$, and run our algorithm on this discretized bid space $B$.
If $l^o$ contains an optimal bid, then any bid $b\in l^o$ is also optimal in-hindsight, since all such bids achieve the same utility. Since $\Delta^o>\epsilon$, there must exist a discretized bid $b_{\epsilon}^*\in B\cap l^o$. Thus, $b_{\epsilon}^*$ is also optimal in hindsight. Hence, regret against the best fixed bid in $[0,1]$ is equal to regret against the best fixed discretized bid in $B$. By our Theorem \ref{thm:outcome-based-main}, the latter regret is $4\sqrt{T\log(\nicefrac{1}{\epsilon})}$, which can be made arbitrarily close to the regret bound achieved by \citet{WRP16}, who use a more intricate adaptive discretization. Similar to \citet{WRP16}, knowledge of $\Delta^o$ can be bypassed by instead defining $\Delta^o$ as the length of the smallest interval in ${\cal L}_T$ and then using the standard doubling trick, i.e.: keep an estimate of $\Delta^o$ and once this estimate is violated, divide $\Delta^o$ in half and re-start your algorithm. The latter only increases the regret by a constant factor.

\section{Notes on Subsection \ref{SEC:BATCH}}\label{appendix:notes}

If one is interested in optimizing the \emph{sum} of utilities at each iteration rather than the \emph{average}, then if all iterations have the same number of batches $|I|$, this simply amounts to rescaling everything by $|I|$, which would lead to an $|I|$ blow up in the regret. 

If different periods have different number of batches and $I_{\max}$ is the maximum number of batches per iteration, then we can always pad the extra batches with all zero rewards. This would amount to again multiplying the regret by $I_{\max}$ and would change the unbiased estimates at each period to be scaled by the number of iterations in that period:  
\begin{equation}\label{eqn:batch-unbiased-3}
\tilde{u}_t(b) = \frac{|I_t|}{I_{\max}}\sum_{o\in O}\frac{\Pr_t[o | b]\cdot \Pr_t[o | b_t]}{\Pr_t[o]}  \left(Q_t(b, o) - 1\right) 
\end{equation}
and then we would invoke the same algorithm. This essentially puts more weight on iterations with more auctions, so that the "step-size" of the algorithm depends on how many auctions were run during that period. It is easy to see that the latter modification would lead to regret $4I_{\max}\sqrt{T\log\left(|B|\right)}$ in the sponsored search auction application.

\section{Omitted Proofs from Section \ref{SEC:BATCH}}\label{appendix:batch}
We first prove an upper bound on the moments of our estimates used in the case of batch rewards. 
\begin{lemma}\label{lem:moments4}
At each iteration $t$, for any action $b\in B$, the random variable $\tilde{u}_t(b)$ is an unbiased estimate of $u_t(b) - 1$ and can actually be constructed based on the feedback that the learner receives: $\forall b\in B: \tilde{u}_t(b)= \sum_{o\in O}\frac{\Pr_t[o | b]}{\Pr_t[o]} f_t(o) \left(Q_t(b, o) - 1\right) \label{eqn:batch-unbiased}$ 
and has expected second moment bounded by: $\forall b\in B: \E\left[\left(\tilde{u}_t(b)\right)^2\right]\leq 4\sum_{o\in O}\frac{\Pr_t[o | b]}{\Pr_t[o]}$.
\end{lemma}
\begin{proof}[Proof of Lemma \ref{lem:moments4}]
For the estimate of the utility it holds that:
\begin{align}
\tilde{u}_t(b) =~& 
\frac{1}{|I_t|}\sum_{\tau\in I_t}\frac{(r_\tau(b, o_\tau) - 1) \Pr_t[o_\tau | b]}{\Pr_t[o_\tau]}\nonumber\\
=~& 
\frac{1}{|I_t|}\sum_{o\in O}\sum_{\tau\in I_{to}}\frac{(r_\tau(b, o) - 1) \Pr_t[o | b]}{\Pr_t[o]}\nonumber\\
=~& 
\sum_{o\in O: |I_{to}|>0}\frac{\Pr_t[o | b]}{\Pr_t[o]} f_t(o) \frac{1}{|I_{to}|}\sum_{\tau\in I_{to}}(r_\tau(b, o) - 1) \nonumber\\
=~& 
\sum_{o\in O}\frac{\Pr_t[o | b]}{\Pr_t[o]} f_t(o) \left(Q_t(b, o) - 1\right) \label{eqn:batch-unbiased}
\end{align}
From the first equation it follows along identical lines, that this is an unbiased estimate, while from the last equation it is easy to see that this unbiased estimate can be constructed based on the feedback that the learner receives.

Moreover, we can also bound the second moment of these estimates by a similar quantity as in the previous section:
\begin{align*}
\E[\tilde{u}_t(b)^2]=~&\sum_{b_t\in B}\E\left[\left(\sum_{o\in O}\frac{\Pr_t[o | b]}{\Pr_t[o]} f_t(o) \left(Q_t(b, o) - 1\right)\right)^2 \bigg| b_t\right] \pi_t(b_t)\\
\leq~&\sum_{b_t\in B}\E\left[\sum_{o\in O}\left(\frac{\Pr_t[o | b]}{\Pr_t[o]} \left(Q_t(b, o) - 1\right)\right)^2  f_t(o) \bigg| b_t\right] \pi_t(b_t) \tag{By Jensen's inequality}\\
=~& \sum_{b_t\in B}\sum_{o\in O}\left(\frac{\Pr_t[o | b]}{\Pr_t[o]}  \left(Q_t(b, o) - 1\right)\right)^2 \E[f_t(o)|b_t]\cdot \pi_t(b_t) \\
=~& \sum_{o\in O}\left(\frac{\Pr_t[o | b]}{\Pr_t[o]}  \left(Q_t(b, o) - 1\right)\right)^2 \sum_{b_t\in B}\E[f_t(o)|b_t]\cdot \pi_t(b_t) \\
=~& \sum_{o\in O}\left(\frac{\Pr_t[o | b]}{\Pr_t[o]}  \left(Q_t(b, o) - 1\right)\right)^2 \textstyle{\Pr_t}[o]\\
\leq~& 4\sum_{o\in O}\frac{\Pr_t[o | b]}{\Pr_t[o]}
\end{align*}
\end{proof}

Then following the same techniques in Theorem~\ref{thm:outcome-based-main}, it is straightforward to conclude the proof of the corollary.

\section{Omitted Proofs from Section \ref{SEC:CONTINUOUS}} 
\label{appendix:cont}\label{sec:app-doubling-trick}

\begin{proof}[Proof of Lemma \ref{lem:de-piece}]
Let $\opt = \argsup_{b \in \B} \sum_{t=1}^T u_t(b)$ be the best fixed action in the continuous action space $\B$ in hindsight. Since $\eps < \Delta^o$, then $b^*$ must belong to some $d$-dimensional $\eps$-cube, either as an interior point or as a limit of interior points, as expressed by Definition \ref{defn:piecewise}. %
The utility is $L$-Lipschitz within this $\eps$-cube and since $\epsilon< \Delta^o$, each cube contains at least one point in the discretized space $B$. For the case where $\opt$ is achieved as the limit of interior points then for every $\delta>0$ there exist an interior point of some cube $\tilde{b}$, such that $\sum_{t=1}^T u_t(\tilde{b})\geq \opt-\delta$. The same obviously holds when $\opt$ is achieved by an interior point. Let $\hat{b}$ be the closest discretized point to $\tilde{b}$ that lies in the same cube as $\tilde{b}$. 
Since $\|\hat{b}-\tilde{b}\|_{\infty}\leq \epsilon$, by the Lipschitzness of the average reward function within each cube, we get:
\begin{align*}
\opt \leq \sum_{t=1}^T u_t(\tilde{b}) + \delta \leq  \sum_{t=1}^T u_t(\hat{b}) + \delta + \epsilon L T \leq \sup_{b\in B} \sum_{t=1}^T u_t(\hat{b}) + \delta + \epsilon L T
\end{align*}
Since we can take $\delta$ as close to zero as we want, we get the lemma.
\end{proof}

\begin{proof}[Proof of Theorem \ref{thm:continuous-lipschitz-known}]
From Lemma \ref{lem:de-piece} we know that for $\eps < \Delta^o$, the discretization error is $DE(B, \B) \leq \epsilon L T$. Combining Lemma \ref{lem:regret-de} and Corollary \ref{corol:batch-rewards}, we have
\begin{align*}
R(T, \B) &\leq R(T, B) + DE(B, \B) = 2\sqrt{2T|O|\log(|B|)} + \eps L T\\
& = 2\sqrt{2T|O|\log \left(\frac{1}{\eps^d}\right)} + \eps L T\\
& = 2\sqrt{2dT|O|\log\left(\frac{1}{\eps}\right)} + \eps L T\\
& = 2\sqrt{2dT|O|\log\left(\max \left\{LT, \frac{1}{\Delta^o} \right\}\right)} + \min \left\{\frac{1}{LT},\Delta^o\right\}\\ 
& \leq 2\sqrt{2dT|O|\log\left(\max \left\{LT, \frac{1}{\Delta^o} \right\}\right)} + 1%
\end{align*}
\end{proof}

\paragraph{Unknown Lipschitzness constant.} In Theorem \ref{thm:continuous-lipschitz-known} the discretization parameter $\eps$ depends on the prior knowledge of the Lipschitzness constant, $L$, the number of rounds, $T$ and the minimum edge length of each $d$-dimensional cube, $\Delta^o$. In order to address the problem that in general we do not know any of those constants a priori, we will apply a standard doubling trick (\cite{ACBFS02}) to remove this dependence. We assume that $T$ is upper bounded by a constant $T_M$ and similarly we also assume that $\log\left(\max\left\{LT, \frac{1}{\Delta^o}\right\}\right)$ is upper bounded by a constant. 

We will then initialize two bounds: $B_T=1$ and $B_{\Delta^o,LT}=1$ and run the $\winexp$ algorithm with step size $\sqrt{\frac{\log(\nicefrac{1}{\eps})}{2B_T|O|}}$ and $\eps=\min \left\{\frac{1}{L T},\Delta^o \right\}$ until $t\leq B_T$ or $\log\left(\max\left\{tL,\frac{1}{\Delta^o}\right\}\right) \leq B_{\Delta^o,LT}$ fails to hold. If one of these discriminants fails, then we double the bound and restart the algorithm. This modified strategy only increases the regret by a constant factor.

\begin{corollary}\label{cor:doubling}
The $\winexp$ algorithm run with the above doubling trick achieves an expected regret bound
$\mathcal{R}(T) \leq 25\sqrt{2dT|O|\log\left(\max \left\{LT, \frac{1}{\Delta^o} \right\}\right)} +1$
\end{corollary}

\begin{proof}[Proof of Corollary \ref{cor:doubling}]
Based on the doubling trick that we described above, we divide the algorithm into stages in which $B_T$ and $B_{\Delta^o, LT}$ are constants.  Let $B^*_L$, and $B^*_{\Delta^o,LT}$ be the values of $B_L$ and $B_{\Delta,LT}$ respectively when the algorithm terminates. There is at most a total of $\log\left(B_T^* \right) +\log\left(B_{\Delta^o,LT}^* \right) + 1$ stages in this doubling process. Since the actual expected regret is bounded by the sum of the regret of each stage, following the result of Theorem \ref{thm:continuous-lipschitz-known}, we have
\begin{align*}
R(T) &\leq \sum_{i =0}^{\left\lceil\log\left(B_T^*\right)\right\rceil} \sum_{j=0}^{\left\lceil\log\left(B_{\Delta^o,LT}^*\right)\right\rceil}\left(2 \sqrt{2d2^i|O|2^j}\right) + \log\left(B_T^* \right)+\log\left(B_{\Delta^o,LT}^* \right) + 1 \\
& = \sum_{i =0}^{\left\lceil\log\left(B_T^*\right)\right\rceil} \sum_{j=0}^{\left\lceil\log\left(B_{\Delta^o,LT}^*\right)\right\rceil} \left(2 \sqrt{2d|O|2^i \cdot 2^j}\right) + \log\left(B_T^* B_{\Delta^o,LT}^*\right) + 1\\
&= \left[\sum_{i =0}^{\left\lceil\log\left(B_T^*\right)\right\rceil} \left(\sqrt{2} \right)^i\right] \cdot \left[\sum_{j =0}^{\left\lceil\log\left(B_{\Delta^o,LT}^*\right)\right\rceil} \left(\sqrt{2} \right)^j\right] 2\sqrt{2d|O|} + \log\left(B_T^* B_{LT,\Delta^o}^*\right) + 1\\
&= \frac{1-\sqrt{2}^{\lceil\log(B^*_T)\rceil+1}}{1-\sqrt{2}}\cdot \frac{1-\sqrt{2}^{\lceil\log(B^*_{\Delta^o,LT})\rceil+1}}{1-\sqrt{2}} \cdot 2\sqrt{2d|O|} +\log\left(B_T^* B_{\Delta^o,LT}^*\right) + 1\\
&\leq \left(\frac{\sqrt{2}}{\sqrt{2}-1}\right)^2\sqrt{B^*_T B^*_{\Delta^o, LT}} \cdot 2 \sqrt{2d|O|} + \log\left(B_T^* B_{\Delta^o,LT}^*\right) + 1\\
&= \left(\frac{\sqrt{2}}{\sqrt{2}-1}\right)^2 \cdot 2\sqrt{2d|O|B^*_T B^*_{\Delta^o,LT}} + \log\left(B_T^* B_{\Delta^o,LT}^*\right) + 1\\
& \leq 25\sqrt{2d|O|B^*_T B^*_{\Delta^o,LT}} + 1
\end{align*}

Combining the fact that $B^*_T\leq T$ and $B^*_{\Delta^o,LT}\leq \log\left(\max\left\{LT, \frac{1}{\Delta^o}\right\}\right)$ as well as the above inequalities, we complete the proof.
\end{proof}

\subsection{Omitted Proofs from Section \ref{sec:sponsored-lipschitz}}
\begin{proof}[Proof of Theorem \ref{thm:lipschitz-weighted-gsp}]
Consider a bidder $i$. Observe that conditional on the bidder's score $s_i$, his utility remains constant if he is allocated the same slot. Moreover, when the slots are different, then the difference in utilities is at most $2$, since utilities lie in $[-1,1]$. Moreover, because the slots are allocated in decreasing order of rank scores, the slot allocation of a bidder is different under $b_i$ and $b_i'$ only if there exists a bidder $j$, who passes the rank-score reserve (i.e. $s_j\cdot b_j\geq r$) and whose rank-score $s_j\cdot b_j$ lies in the interval $[s_i\cdot b_i, s_i\cdot b_i']$. Hence, conditional on $s_i$, the absolute difference between the bidder's expected utility when he bids $b_i$ and when he bids $b_i+\eps$, with $\eps > 0$, is upper bounded by:
\begin{align*}
2\cdot\Pr\left[\exists j\neq i \text{ s.t }s_j\cdot b_j \in [s_i\cdot b_i, s_i\cdot (b_i + \eps)] \text{ and } s_j\cdot b_j \geq r ~|~ s_i\right] 
\end{align*}
By a union bound the latter is at most:
\begin{align*}
2\cdot \sum_{j \neq i} \Pr \left[s_j \in \left[\frac{s_ib_i}{b_j}, \frac{s_i(b_i+\eps)}{b_j}\right] \text{ and } s_j\cdot b_j \geq r ~|~ s_i\right]
\end{align*}
Since $s_j\in [0,1]$, the previous quantity is upper bounded by replacing the event $s_j\cdot b_j\geq r$ by $b_j\geq r$. This event is independent of the scores and we can then write the above bound as:
\begin{align*}
2 \cdot \sum_{j \neq i \text{ s.t. } b_j \geq r} \Pr \left[s_j \in \left[\frac{s_ib_i}{b_j}, \frac{s_i(b_i+\eps)}{b_j}\right] \Big| s_i\right]
\end{align*}
Since each quality score $s_j$ is drawn independently from an $L$-Lipschitz CDF $F_j$, we can further simplify the bound by:
\begin{align*}
2 \cdot \sum_{j \neq i \text{ s.t. } b_j \geq r} \left[F_j \left(\frac{s_i(b_i+\eps)}{b_j} \right) - F_j \left( \frac{s_i b_i}{b_j}\right)\right] &\leq 2 \cdot \sum_{j \neq i \text{ s.t. } b_j \geq r} L \frac{s_i\eps}{b_j} \leq 2 \cdot \sum_{j \neq i \text{ s.t. } b_j \geq r} L \frac{s_i\eps}{r} \leq \frac{2nL}{r} \epsilon
\end{align*}
Since the absolute difference of utilities between these two bids is upper bounded conditional on $s_i$, by the triangle inequality it is also upper bounded even unconditional on $s_i$, which leads to the Lipschitz property we want:
\begin{equation}
\big| u_i(b_i, \mbf{b}_{-i}, r) - u_i(b_i+\epsilon, \mbf{b}_{-i}, r)\big| \leq \frac{2nL}{r} \epsilon
\end{equation}
\end{proof}

\section{Omitted proofs from section \ref{SEC:SWITCH-POA}}\label{appendix:switch}

\subsection{Switching Regret and PoA}

\begin{proof}[Proof of Corollary \ref{cor:switch}]
We first observe that the results proven in \cite{GLL12} for a prediction algorithm $\mathcal{A}$ with \emph{regret} upper bounded by $\rho(T)$ hold also for algorithms $\mathcal{A}$ for which we know upper bound of their expected regrets. Specifically, if algorithm $\mathcal{A}$ has an upper bound of $\rho(T)$ for its expected regret, where $\rho(T)$ is a concave, non-decreasing, $[0,+\infty) \to [0,+\infty)$ function, then Lemma $1$ from \cite{GLL12} holds for \emph{expected} regret. With that in mind, we can directly apply the \emph{Randomized Tracking Algorithm} and get expected switching regret upper bounded by: 
\begin{equation}\label{eq:exp-switch}
\left(C(TP) +1 \right)L_{C(TP),T}\rho \left(\frac{T}{\left(C(TP) +1 \right)L_{C(TP),T}} \right) + \sum_{t=1}^T \frac{\eta_t}{8} + \frac{r_T\left( \left(C(TP) +1 \right)L_{C(TP),T - 1}-1\right)}{\eta_T}
\end{equation}
where $TP$ is the switching path of the optimal bids and $C(TP)$ is the number of switches in the optimal bid according to this path.

We proceed by making sure that the conditions for the upper bound of the expected regret of $\winexp$ satisfy the conditions required by algorithm $\mathcal{A}$ in \cite{GLL12}. Indeed, the upper bound of the expected regret of our algorithm, $\sqrt{2dT|O|\log \left(\max \left\{LT, \frac{1}{\Delta^o} \right\} \right)} +1$, is non decreasing in $T$. Also, at timestep $t=0$, we incur no regret. We also apply the following slight modifications in Algorithm $2$ in \cite{GLL12} so as to match the nature of our problem. First, instead of computing the expected loss at each timestep $t$, we will now compute the expected outcome-based utility, i.e. $\bar{u}_t\left(\pi_t \right) = \sum_{b \in B} \pi_t(b)\E_{o_t} \left[\tilde{u}_t(b)\right]$. Second, instead of the cumulative loss of their algorithm $\mathcal{A}$ we will now use the cumulative outcome-based expected utility of $\winexp$, i.e. $\bar{U}_t\left(\winexp, T \right) = \sum_{c=0}^C \bar{U}_{\winexp}(t_c, t_{c+1})$, where 
\begin{equation*}
\bar{U}_{\winexp}(t_c, t_{c+1}) = \sum_{s=t_c}^{t_{c+1}-1} \bar{u}_s\left(\pi_{\winexp,s}(t_c) \right)
\end{equation*}
is the cumulative outcome-based expected utility gained from our $\winexp$ algorithm in the time interval $[t_c, t_{c+1})$\footnote{We clarify here that these time intervals are with respect to the switching bids.} with respect to $\bar{u}_s$ for $s \in [t_c, t_{c+1})$. Now, we are computing the regret components of \cite{GLL12} so as to achieve the desired result. 

Before we show the specifics of the computation, we note here that $g > 0$ is a \emph{parameter} of the Tracking Regret algorithm presented by \cite{GLL12} and can be set a priori from the designer of the algorithm. The complexity of $g$ affects the computational complexity of the algorithm and there is a tradeoff between the computational complexity and the regret of the algorithm. For our computations here, we will set
\begin{equation}\label{eq:g}
g+1 = \left(\frac{T}{C(TP)+1}\right)^\alpha
\end{equation}

where $0 < \alpha < 1$ is a constant. Now, we are ready to compute the components of the regret: 
\begin{align*}
A   &= L_{C(TP),T} \left(C(TP) +1 \right) R_{\winexp} \left(\frac{T}{L_{C(TP),T}\left(C(TP) +1 \right)} \right) \\
    &\leq 25 \left(\frac{\log \left(\frac{T}{C(TP)+1}\right)}{\log (g+1)}+2 \right) \left(C(TP)+1\right) \left(\sqrt{2d|O|\frac{T\log (g+1) \log (m)}{\log \left(\frac{T}{C(TP)+1} \right) + 2\log (g+1)}} + 1 \right) \\
    &= 50 \cdot \left(2 + \frac{1}{\alpha}\right)\cdot \left(C(TP)+1\right) \sqrt{2d|O|\cdot \frac{\alpha}{1+2\alpha} \cdot T\log (m)}\\
    &\leq 50 \sqrt{\frac{1+2\alpha}{\alpha}\cdot\left(C(TP)+1\right)^2 2d|O|T\log (m)}\\
    &\leq 50 \sqrt{\left(2+\frac{1}{\alpha}\right)\cdot\left(C+1\right)^2 2d|O|T\log (m)}
\end{align*}
where in the second equality we have denoted $\log (m) = \log \left(\max \left\{LT, \frac{1}{\Delta^o}\right\} \right)$ and the last inequality comes from the fact that $C$ is the upper bound on the number of switches that the transition path $TP$ can have. %
Moving on to the computation of the rest of the components of the regret: 
\begin{align*}
B &=\sum_{t=1}^T \frac{\eta_t}{8} \leq \frac{1}{8}\sqrt{\frac{T\log\left(\nicefrac{1}{\eps}\right)}{2|O|}} = O\left(\sqrt{\frac{T}{|O|}} \right)\\
D   &= r_T \left( L_{C(TP),T} \left( C(TP) + 1 \right) - 1 \right) \\
    &= \left(\frac{\alpha+1}{\alpha} + \eps_2 \right) \log T + \log \left(1+\eps_2\right) - \left(\frac{\alpha+1}{\alpha} \right)\log\eps_2
\end{align*}
where $\eps_2 \in (0,1)$ is a constant. Before we conclude, we observe that even though Corollary $1$ of \cite{GLL12} is stated as a high-probability ex post result, the proof uses a result from \cite{BL06} (Lemma $4.1$) which also holds for the expected regret. According to \cite{GLL12} the switching regret is the sum of the aforementioned $A, B, D$. Thus, we get the result.
\end{proof}

\subsection{Feedback Graphs over Outcomes}\label{appendix:outcome-feedback-graph}

We first prove bounds on the moments of our unbiased estimates used in the case of a feedback graph over outcomes.

\begin{lemma}\label{lem:moments3}
At each iteration $t$, for any action $b\in B$, the random variable $\tilde{u}_t(b)$ has bias with respect to $u_t(b)-1$ bounded by:
$\big|\E\left[\tilde{u}_t(b)\right] - (u_t(b) - 1)\big| \leq 2 \epsilon |O|$ and has expected second moment bounded by: $\forall b\in B: \E\left[\tilde{u}_t(b)^2\right]\leq4\sum_{o\in O_{\epsilon}} \frac{\Pr_t[o|b]}{\sum_{o'\in N_{\epsilon}^{in}(o)}\Pr_t[o']}$.
\end{lemma}
\begin{proof}[Proof of Lemma \ref{lem:moments3}]
For the expected utility we have: 
\begin{align*}
\E\left[\tilde{u}_t(b)\right] &= \E_{o_t}\left[\1\{o_t\in O_{\epsilon}\} \sum_{o \in N_{\epsilon}^{out}(o_t)}
\frac{(r_t(b, o) - 1) \Pr_t[o|b]}{\sum_{o'\in N_{\epsilon}^{in}(o)}\Pr_t[o']}\right]\\
&=\sum_{o_t\in O_{\epsilon}} \sum_{o \in N_{\epsilon}^{out}(o_t)}
\frac{(r_t(b, o) - 1) \Pr_t[o|b]}{\sum_{o'\in N_{\epsilon}^{in}(o)}\Pr_t[o']} \Pr_t[o_t]\\
&= \sum_{o \in O_{\epsilon}} \sum_{o_t\in N_{\epsilon}^{in}(o)} 
\frac{(r_t(b, o) - 1) \Pr_t[o|b]}{\sum_{o'\in N_{\epsilon}^{in}(o)}\Pr_t[o']} \Pr_t[o_t]\\
&= \sum_{o \in O_{\epsilon}} \frac{(r_t(b, o) - 1) \Pr_t[o|b]}{\sum_{o'\in N_{\epsilon}^{in}(o)}\Pr_t[o']}\sum_{o_t\in N_{\epsilon}^{in}(o)} 
 \Pr_t[o_t]\\
& = \sum_{o\in O_{\epsilon}} (r_t(b, o) - 1) \Pr_t[o|b]\\
& = \sum_{o\in O} (r_t(b, o) - 1) \Pr_t[o|b] - \sum_{o\notin O_{\epsilon}} (r_t(b, o) - 1) \Pr_t[o|b]\\
& = u_t(b) - 1   - \sum_{o\notin O_{\epsilon}} (r_t(b, o) - 1) \Pr_t[o|b]
\end{align*}
Thus, we get that the bias of $\tilde{u}$ with respect to $u_t - 1$ is bounded by:
\begin{equation}
\big|\E\left[\tilde{u}_t(b)\right] - (u_t(b) - 1)\big| \leq 2 \epsilon |O| 
\end{equation}
Similarly for the second moment:
\begin{align}
\E\left[\tilde{u}_t(b)^2\right] &\leq \E_{o_t}\left[\left(\1\{o_t\in O_{\epsilon}\} \sum_{o \in N_{\epsilon}^{out}(o_t)}
\frac{(r_t(b, o) - 1) \Pr_t[o|b]}{\sum_{o'\in N_{\epsilon}^{in}(o)}\Pr_t[o']}\right)^2 \right] \nonumber\\
&= \sum_{o_t\in O_{\epsilon}} \left(\sum_{o \in N_{\epsilon}^{out}(o_t)}
\frac{(r_t(b, o) - 1) \Pr_t[o|b]}{\sum_{o'\in N_{\epsilon}^{in}(o)}\Pr_t[o']}\right)^2 \Pr_t[o_t] \label{eqn:variance-bnd-feedback}
\end{align}
Observe that the quantity inside the square:
\begin{align*}
\sum_{o \in N_{\epsilon}^{out}(o_t)}
\frac{(r_t(b, o) - 1)}{\sum_{o'\in N_{\epsilon}^{in}(o)}\Pr_t[o']}  \Pr_t[o|b]
\end{align*}
can be thought of as an expected value of the quantity $\frac{(r_t(b, o) - 1)}{\sum_{o'\in N_{\epsilon}^{in}(o)}\Pr_t[o']}$, were $o$ is the random variable and is drawn from the distribution of outcomes conditional on a bid $b$. 
Thus, by Jensen's inequality, the square of the latter expectation is at most the expectation of the square, i.e.:
\begin{align*}
\left(\sum_{o \in N_{\epsilon}^{out}(o_t)}
\frac{(r_t(b, o) - 1)}{\sum_{o'\in N_{\epsilon}^{in}(o)}\Pr_t[o']}  \Pr_t[o|b]\right)^2 \leq \sum_{o\in N_{\epsilon}^{out}(o_t)}  \frac{(r_t(b, o) - 1)^2}{\left(\sum_{o'\in N_{\epsilon}^{in}(o)}\Pr_t[o']\right)^2} \Pr_t[o|b]
\end{align*}
Combining with Equation \eqref{eqn:variance-bnd-feedback}, we get:
\begin{align*}
\E\left[\tilde{u}_t(b)^2\right] &\leq \sum_{o_t\in O_{\epsilon}} \sum_{o\in N_{\epsilon}^{out}(o_t)}  \frac{(r_t(b, o) - 1)^2}{\left(\sum_{o'\in N_{\epsilon}^{in}(o)}\Pr_t[o']\right)^2} \Pr_t[o|b] \Pr_t[o_t]\\
&= \sum_{o\in O_{\epsilon}} \sum_{o_t\in N_{\epsilon}^{in}(o)}   \frac{(r_t(b, o) - 1)^2}{\left(\sum_{o'\in N_{\epsilon}^{in}(o)}\Pr_t[o']\right)^2} \Pr_t[o|b] \Pr_t[o_t]\\
&= \sum_{o\in O_{\epsilon}} \frac{(r_t(b, o) - 1)^2}{\left(\sum_{o'\in N_{\epsilon}^{in}(o)}\Pr_t[o']\right)^2}\Pr_t[o|b] \sum_{o_t\in N_{\epsilon}^{in}(o)} \Pr_t[o_t]\\
&= \sum_{o\in O_{\epsilon}} \frac{(r_t(b, o) - 1)^2}{\sum_{o'\in N_{\epsilon}^{in}(o)}\Pr_t[o']}\Pr_t[o|b]\\
&\leq 4\sum_{o\in O_{\epsilon}} \frac{\Pr_t[o|b]}{\sum_{o'\in N_{\epsilon}^{in}(o)}\Pr_t[o']}
\end{align*}
where the last inequality holds since $r_t(\cdot,\cdot)\in [-1,1]$.
\end{proof}

\begin{proof}[Proof of Theorem \ref{thm:feedback-graph}] 
Observe that regret with respect to utilities $u_t(\cdot)$ is equal to regret with respect to the translated utilities $u_t(\cdot) -1$. We use the fact that the exponential weight updates with an estimate $\tilde{u}_t(\cdot) \leq 0$ which has bias with respect to the true utilities, bounded by $\kappa$, achieves expected regret of the form: %
\begin{align*}
R(T) \leq~& \frac{\eta}{2} \sum_{t=1}^T \sum_{b\in B} \pi_t(b) \cdot \E\left[\tilde{u}_t(b)^2\right] + \frac{1}{\eta} \log(|B|) + 2\kappa T
\end{align*}
For the detailed proof of the above claim, please see Appendix \ref{appendix:a}. Invoking the bound on the bias and the second moment by Lemma \ref{lem:moments3}, we get:
\begin{align*}
R(T) \leq~& 2\eta \sum_{t=1}^T \sum_{b\in B} \pi_t(b) \cdot \sum_{o\in O_{\epsilon}} \frac{\Pr_t[o|b]}{\sum_{o'\in N_{\epsilon}^{in}(o)}\Pr_t[o']} + \frac{1}{\eta} \log(|B|) + 4\epsilon |O| T\\
=~& 2\eta \sum_{t=1}^T \sum_{o\in O_{\epsilon}} \sum_{b\in B} \pi_t(b) \cdot  \frac{\Pr_t[o|b]}{\sum_{o'\in N_{\epsilon}^{in}(o)}\Pr_t[o']} + \frac{1}{\eta} \log(|B|) + 4\epsilon |O| T\\
=~& 2\eta \sum_{t=1}^T \sum_{o\in O_{\epsilon}} \frac{\Pr[o]}{\sum_{o'\in N_{\epsilon}^{in}(o)}\Pr_t[o']} + \frac{1}{\eta} \log(|B|) + 4\epsilon |O| T
\end{align*}
We can now invoke Lemma 5 of \cite{ACDK15}, which states that:
\begin{lemma}[\cite{ACDK15}]  Let $G = (V,E)$ be a directed graph with $|V| = K$, in which each node $i\in V$ is assigned a positive weight $w_i$. Assume that $\sum_{i\in V}w_i\leq 1$, and that $w_i\geq \epsilon$ for all $i \in V$ for some constant $0 < \epsilon <1/2$. Then 
\begin{equation}
\sum_{i\in V}\frac{w_i}{\sum_{j\in N^{in}(i)} w_j} \leq 4\alpha \ln\frac{4K}{\alpha \epsilon}
\end{equation}
where neighborhoods include self-loops and $\alpha$ is the independence number of the graph.
\end{lemma}

Invoking the above lemma for the feedback graph $G_{\epsilon}$ (and noting that the independence number cannot increase by restricting to a sub-graph), we get:
\begin{equation}
\sum_{o\in O_{\epsilon}} \frac{\Pr[o]}{\sum_{o'\in N_{\epsilon}^{in}(o)}\Pr_t[o']}\leq 4\alpha \ln \frac{4|O|}{\alpha \epsilon}
\end{equation}
Thus, we get a bound on the regret of:
\begin{align*}
R(T) \leq~&8\eta \alpha \ln\left(\frac{4|O|}{\alpha \epsilon}\right) T+ \frac{1}{\eta} \log(|B|) + 4\epsilon |O| T
\end{align*}
Picking $\epsilon = \frac{1}{4|O| T}$, we get: 
\begin{align*}
R(T) \leq~&8\eta \alpha \ln\left(\frac{16|O|^2 T}{\alpha}\right) T+ \frac{1}{\eta} \log(|B|) + 1
\end{align*}
Picking $\eta = \sqrt{\frac{\log(|B|)}{8T\alpha \ln\left(\frac{16|O|^2 T}{\alpha}\right)}}$, we get the theorem. 
\end{proof}

\section{Omitted proof for the regret of the exponential weights update}\label{appendix:a}
\begin{lemma}
The exponential weights update with an estimate $\tilde{u}_t(\cdot)\leq 0$ such that for any $b\in B$ and $t$, $\left\vert \E\left[\tilde{u}_t(b)\right] - (u_t(b)-1)\right\vert\leq \kappa$, achieves expected regret on the form:
\begin{align*}
R(T) \leq~& \frac{\eta}{2} \sum_{t=1}^T \sum_{b\in B} \pi_t(b) \cdot \E\left[\tilde{u}_t(b)^2\right] + \frac{1}{\eta} \log(|B|) + 2\kappa T
\end{align*}
\end{lemma}
\begin{proof}
Following the standard analysis of the exponential weight updates algorithm \cite{AHK12} and the fact that $\forall x \leq 0$, $e^x \leq 1 + x + \frac{x^2}{2}$ as well as let $b^* = \argmax_{b\in B} \E\left[\sum_{t=1}^T u_t(b)\right]$, we have 
\begin{align*}
 \E\left[\sum_{t=1}^{T} \tilde{u}_t(b^*)\right] &\leq \sum_{t=1}^T \sum_{b\in B}\pi_t(b)\E\left[\tilde{u}_t(b)\right] + \frac{\eta}{2} \sum_{t=1}^T \sum_{b\in B} \pi_t(b) \cdot \E\left[\tilde{u}_t(b)^2\right] + \frac{1}{\eta} \log(|B|)\\
& \leq \sum_{t=1}^T \sum_{b\in B}\pi_t(b)(u_t(b) - 1 + \kappa) + \frac{\eta}{2} \sum_{t=1}^T \sum_{b\in B} \pi_t(b) \cdot \E\left[\tilde{u}_t(b)^2\right] + \frac{1}{\eta} \log(|B|)\\
& =\E\left[ \sum_{t=1}^T u_t(b_t)\right] + \frac{\eta}{2} \sum_{t=1}^T \sum_{b\in B} \pi_t(b) \cdot \E\left[\tilde{u}_t(b)^2\right] + \frac{1}{\eta} \log(|B|) + \kappa T - T
\end{align*}

which implies that 
\begin{align*}
R(T) &~=~ \E\left[\sum_{t=1}^T u_t(b^*)\right] - \E\left[ \sum_{t=1}^T u_t(b_t)\right]\leq  \E\left[\sum_{t=1}^{T} \tilde{u}_t(b^*)\right]- \E\left[ \sum_{t=1}^T u_t(b_t)\right] + \kappa T + T\\
&~\leq~ \frac{\eta}{2} \sum_{t=1}^T \sum_{b\in B} \pi_t(b) \cdot \E\left[\tilde{u}_t(b)^2\right] + \frac{1}{\eta} \log(|B|) + 2\kappa T
\end{align*}
\end{proof}

\paragraph{Remark.} Let the estimator $\tilde{u}_t(b)$ be unbiased for any $t$ and any $b\in B$, then the expected regret is
\begin{align*}
R(T) \leq~& \frac{\eta}{2} \sum_{t=1}^T \sum_{b\in B} \pi_t(b) \cdot \E\left[\tilde{u}_t(b)^2\right] + \frac{1}{\eta} \log(|B|)
\end{align*}

\end{appendix}

%% file: ec2018.bbl

\begin{thebibliography}{00}


\ifx \showCODEN    \undefined \def \showCODEN     #1{\unskip}     \fi
\ifx \showDOI      \undefined \def \showDOI       #1{{\tt DOI:}\penalty0{#1}\ }
  \fi
\ifx \showISBNx    \undefined \def \showISBNx     #1{\unskip}     \fi
\ifx \showISBNxiii \undefined \def \showISBNxiii  #1{\unskip}     \fi
\ifx \showISSN     \undefined \def \showISSN      #1{\unskip}     \fi
\ifx \showLCCN     \undefined \def \showLCCN      #1{\unskip}     \fi
\ifx \shownote     \undefined \def \shownote      #1{#1}          \fi
\ifx \showarticletitle \undefined \def \showarticletitle #1{#1}   \fi
\ifx \showURL      \undefined \def \showURL       #1{#1}          \fi
\providecommand\bibfield[2]{#2}
\providecommand\bibinfo[2]{#2}
\providecommand\natexlab[1]{#1}
\providecommand\showeprint[2][]{arXiv:#2}

\bibitem[\protect\citeauthoryear{Adlakha and Johari}{Adlakha and
  Johari}{2013}]%
        {AJ13}
\bibfield{author}{\bibinfo{person}{Sachin Adlakha} {and}
  \bibinfo{person}{Ramesh Johari}.} \bibinfo{year}{2013}\natexlab{}.
\newblock \showarticletitle{Mean field equilibrium in dynamic games with
  strategic complementarities}.
\newblock \bibinfo{journal}{{\em Operations Research\/}} \bibinfo{volume}{61},
  \bibinfo{number}{4} (\bibinfo{year}{2013}), \bibinfo{pages}{971--989}.
\newblock


\bibitem[\protect\citeauthoryear{Agarwal, Hsu, Kale, Langford, Li, and
  Schapire}{Agarwal et~al\mbox{.}}{2014}]%
        {Agarwal14}
\bibfield{author}{\bibinfo{person}{Alekh Agarwal}, \bibinfo{person}{Daniel
  Hsu}, \bibinfo{person}{Satyen Kale}, \bibinfo{person}{John Langford},
  \bibinfo{person}{Lihong Li}, {and} \bibinfo{person}{Robert Schapire}.}
  \bibinfo{year}{2014}\natexlab{}.
\newblock \showarticletitle{Taming the Monster: A Fast and Simple Algorithm for
  Contextual Bandits}. In \bibinfo{booktitle}{{\em Proceedings of the 31st
  International Conference on Machine Learning}} {\em
  (\bibinfo{series}{Proceedings of Machine Learning Research})},
  \bibfield{editor}{\bibinfo{person}{Eric~P. Xing} {and} \bibinfo{person}{Tony
  Jebara}} (Eds.), Vol.~\bibinfo{volume}{32}. \bibinfo{publisher}{PMLR},
  \bibinfo{address}{Bejing, China}, \bibinfo{pages}{1638--1646}.
\newblock


\bibitem[\protect\citeauthoryear{Alon, Cesa-Bianchi, Dekel, and Koren}{Alon
  et~al\mbox{.}}{2015}]%
        {ACDK15}
\bibfield{author}{\bibinfo{person}{Noga Alon}, \bibinfo{person}{Nicolo
  Cesa-Bianchi}, \bibinfo{person}{Ofer Dekel}, {and} \bibinfo{person}{Tomer
  Koren}.} \bibinfo{year}{2015}\natexlab{}.
\newblock \showarticletitle{Online learning with feedback graphs: Beyond
  bandits}. In \bibinfo{booktitle}{{\em Conference on Learning Theory}}.
  \bibinfo{pages}{23--35}.
\newblock


\bibitem[\protect\citeauthoryear{Alon, Cesa-bianchi, Gentile, and Mansour}{Alon
  et~al\mbox{.}}{2013}]%
        {ACGM13}
\bibfield{author}{\bibinfo{person}{Noga Alon}, \bibinfo{person}{Nicolo
  Cesa-bianchi}, \bibinfo{person}{Claudio Gentile}, {and}
  \bibinfo{person}{Yishay Mansour}.} \bibinfo{year}{2013}\natexlab{}.
\newblock \showarticletitle{From Bandits to Experts: A Tale of Domination and
  Independence}.
\newblock In \bibinfo{booktitle}{{\em Advances in Neural Information Processing
  Systems 26}}, \bibfield{editor}{\bibinfo{person}{C.j.c. Burges},
  \bibinfo{person}{L.~Bottou}, \bibinfo{person}{M.~Welling},
  \bibinfo{person}{Z.~Ghahramani}, {and} \bibinfo{person}{K.q. Weinberger}}
  (Eds.). \bibinfo{pages}{1610--1618}.
\newblock


\bibitem[\protect\citeauthoryear{Amin, Cummings, Dworkin, Kearns, and
  Roth}{Amin et~al\mbox{.}}{2015}]%
        {ACDKR15}
\bibfield{author}{\bibinfo{person}{Kareem Amin}, \bibinfo{person}{Rachel
  Cummings}, \bibinfo{person}{Lili Dworkin}, \bibinfo{person}{Michael Kearns},
  {and} \bibinfo{person}{Aaron Roth}.} \bibinfo{year}{2015}\natexlab{}.
\newblock \showarticletitle{Online Learning and Profit Maximization from
  Revealed Preferences.}. In \bibinfo{booktitle}{{\em AAAI}}.
  \bibinfo{pages}{770--776}.
\newblock


\bibitem[\protect\citeauthoryear{Amin, Rostamizadeh, and Syed}{Amin
  et~al\mbox{.}}{2014}]%
        {ARS14}
\bibfield{author}{\bibinfo{person}{Kareem Amin}, \bibinfo{person}{Afshin
  Rostamizadeh}, {and} \bibinfo{person}{Umar Syed}.}
  \bibinfo{year}{2014}\natexlab{}.
\newblock \showarticletitle{Repeated contextual auctions with strategic
  buyers}. In \bibinfo{booktitle}{{\em Advances in Neural Information
  Processing Systems}}. \bibinfo{pages}{622--630}.
\newblock


\bibitem[\protect\citeauthoryear{Arora, Hazan, and Kale}{Arora
  et~al\mbox{.}}{2012}]%
        {AHK12}
\bibfield{author}{\bibinfo{person}{Sanjeev Arora}, \bibinfo{person}{Elad
  Hazan}, {and} \bibinfo{person}{Satyen Kale}.}
  \bibinfo{year}{2012}\natexlab{}.
\newblock \showarticletitle{The Multiplicative Weights Update Method: a
  Meta-Algorithm and Applications.}
\newblock \bibinfo{journal}{{\em Theory of Computing\/}} \bibinfo{volume}{8},
  \bibinfo{number}{1} (\bibinfo{year}{2012}), \bibinfo{pages}{121--164}.
\newblock


\bibitem[\protect\citeauthoryear{Auer, Cesa-Bianchi, Freund, and Schapire}{Auer
  et~al\mbox{.}}{2002}]%
        {ACBFS02}
\bibfield{author}{\bibinfo{person}{Peter Auer}, \bibinfo{person}{Nicolo
  Cesa-Bianchi}, \bibinfo{person}{Yoav Freund}, {and} \bibinfo{person}{Robert~E
  Schapire}.} \bibinfo{year}{2002}\natexlab{}.
\newblock \showarticletitle{The nonstochastic multiarmed bandit problem}.
\newblock \bibinfo{journal}{{\em SIAM journal on computing\/}}
  \bibinfo{volume}{32}, \bibinfo{number}{1} (\bibinfo{year}{2002}),
  \bibinfo{pages}{48--77}.
\newblock


\bibitem[\protect\citeauthoryear{Balseiro and Gur}{Balseiro and Gur}{2017}]%
        {BG17}
\bibfield{author}{\bibinfo{person}{Santiago Balseiro} {and}
  \bibinfo{person}{Yonatan Gur}.} \bibinfo{year}{2017}\natexlab{}.
\newblock \showarticletitle{Learning in Repeated Auctions with Budgets: Regret
  Minimization and Equilibrium}.
\newblock  (\bibinfo{year}{2017}).
\newblock


\bibitem[\protect\citeauthoryear{Balseiro, Besbes, and Weintraub}{Balseiro
  et~al\mbox{.}}{2015}]%
        {BBW15}
\bibfield{author}{\bibinfo{person}{Santiago~R Balseiro}, \bibinfo{person}{Omar
  Besbes}, {and} \bibinfo{person}{Gabriel~Y Weintraub}.}
  \bibinfo{year}{2015}\natexlab{}.
\newblock \showarticletitle{Repeated auctions with budgets in ad exchanges:
  Approximations and design}.
\newblock \bibinfo{journal}{{\em Management Science\/}} \bibinfo{volume}{61},
  \bibinfo{number}{4} (\bibinfo{year}{2015}), \bibinfo{pages}{864--884}.
\newblock


\bibitem[\protect\citeauthoryear{Blum, Hajiaghayi, Ligett, and Roth}{Blum
  et~al\mbox{.}}{2008}]%
        {BHLR08}
\bibfield{author}{\bibinfo{person}{Avrim Blum}, \bibinfo{person}{MohammadTaghi
  Hajiaghayi}, \bibinfo{person}{Katrina Ligett}, {and} \bibinfo{person}{Aaron
  Roth}.} \bibinfo{year}{2008}\natexlab{}.
\newblock \showarticletitle{Regret minimization and the price of total
  anarchy}. In \bibinfo{booktitle}{{\em Proceedings of the fortieth annual ACM
  symposium on Theory of computing}}. ACM, \bibinfo{pages}{373--382}.
\newblock


\bibitem[\protect\citeauthoryear{Blum, Kumar, Rudra, and Wu}{Blum
  et~al\mbox{.}}{2004}]%
        {BKRW04}
\bibfield{author}{\bibinfo{person}{Avrim Blum}, \bibinfo{person}{Vijay Kumar},
  \bibinfo{person}{Atri Rudra}, {and} \bibinfo{person}{Felix Wu}.}
  \bibinfo{year}{2004}\natexlab{}.
\newblock \showarticletitle{Online learning in online auctions}.
\newblock \bibinfo{journal}{{\em Theoretical Computer Science\/}}
  \bibinfo{volume}{324}, \bibinfo{number}{2-3} (\bibinfo{year}{2004}),
  \bibinfo{pages}{137--146}.
\newblock


\bibitem[\protect\citeauthoryear{Blum, Mansour, and Morgenstern}{Blum
  et~al\mbox{.}}{2015}]%
        {BMM15}
\bibfield{author}{\bibinfo{person}{Avrim Blum}, \bibinfo{person}{Yishay
  Mansour}, {and} \bibinfo{person}{Jamie Morgenstern}.}
  \bibinfo{year}{2015}\natexlab{}.
\newblock \showarticletitle{Learning Valuation Distributions from Partial
  Observation.}. In \bibinfo{booktitle}{{\em AAAI}}. \bibinfo{pages}{798--804}.
\newblock


\bibitem[\protect\citeauthoryear{Bubeck, Cesa-Bianchi, et~al\mbox{.}}{Bubeck
  et~al\mbox{.}}{2012}]%
        {BCB2012}
\bibfield{author}{\bibinfo{person}{S{\'e}bastien Bubeck},
  \bibinfo{person}{Nicolo Cesa-Bianchi}, {and} \bibinfo{person}{others}.}
  \bibinfo{year}{2012}\natexlab{}.
\newblock \showarticletitle{Regret analysis of stochastic and nonstochastic
  multi-armed bandit problems}.
\newblock \bibinfo{journal}{{\em Foundations and Trends{\textregistered} in
  Machine Learning\/}} \bibinfo{volume}{5}, \bibinfo{number}{1}
  (\bibinfo{year}{2012}), \bibinfo{pages}{1--122}.
\newblock


\bibitem[\protect\citeauthoryear{Caragiannis, Kaklamanis, Kanellopoulos,
  Kyropoulou, Lucier, Leme, and Tardos}{Caragiannis et~al\mbox{.}}{2015}]%
        {CKKKLLT12}
\bibfield{author}{\bibinfo{person}{Ioannis Caragiannis},
  \bibinfo{person}{Christos Kaklamanis}, \bibinfo{person}{Panagiotis
  Kanellopoulos}, \bibinfo{person}{Maria Kyropoulou}, \bibinfo{person}{Brendan
  Lucier}, \bibinfo{person}{Renato~Paes Leme}, {and} \bibinfo{person}{{\'E}va
  Tardos}.} \bibinfo{year}{2015}\natexlab{}.
\newblock \showarticletitle{Bounding the inefficiency of outcomes in
  generalized second price auctions}.
\newblock \bibinfo{journal}{{\em Journal of Economic Theory\/}}
  \bibinfo{volume}{156} (\bibinfo{year}{2015}), \bibinfo{pages}{343--388}.
\newblock


\bibitem[\protect\citeauthoryear{Cesa-Bianchi, Gentile, and
  Mansour}{Cesa-Bianchi et~al\mbox{.}}{2015}]%
        {CBGM13}
\bibfield{author}{\bibinfo{person}{Nicolo Cesa-Bianchi},
  \bibinfo{person}{Claudio Gentile}, {and} \bibinfo{person}{Yishay Mansour}.}
  \bibinfo{year}{2015}\natexlab{}.
\newblock \showarticletitle{Regret minimization for reserve prices in
  second-price auctions}.
\newblock \bibinfo{journal}{{\em IEEE Transactions on Information Theory\/}}
  \bibinfo{volume}{61}, \bibinfo{number}{1} (\bibinfo{year}{2015}),
  \bibinfo{pages}{549--564}.
\newblock


\bibitem[\protect\citeauthoryear{Cesa{-}Bianchi and Lugosi}{Cesa{-}Bianchi and
  Lugosi}{2006}]%
        {BL06}
\bibfield{author}{\bibinfo{person}{Nicol{\`{o}} Cesa{-}Bianchi} {and}
  \bibinfo{person}{G{\'{a}}bor Lugosi}.} \bibinfo{year}{2006}\natexlab{}.
\newblock \bibinfo{booktitle}{{\em Prediction, learning, and games}}.
\newblock \bibinfo{publisher}{Cambridge University Press}.
\newblock
\showISBNx{978-0-521-84108-5}


\bibitem[\protect\citeauthoryear{Chawla, Hartline, and Nekipelov}{Chawla
  et~al\mbox{.}}{2014}]%
        {CHN14}
\bibfield{author}{\bibinfo{person}{Shuchi Chawla}, \bibinfo{person}{Jason~D.
  Hartline}, {and} \bibinfo{person}{Denis Nekipelov}.}
  \bibinfo{year}{2014}\natexlab{}.
\newblock \showarticletitle{Mechanism design for data science}. In
  \bibinfo{booktitle}{{\em {ACM} Conference on Economics and Computation, {EC}
  '14, Stanford , CA, USA, June 8-12, 2014}}. \bibinfo{pages}{711--712}.
\newblock


\bibitem[\protect\citeauthoryear{Cohen, Hazan, and Koren}{Cohen
  et~al\mbox{.}}{2016}]%
        {CHK16}
\bibfield{author}{\bibinfo{person}{Alon Cohen}, \bibinfo{person}{Tamir Hazan},
  {and} \bibinfo{person}{Tomer Koren}.} \bibinfo{year}{2016}\natexlab{}.
\newblock \showarticletitle{Online learning with feedback graphs without the
  graphs}. In \bibinfo{booktitle}{{\em International Conference on Machine
  Learning}}. \bibinfo{pages}{811--819}.
\newblock


\bibitem[\protect\citeauthoryear{Cole and Roughgarden}{Cole and
  Roughgarden}{2014}]%
        {CR15}
\bibfield{author}{\bibinfo{person}{Richard Cole} {and} \bibinfo{person}{Tim
  Roughgarden}.} \bibinfo{year}{2014}\natexlab{}.
\newblock \showarticletitle{The sample complexity of revenue maximization}. In
  \bibinfo{booktitle}{{\em Proceedings of the forty-sixth annual ACM symposium
  on Theory of computing}}. ACM, \bibinfo{pages}{243--252}.
\newblock


\bibitem[\protect\citeauthoryear{Dhangwatnotai, Roughgarden, and
  Yan}{Dhangwatnotai et~al\mbox{.}}{2015}]%
        {DRY15}
\bibfield{author}{\bibinfo{person}{Peerapong Dhangwatnotai},
  \bibinfo{person}{Tim Roughgarden}, {and} \bibinfo{person}{Qiqi Yan}.}
  \bibinfo{year}{2015}\natexlab{}.
\newblock \showarticletitle{Revenue maximization with a single sample}.
\newblock \bibinfo{journal}{{\em Games and Economic Behavior\/}}
  \bibinfo{volume}{91} (\bibinfo{year}{2015}), \bibinfo{pages}{318--333}.
\newblock


\bibitem[\protect\citeauthoryear{Dikkala and Tardos}{Dikkala and
  Tardos}{2013}]%
        {DT13}
\bibfield{author}{\bibinfo{person}{Nishanth Dikkala} {and}
  \bibinfo{person}{{\'{E}}va Tardos}.} \bibinfo{year}{2013}\natexlab{}.
\newblock \showarticletitle{Can Credit Increase Revenue?}. In
  \bibinfo{booktitle}{{\em Web and Internet Economics - 9th International
  Conference, {WINE} 2013, Cambridge, MA, USA, December 11-14, 2013,
  Proceedings}}. \bibinfo{pages}{121--133}.
\newblock


\bibitem[\protect\citeauthoryear{Feldman, Koren, Livni, Mansour, and
  Zohar}{Feldman et~al\mbox{.}}{2016}]%
        {Feldman2016}
\bibfield{author}{\bibinfo{person}{Michal Feldman}, \bibinfo{person}{Tomer
  Koren}, \bibinfo{person}{Roi Livni}, \bibinfo{person}{Yishay Mansour}, {and}
  \bibinfo{person}{Aviv Zohar}.} \bibinfo{year}{2016}\natexlab{}.
\newblock \showarticletitle{Online Pricing with Strategic and Patient Buyers}.
\newblock In \bibinfo{booktitle}{{\em Advances in Neural Information Processing
  Systems 29}}, \bibfield{editor}{\bibinfo{person}{D.~D. Lee},
  \bibinfo{person}{M.~Sugiyama}, \bibinfo{person}{U.~V. Luxburg},
  \bibinfo{person}{I.~Guyon}, {and} \bibinfo{person}{R.~Garnett}} (Eds.).
  \bibinfo{publisher}{Curran Associates, Inc.}, \bibinfo{pages}{3864--3872}.
\newblock


\bibitem[\protect\citeauthoryear{Google}{Google}{2018a}]%
        {bidsim1}
\bibfield{author}{\bibinfo{person}{Google}.} \bibinfo{year}{2018}\natexlab{a}.
\newblock \bibinfo{title}{{AdWords Bid Simulator}}.
\newblock
  \bibinfo{howpublished}{\url{https://support.google.com/adwords/answer/2470105?hl=en&ref_topic=3122864}}.
    (\bibinfo{year}{2018}).
\newblock
\newblock
\shownote{[Online; accessed 15-February-2018].}


\bibitem[\protect\citeauthoryear{Google}{Google}{2018b}]%
        {bidsim2}
\bibfield{author}{\bibinfo{person}{Google}.} \bibinfo{year}{2018}\natexlab{b}.
\newblock \bibinfo{title}{{Bid Landscapes}}.
\newblock
  \bibinfo{howpublished}{\url{https://developers.google.com/adwords/api/docs/guides/bid-landscapes}}.
    (\bibinfo{year}{2018}).
\newblock
\newblock
\shownote{[Online; accessed 15-February-2018].}


\bibitem[\protect\citeauthoryear{Google}{Google}{2018c}]%
        {bidsim3}
\bibfield{author}{\bibinfo{person}{Google}.} \bibinfo{year}{2018}\natexlab{c}.
\newblock \bibinfo{title}{{Bid Lanscapes}}.
\newblock
  \bibinfo{howpublished}{\url{https://developers.google.com/adwords/api/docs/reference/v201710/DataService.BidLandscape}}.
    (\bibinfo{year}{2018}).
\newblock
\newblock
\shownote{[Online; accessed 15-February-2018].}


\bibitem[\protect\citeauthoryear{Gyorgy, Linder, and Lugosi}{Gyorgy
  et~al\mbox{.}}{2012}]%
        {GLL12}
\bibfield{author}{\bibinfo{person}{Andr{\'a}s Gyorgy},
  \bibinfo{person}{Tam{\'a}s Linder}, {and} \bibinfo{person}{G{\'a}bor
  Lugosi}.} \bibinfo{year}{2012}\natexlab{}.
\newblock \showarticletitle{Efficient tracking of large classes of experts}.
\newblock \bibinfo{journal}{{\em IEEE Transactions on Information Theory\/}}
  \bibinfo{volume}{58}, \bibinfo{number}{11} (\bibinfo{year}{2012}),
  \bibinfo{pages}{6709--6725}.
\newblock


\bibitem[\protect\citeauthoryear{Iyer, Johari, and Sundararajan}{Iyer
  et~al\mbox{.}}{2011}]%
        {IJS11}
\bibfield{author}{\bibinfo{person}{Krishnamurthy Iyer}, \bibinfo{person}{Ramesh
  Johari}, {and} \bibinfo{person}{Mukund Sundararajan}.}
  \bibinfo{year}{2011}\natexlab{}.
\newblock \showarticletitle{Mean field equilibria of dynamic auctions with
  learning}.
\newblock \bibinfo{journal}{{\em ACM SIGecom Exchanges\/}}
  \bibinfo{volume}{10}, \bibinfo{number}{3} (\bibinfo{year}{2011}),
  \bibinfo{pages}{10--14}.
\newblock


\bibitem[\protect\citeauthoryear{Kanoria and Nazerzadeh}{Kanoria and
  Nazerzadeh}{2014}]%
        {KN14}
\bibfield{author}{\bibinfo{person}{Yash Kanoria} {and} \bibinfo{person}{Hamid
  Nazerzadeh}.} \bibinfo{year}{2014}\natexlab{}.
\newblock \showarticletitle{Dynamic Reserve Prices for Repeated Auctions:
  Learning from Bids - Working Paper}. In \bibinfo{booktitle}{{\em Web and
  Internet Economics - 10th International Conference, {WINE} 2014, Beijing,
  China, December 14-17, 2014. Proceedings}}. \bibinfo{pages}{232}.
\newblock


\bibitem[\protect\citeauthoryear{Kleinberg, Slivkins, and Upfal}{Kleinberg
  et~al\mbox{.}}{2008}]%
        {KSU08}
\bibfield{author}{\bibinfo{person}{Robert Kleinberg},
  \bibinfo{person}{Aleksandrs Slivkins}, {and} \bibinfo{person}{Eli Upfal}.}
  \bibinfo{year}{2008}\natexlab{}.
\newblock \showarticletitle{Multi-armed bandits in metric spaces}. In
  \bibinfo{booktitle}{{\em Proceedings of the fortieth annual ACM symposium on
  Theory of computing}}. ACM, \bibinfo{pages}{681--690}.
\newblock


\bibitem[\protect\citeauthoryear{Kleinberg}{Kleinberg}{2005}]%
        {K05}
\bibfield{author}{\bibinfo{person}{Robert~D Kleinberg}.}
  \bibinfo{year}{2005}\natexlab{}.
\newblock \showarticletitle{Nearly tight bounds for the continuum-armed bandit
  problem}. In \bibinfo{booktitle}{{\em Advances in Neural Information
  Processing Systems}}. \bibinfo{pages}{697--704}.
\newblock


\bibitem[\protect\citeauthoryear{Koren, Livni, and Mansour}{Koren
  et~al\mbox{.}}{2017}]%
        {Koren17a}
\bibfield{author}{\bibinfo{person}{Tomer Koren}, \bibinfo{person}{Roi Livni},
  {and} \bibinfo{person}{Yishay Mansour}.} \bibinfo{year}{2017}\natexlab{}.
\newblock \showarticletitle{Bandits with Movement Costs and Adaptive Pricing}.
  In \bibinfo{booktitle}{{\em Proceedings of the 2017 Conference on Learning
  Theory}} {\em (\bibinfo{series}{Proceedings of Machine Learning Research})},
  \bibfield{editor}{\bibinfo{person}{Satyen Kale} {and} \bibinfo{person}{Ohad
  Shamir}} (Eds.), Vol.~\bibinfo{volume}{65}. \bibinfo{publisher}{PMLR},
  \bibinfo{address}{Amsterdam, Netherlands}, \bibinfo{pages}{1242--1268}.
\newblock


\bibitem[\protect\citeauthoryear{Land}{Land}{2014}]%
        {MicrosoftImage}
\bibfield{author}{\bibinfo{person}{Search~Engine Land}.}
  \bibinfo{year}{2014}\natexlab{}.
\newblock \bibinfo{title}{Bing Ads Launches “Bid Landscape”, A Keyword
  Level Bid Simulator Tool}.
\newblock
  \bibinfo{howpublished}{\url{https://searchengineland.com/bing-ads-launches-bid-landscape-keyword-level-bid-simulator-tool-187219}}.
    (\bibinfo{year}{2014}).
\newblock
\newblock
\shownote{[Online; accessed 15-February-2018].}


\bibitem[\protect\citeauthoryear{Lykouris, Syrgkanis, and Tardos}{Lykouris
  et~al\mbox{.}}{2016}]%
        {LST16}
\bibfield{author}{\bibinfo{person}{Thodoris Lykouris}, \bibinfo{person}{Vasilis
  Syrgkanis}, {and} \bibinfo{person}{{\'E}va Tardos}.}
  \bibinfo{year}{2016}\natexlab{}.
\newblock \showarticletitle{Learning and efficiency in games with dynamic
  population}. In \bibinfo{booktitle}{{\em Proceedings of the Twenty-Seventh
  Annual ACM-SIAM Symposium on Discrete Algorithms}}. Society for Industrial
  and Applied Mathematics, \bibinfo{pages}{120--129}.
\newblock


\bibitem[\protect\citeauthoryear{Mannor and Shamir}{Mannor and Shamir}{2011}]%
        {MS11}
\bibfield{author}{\bibinfo{person}{Shie Mannor} {and} \bibinfo{person}{Ohad
  Shamir}.} \bibinfo{year}{2011}\natexlab{}.
\newblock \showarticletitle{From Bandits to Experts: On the Value of
  Side-Observations.}. In \bibinfo{booktitle}{{\em NIPS}},
  \bibfield{editor}{\bibinfo{person}{John Shawe-Taylor},
  \bibinfo{person}{Richard~S. Zemel}, \bibinfo{person}{Peter~L. Bartlett},
  \bibinfo{person}{Fernando C.~N. Pereira}, {and} \bibinfo{person}{Kilian~Q.
  Weinberger}} (Eds.). \bibinfo{pages}{684--692}.
\newblock


\bibitem[\protect\citeauthoryear{Medina and Mohri}{Medina and Mohri}{2014}]%
        {MM14}
\bibfield{author}{\bibinfo{person}{Andres~M Medina} {and}
  \bibinfo{person}{Mehryar Mohri}.} \bibinfo{year}{2014}\natexlab{}.
\newblock \showarticletitle{Learning theory and algorithms for revenue
  optimization in second price auctions with reserve}. In
  \bibinfo{booktitle}{{\em Proceedings of the 31st International Conference on
  Machine Learning (ICML-14)}}. \bibinfo{pages}{262--270}.
\newblock


\bibitem[\protect\citeauthoryear{Medina and Vassilvitskii}{Medina and
  Vassilvitskii}{2017}]%
        {MV17}
\bibfield{author}{\bibinfo{person}{Andr{\'{e}}s~Mu{\~{n}}oz Medina} {and}
  \bibinfo{person}{Sergei Vassilvitskii}.} \bibinfo{year}{2017}\natexlab{}.
\newblock \showarticletitle{Revenue Optimization with Approximate Bid
  Predictions}.
\newblock \bibinfo{journal}{{\em CoRR\/}}  \bibinfo{volume}{abs/1706.04732}
  (\bibinfo{year}{2017}).
\newblock
\showeprint[arxiv]{1706.04732}
\showURL{%
\url{http://arxiv.org/abs/1706.04732}}


\bibitem[\protect\citeauthoryear{Microsoft}{Microsoft}{2018}]%
        {bidsim4}
\bibfield{author}{\bibinfo{person}{Microsoft}.}
  \bibinfo{year}{2018}\natexlab{}.
\newblock \bibinfo{title}{{BingAds, Bid Landscapes}}.
\newblock
  \bibinfo{howpublished}{\url{https://advertise.bingads.microsoft.com/en-us/resources/training/bidding-and-traffic-estimation}}.
    (\bibinfo{year}{2018}).
\newblock
\newblock
\shownote{[Online; accessed 15-February-2018].}


\bibitem[\protect\citeauthoryear{Ostrovsky and Schwarz}{Ostrovsky and
  Schwarz}{2011}]%
        {OS11}
\bibfield{author}{\bibinfo{person}{Michael Ostrovsky} {and}
  \bibinfo{person}{Michael Schwarz}.} \bibinfo{year}{2011}\natexlab{}.
\newblock \showarticletitle{Reserve prices in internet advertising auctions: A
  field experiment}. In \bibinfo{booktitle}{{\em Proceedings of the 12th ACM
  conference on Electronic commerce}}. ACM, \bibinfo{pages}{59--60}.
\newblock


\bibitem[\protect\citeauthoryear{Roughgarden}{Roughgarden}{2009}]%
        {R09}
\bibfield{author}{\bibinfo{person}{Tim Roughgarden}.}
  \bibinfo{year}{2009}\natexlab{}.
\newblock \showarticletitle{Intrinsic robustness of the price of anarchy}. In
  \bibinfo{booktitle}{{\em Proceedings of the forty-first annual ACM symposium
  on Theory of computing}}. ACM, \bibinfo{pages}{513--522}.
\newblock


\bibitem[\protect\citeauthoryear{Standard}{Standard}{2014}]%
        {GoogleImage}
\bibfield{author}{\bibinfo{person}{Search~Marketing Standard}.}
  \bibinfo{year}{2014}\natexlab{}.
\newblock \bibinfo{title}{Google AdWords Improves The Bid Simulator Tool
  Feature}.
\newblock
  \bibinfo{howpublished}{\url{http://www.searchmarketingstandard.com/google-adwords-improves-the-bid-simulator-tool-feature}}.
    (\bibinfo{year}{2014}).
\newblock
\newblock
\shownote{[Online; accessed 15-February-2018].}


\bibitem[\protect\citeauthoryear{Weed, Perchet, and Rigollet}{Weed
  et~al\mbox{.}}{2016}]%
        {WRP16}
\bibfield{author}{\bibinfo{person}{Jonathan Weed}, \bibinfo{person}{Vianney
  Perchet}, {and} \bibinfo{person}{Philippe Rigollet}.}
  \bibinfo{year}{2016}\natexlab{}.
\newblock \showarticletitle{Online learning in repeated auctions}. In
  \bibinfo{booktitle}{{\em Conference on Learning Theory}}.
  \bibinfo{pages}{1562--1583}.
\newblock


\bibitem[\protect\citeauthoryear{Wordstream}{Wordstream}{2018}]%
        {PPCTool}
\bibfield{author}{\bibinfo{person}{Wordstream}.}
  \bibinfo{year}{2018}\natexlab{}.
\newblock \bibinfo{title}{Bid Management Tools}.
\newblock
  \bibinfo{howpublished}{\url{https://www.wordstream.com/bid-management-tools}}.
    (\bibinfo{year}{2018}).
\newblock
\newblock
\shownote{[Online; accessed 15-February-2018].}


\end{thebibliography}
